\documentclass[11pt,a4paper]{article}
\usepackage{fullpage}
\usepackage[utf8]{inputenc}
\usepackage{amsmath,amssymb,mathtools}
\usepackage{amsthm}
\usepackage{xcolor}
\usepackage{tikz}
\usepackage{graphicx,color}
\definecolor{myblue}{RGB}{0,2,127} 
\definecolor{mygreen}{RGB}{30,150,34} 
\definecolor{myred}{rgb}{220,0,0}
\definecolor{blue}{RGB}{68,118,170} 
\definecolor{green}{RGB}{34,136,51} 
\definecolor{cyan}{RGB}{102,204,238}
\definecolor{coralred}{RGB}{238,102,119}
\definecolor{yellow}{RGB}{204,187,68}
\definecolor{purple}{RGB}{170,51,119}
\definecolor{grey}{RGB}{187,187,187}
\colorlet{loopcolor}{coralred}
\colorlet{treecolor}{purple}
\colorlet{alterpathcolor}{purple}

\usepackage{enumitem}
\usepackage{thmtools}
\usepackage{subcaption}
\usepackage[sort]{cite}
\usepackage{hyperref}
\hypersetup{
    colorlinks=true,       
    linkcolor=myblue,        
    citecolor=myred,         
    urlcolor=black,
	linktocpage=true
}
\usepackage{cleveref}
\usepackage{algorithm}
\usepackage{algorithmicx}
\usepackage[noend]{algpseudocode}
\usetikzlibrary {arrows.meta}

\algrenewcommand\algorithmicrequire{\textbf{Input:}}
\algrenewcommand\algorithmicensure{\textbf{Output:}}

\newenvironment{varalgorithm}[1]
  {\algorithm\renewcommand{\thealgorithm}{#1}}
  {\endalgorithm}

\newcommand{\mynearrow}{\scalebox{0.5}{$\nearrow$}}
\newcommand{\mysearrow}{\scalebox{0.5}{$\searrow$}}

\makeatletter
\newcommand{\algmargin}{\the\ALG@thistlm}
\makeatother
\algnewcommand{\parState}[1]{\State%
    \parbox[t]{\dimexpr\linewidth-\algmargin}{\strut\hangindent=\algorithmicindent \hangafter=1 #1\strut}}

\theoremstyle{plain}
\newtheorem{theorem}{Theorem}
\newtheorem{lemma}[theorem]{Lemma}
\newtheorem{corollary}[theorem]{Corollary}

\newtheorem{claim}[theorem]{Claim}
\newtheorem{observation}[theorem]{Observation}

\theoremstyle{definition}
\newtheorem{definition}[theorem]{Definition}

\newenvironment{claimproof}[1]{\par\noindent\textit{Proof:}\space#1}{\hfill $\lhd$}

\def\DISP{\textsc{Shortest Two Disjoint Paths}}
\def\Q{\widetilde{Q}}
\def\R{\widetilde{R}}
\def\T{\mathcal{T}}
\def\F{\mathcal{F}}

\def\NP{\mathsf{NP}}
\def\Amend{\texttt{Amend}}

\author{Ildik\'o Schlotter$^{1,2}$}
\title{Shortest Two Disjoint Paths in Conservative Graphs}
\date{
\normalsize
    $^1$Centre for Economic and Regional Studies, Hungary \\%
    $^2$Budapest University of Technology and Economics, Hungary \\[2ex]%
    }
\begin{document}

\maketitle


\begin{abstract}
We consider the following problem that we call the \textsc{Shortest Two Disjoint Paths} problem: 
given an undirected graph~$G=(V,E)$ with edge weights $w:E\rightarrow \mathbb{R}$, two terminals $s$ and~$t$ in~$G$, 
find two internally vertex-disjoint paths between $s$ and $t$ with minimum total weight.
As shown recently by Schlotter and Seb\H{o} (2022), this problem becomes $\NP$-hard if edges can have negative weights, 
even if the weight function is conservative, i.e., there are no cycles in~$G$ with negative total weight. 
We propose a polynomial-time algorithm that solves the \textsc{Shortest Two Disjoint Paths} problem for conservative weights 
in the case when the negative-weight edges form a constant number of trees in~$G$. 
\end{abstract}

\section{Introduction}

Finding disjoint paths between given terminals is a fundamental problem in algorithmic graph theory and combinatorial optimization.
Besides its theoretical importance, it is also motivated by numerous applications in transportation, VLSI design, and network routing.  
In the \textsc{Disjoint Paths} problem, we are given $k$ terminal pairs $(s_i,t_i)$ for $i \in \{1,\dots,k\}$ in an undirected graph~$G$, 
and the task is to find pairwise vertex-disjoint paths $P_1,\dots, P_k$ 
so that $P_i$ connects $s_i$ with $t_i$ for each $i \in \{1,\dots,k\}$. 
This problem was shown to be $\NP$-hard by Karp~\cite{Karp75} when $k$ is part of the input, and remains $\NP$-hard even on planar graphs~\cite{Lynch1975}.
Robertson and Seymour~\cite{RS1995} proved that there exists an $f(k) n^3$ algorithm for \textsc{Disjoint Paths} with $k$ terminal pairs, 
where $n$ is the number of vertices in~$G$ and $f$ some computable function; this celebrated result is among the most important achievements of graph minor theory.
In the \textsc{Shortest Disjoint Paths} problem we additionally require that $P_1, \dots, P_k$ have minimum total length (in terms of the number of edges). 
For fixed $k$, the complexity of this problem is one of the  most important open questions in the area. 
Even the case for $k=2$ had been open for a long time, until Bj\"orklund and Husfeldt~\cite{BH2019} gave a randomized polynomial-time algorithm for it in 2019. 
For directed graphs the problem becomes much harder: the \textsc{Directed Disjoint Paths} problem is $\NP$-hard already for $k=2$.
The \textsc{Disjoint Paths} problem and its variants have also received considerable attention 
when restricted to planar graphs~\cite{DSS1992,VS2011,KS2010,AKKLST17,CMPP13,Schrijver1994,LMPSZ20}.

The variant of \textsc{Disjoint Paths} when $s_1=\dots =s_k=s$ and $t_1=\dots =t_k=t$ is considerably easier, since one can find $k$ pairwise 
(openly vertex- or edge-) disjoint paths between $s$ and $t$ using a max-flow computation. 
Applying standard techniques for computing a minimum-cost flow (see e.g.~\cite{schrijver-book}), 
one can even find $k$ pairwise disjoint paths between $s$ and $t$ with minimum total weight, 
given non-negative weights on the edges. 
Notice that if negative weights are allowed, then flow techniques break down for undirected graphs: 
in order to construct an appropriate flow network based on our undirected graph~$G$, 
the standard technique is to direct each edge of $G$ in both directions;
however, if edges can have negative weight, then this operation creates negative cycles consisting of two arcs, an obstacle for computing a minimum-cost flow.
Recently, Schlotter and Seb\H{o}~\cite{SS2023} have shown that this issue is a manifestation of a complexity barrier: 
finding two openly disjoint paths with minimum total weight between two vertices in an undirected edge-weighted graph is $\NP$-hard, even if weights are \emph{conservative} (i.e., no cycle has negative total weight) and each edge has weight in $\{-1,1\}$.\footnote{In fact, Schlotter and Seb\H{o} use an equivalent formulation of the problem where, instead of finding two openly disjoint paths between $s$ and $t$, the task  is to find two vertex-disjoint paths between $\{s_1,s_2\}$ and $\{t_1,t_2\}$ for four vertices $s_1,s_2,t_1,t_2 \in V$.} 
Note that negative edge weights occur in network problems due to various reasons:
for example, they might arise as a result of some reduction 
(e.g., deciding the feasibility of certain scheduling problems with deadlines translates into finding negative-weight cycles), 
or as a result of data that is represented on a logarithmic scale.
We remark that the \textsc{Single-Source Shortest Paths} problem is the subject of active research for the case when negative edges are allowed; 
see Bernstein et al.~\cite{BNWN22} for an overview of the area and their state-of-the-art algorithm running in near-linear time on directed graphs.

\medskip
\noindent
{\bf Our contribution.}
We consider the following problem
which concerns finding paths between two fixed terminals (as opposed to the classic \textsc{Shortest Disjoint Paths} problem): 
\begin{center}
\fbox{ 
\parbox{13.6cm}{
\begin{tabular}{l}\DISP{}:  \end{tabular} \\
\begin{tabular}{p{1cm}p{11.5cm}}
Input: & An undirected graph $G=(V,E)$, a weight function~$w\colon E \to \mathbb{R}$ that is conservative on~$G$, and two vertices $s$ and $t$ in~$G$. \\
Task: & Find two paths~$P_1$ and $P_2$ between $s$ and $t$ with $V(P_1) \cap V(P_2)=\{s,t\}$ that minimizes $w(P_1)+w(P_2)$.
\end{tabular}
}}
\end{center}
A \emph{solution} for an instance $(G,w,s,t)$ of \DISP{} is a pair of $(s,t)$-paths that are openly disjoint, i.e., do not share vertices other than their endpoints.

From the $\NP$-hardness proof for \DISP{} by Schlotter and Seb\H{o}~\cite{SS2023} it follows that the problem remains $\NP$-hard even if the set of negative-weight edges forms a perfect matching. 
Motivated by this intractability, we focus on the ``opposite'' case when 
the subgraph of~$G$ spanned by the set $E^-=\{e \in E:w(e)<0\}$ of negative-weight edges, denoted by $G[E^-]$, has only few connected components.\footnote{See 
Section~\ref{sec:prelim} for the precise definition of a subgraph spanned by an edge set.}
Note that since $w$ is conservative on~$G$, the graph $G[E^-]$ is acyclic.
Hence, if $c$ denotes the number of connected components in~$G[E^-]$, then 
$G[E^-]$ in fact consists of~$c$ trees.

We can think of our assumption that $c$ is constant as a compromise for allowing negative-weight edges but requiring that they be confined to a small part of the graph. For a motivation, consider a network where negative-weight edges arise as some rare anomaly. Such an anomaly may occur when, in a certain part of a computer network, some information can be collected while traversing the given edge. If such information concerns, e.g., the detection of (possibly) faulty nodes or edges in the network, then it is not unreasonable to assume that these faults are concentrated to a certain part of the network, due to underlying physical causes that are responsible for the fault.

Ideally, one would aim for an algorithm that is fixed-parameter tractable when parameterized by~$c$;
however, already the case $c=1$ turns out to be challenging.
We prove the following result, which can be thought of as a first step towards an FPT algorithm: 
\begin{theorem}
\label{thm:DISP-main}
For each constant~$c \in \mathbb{N}$, \DISP{} can be solved in polynomial time on instances where the set of negative edges spans $c$ trees in~$G$.
\end{theorem} 

Our algorithm first applies standard flow techniques to find minimum-weight solutions among those that have a simple structure
in the sense that there is no negative tree in~$G[E^-]$ used by both paths.
To deal with more complex solutions where there is at least one tree~$T$ in~$G[E^-]$ used by both paths, 
we use recursion to find two openly disjoint paths from~$s$ to~$T$, and from~$T$ to~$t$;
to deal with the subpaths of the solution that heavily use negative edges from~$T$, 
we apply an intricate dynamic programming method
that is based on significant insight into the structural properties of such solutions.

\medskip
\noindent
{\bf Organization.}
We give all necessary definitions in Section~\ref{sec:prelim}. 
In Section~\ref{sec:tree-init} we make initial observations about optimal solutions for an instance~$(G,w,s,t)$ of \DISP{},
and we also present a lemma of key importance that will enable us to create solutions by 
combining partial solutions that are easier to find (Lemma~\ref{lem:uncrossing}).
We present the algorithm proving our main result, Theorem~\ref{thm:DISP-main}, in Section~\ref{sec:negtree}. 
In Section~\ref{sec:mainAlg} we give a general description of our algorithm, 
and explain which types of solutions can be found using flow-based techniques. 
We proceed in Section~\ref{sec:properties} by establishing structural observation that we need to exploit in 
order to find those types of solutions where more advanced techniques are necessary. 
Section~\ref{sec:partsol} contains our dynamic programming method for finding partial solutions which, 
together with Lemma~\ref{lem:uncrossing}, form the heart of our algorithm.
We assemble the proof of Theorem~\ref{thm:DISP-main} in Section~\ref{sec:permdis} using our findings in the previous sections, 
and finally pose some questions for further research in Section~\ref{sec:conclusion}.


\section{Notation}
\label{sec:prelim}

For a positive integer $\ell$, we use $[\ell]\coloneqq\{1,2,\dots,\ell\}$.
For two subsets~$X$ and~$Y$ of some universe, let $X \Delta Y=(X \setminus Y) \cup (Y \setminus X)$ denote their symmetric difference.

Let a graph~$G$ be a pair $(V,E)$ where $V$ and $E$ are the set of vertices and edges, respectively.
For two vertices $u$ and $v$ in~$V$, an edge connecting $u$ and $v$ is denoted by $uv$ or $vu$.

For a set of $X$ of vertices (or edges), let $G-X$ denote the subgraph of~$G$ obtained by deleting the vertices (or edges, respectively) of~$X$; 
if $X=\{x\}$ then we may simply write $G-x$ instead of~$G - \{x\}$.
Given a set $U \subseteq V$ of vertices in~$G$, the subgraph of $G$ \emph{induced by}~$U$ is the graph $G-(V \setminus U)$.
Given a set $F \subseteq E$ of edges in~$G$, we denote by $V(F)$ the vertices incident to some edge of~$F$.
The subgraph of $G$ \emph{spanned by}~$F$ is the graph $(V(F),F)$; we denote this subgraph as $G[F]$.

A \emph{walk}~$W$ in~$G$ is a series $e_1, e_2, \dots, e_\ell$ of edges in~$G$ for which there exist 
vertices $v_0, v_1, \dots, v_\ell$ in~$G$ such that $e_i= v_{i-1} v_i$ for each $i \in [\ell]$;
note that both vertices and edges may appear repeatedly on a walk. 
We denote by $V(W)$ the set of vertices \emph{contained by} or \emph{appearing on}~$W$, that is, $V(W)=\{v_0, v_1, \dots, v_\ell\}$. 
The \emph{endpoints} of $W$ are $v_0$ and $v_\ell$, or in other words, it is a  \emph{$(v_0,v_\ell)$-walk}, while all vertices on~$W$ that are not endpoints are \emph{inner vertices}.
If $v_0=v_\ell$, then we say that $W$ is a \emph{closed walk}.

A \emph{path} is a walk on which no vertex appears more than once.
By a slight abuse of notation, we will usually treat a path as a \emph{set} $\{e_1,e_2,\dots, e_\ell\}$ of edges for which there exist distinct vertices $v_0, v_1, \dots, v_\ell$ in~$G$ such that $e_i= v_{i-1} v_i$ for each $i \in [\ell]$.
For any $i$ and $j$ with~$0 \leq i \leq j \leq \ell$ we will write $P[v_i,v_j]$ for the \emph{subpath} of $P$ between~$v_i$ and~$v_j$, consisting of edges $e_{i+1}, \dots, e_j$. 
Note that since we associate no direction with~$P$, we have $P[v_i,v_j]=P[v_j,v_i]$.
Given two vertices $s$ and~$t$, an $(s,t)$-path is a path whose endpoints are~$s$ and~$t$.
Similarly, for two subsets~$S$ and~$T$ of vertices, 
an $(S,T)$-path is a path with one endpoint in~$S$ and the other endpoint in~$T$. 

We say that two paths are \emph{vertex-} or \emph{edge-disjoint}, if they do not share a common vertex or edge, respectively. Two paths are \emph{openly disjoint}, if they share no common vertices apart from possibly their endpoints. 
Given vertices $s_1,s_2,t_1,$ and $t_2$, we say that two $(\{s_1,s_2\},\{t_1,t_2\})$-paths are \emph{permissively disjoint}, if a vertex $v$ can only appear on both paths if either $v=s_1=s_2$ or $v=t_1=t_2$. 
Two paths properly intersect, if they share at least one edge, but neither is the subpath of the other.

A \emph{cycle} in~$G$ is a set $\{e_1, e_2, \dots, e_\ell\}$ of distinct edges in~$G$ such that $e_1, e_2, \dots, e_{\ell-1}$ form a path in~$G - e_\ell$ whose endpoints are connected by $e_\ell$. 
A set $T \subseteq E$ of edges in~$G$ is \emph{connected}, if for every pair of edges $e$ and $e'$ in~$T$, there is a path contained in~$T$ containing both $e$ and $e'$.
If $T$ is connected and \emph{acyclic}, i.e., contains no cycle, then $T$ is a \emph{tree} in~$G$. 
Given two vertices $a$ and $b$ in a tree~$T$, we denote by $T[a,b]$ the unique path contained in~$T$ whose endpoints are $a$ and $b$. 
For an edge $uv \in T$ and a path~$P$ within~$T$ such that $uv \notin P$,
we say that $v$ \emph{is closer to~$P$} in~$T$ than~$u$, if $v \in V(T[u,p])$ for some vertex $p \in V(P)$.

Given a weight function~$w\colon E \to \mathbb{R}$ on the edge set of~$G$, 
we define the \emph{weight} of any edge set $F \subseteq E$ as $w(F)=\sum_{e \in F} w(e)$. 
We extend this notion for any pair $\mathcal{F}=(F_1,F_2)$ of edge sets by letting $w(\mathcal{F})=w(F_1)+w(F_2)$.
The restriction of~$w$ to an edge set $F \subseteq E$,
i.e., the function whose domain is $F$ and has value~$w(f)$ on each $f \in F$, 
is denoted by~$w_{|F}$.
We say that $w$ (or, to make the dependency on~$G$ explicit, the weighted graph~$(G,w)$) is \emph{conservative}, if no cycle in~$G$ has negative total weight. 

\section{Structural Observations}
\label{sec:tree-init}

Let $G=(V,E)$ be an undirected graph with a conservative weight function~$w:E \rightarrow \mathbb{R}$. 
Let $E^-=\{e \in E:w(e)<0\}$ denote the set of negative edges, and $\T$ the set of negative trees they form. More precisely, let~$\T$ be the set of connected components in the subgraph $G[E^-]$; the acyclicity of each $T \in \T$ follows from the conservativeness of~$w$.
For any subset~$\T'$ of~$\T$, we use the notation $E(\T')=\bigcup_{T \in \T'} E(T)$ and $V(\T')=\bigcup_{T \in \T'} V(T)$.

In Section~\ref{sec:init_weight} we gather a few useful properties of conservative weight functions. 
In Section~\ref{sec:init_paths} we collect observations on how an optimal solution can use different trees in~$\T$. 
We close the section with a lemma of key importance in Section~\ref{sec:init_uncrossing} that enables us to compose solutions by 
combining two path pairs without violating our requirement of disjointness.

\subsection{Implications of Conservative Weights}
\label{sec:init_weight}

The next two lemmas establish implications of the conservativeness of our weight function.
Lemma~\ref{lem:closed-walk} concerns closed walks, while Lemma~\ref{lem:T-min-pathlength} 
considers paths running between two vertices on some negative tree in~$\T$.
These lemmas will be useful in proofs where a given hypothetical solution is ``edited'' -- by removing certain subpaths from it and replacing them with paths within some negative tree -- in order to obtain a specific form without increasing its weight.

\begin{lemma}
\label{lem:closed-walk}
If $W$ is a closed walk that does not contain any edge with negative weight more than once, then $w(W) \geq 0$.
\end{lemma}

\begin{proof}
Let $Z$ be the set of edges that includes an edge $e$ of the graph if and only if $e$ appears an odd number of times on~$W$. Note that since any edge used at least twice in~$W$ has non-negative weight, we know $w(Z) \leq w(W)$. Moreover, since $W$ is a walk, each vertex $v \in V$ has an even degree in the graph $(V,Z)$. 
Therefore, $(V,Z)$ is the union of edge-disjoint cycles, and the conservativeness of~$w$ implies $w(Z) \geq 0$.
\end{proof}

\begin{lemma}
\label{lem:T-min-pathlength}
Let $x,y,x',y'$ be four distinct vertices on a tree~$T$ in~$\T$. 
\begin{description}
    \item[(1)] If $Q$ is an $(x,y)$-walk in~$G$ that does not use any edge with negative weight more than once, then $w(Q) \geq w(T[x,y])$. \\
    Moreover, if $T[x,y] \not \subseteq Q$, then the inequality is strict. 
    \item[(2)] If $Q$ is an $(x,y)$-path, $Q'$ is an $(x',y')$-path, and $Q$ and $Q'$ are vertex-disjoint,  then \\
    $w(Q)+w(Q')\geq w(T[x,y] \Delta T[x',y'])$. \\
    Moreover, if $T[x,y]$ and $T[x',y']$ properly intersect, then the inequality is strict.
\end{description}
\end{lemma}

\begin{proof}
We first start with describing a useful procedure. The input of this procedure is a set~$F \subseteq E$ that can be partitioned into a set $\F$ of mutually vertex-disjoint paths, each with both endpoints on~$T$, and its output is a new weight function~$w_F$ on $E$ fulfilling $w(F)=w_F(F)$. 

Initially, we set $w_F \equiv w$.
A \emph{$T$-leap} in~$F$ is a path in~$\F$ that has both endpoints on~$T$, has no inner vertices on~$T$, and contains no edge of~$T$. 
We consider each $T$-leap in~$F$ one-by-one. 
So let $L$ be a $T$-leap in~$F$ with endpoints $a$ and~$b$.
We define the \emph{shadow} of~$L$ as the edge set $F \cap T[a,b]$. 
Then for each edge $f$ in the shadow of $L$ for which $w_Q(f)=w(f)$, we set $w_F(f):=0$ and we decrease $w_F(L)$ by $|w(f)|$; we may do this by decreasing $w_F$ on the edges of~$L$ in any way as long as the total weight of~$L$ is decreased by $|w(f)|$. 
Since $w$ is conservative, $w(L) + w(T[a,b]) \geq 0$ by Lemma~\ref{lem:closed-walk}, and thus $w(L)\geq |w(T[a,b])| \geq \sum \{ |w(f)|:f \textrm{ is in the shadow of }L\}$.
Therefore, after performing this operation for each edge in the shadow of~$L$, $w_F(L) \geq 0$ remains true.
Applying these changes for each $T$-leap in~$F$, the resulting weight function~$w_F$ fulfills 
(i) $w_F(F) = w(F)$, 
(ii) $w_F(L) \geq 0$ for each $T$-leap~$L$ in~$F$, 
and that (iii) $w_F(f)=0$ for each edge~$f$ in the shadow of~$F$.

To show statement~(1) of the lemma, let $Q$ be an $(x,y)$-walk in~$G$. Observe that we can assume w.l.o.g.\ that $Q$ is an $(x,y)$-path: if $Q$ contains cycles, then we can repeatedly delete any cycle from~$Q$, possibly of length~2, so that in the end we obtain an $(x,y)$-path whose weight 
is at most the weight of~$Q$ by Lemma~\ref{lem:closed-walk},
since $w$ is conservative and no edge with negative weight is contained more than once in~$Q$.
Observe that any edge of~$Q \cap T$ that is not in the shadow of any $T$-leap on~$Q$ must lie on~$T[x,y]$: 
indeed, for any edge $uv \in Q \cap T$ that lies outside~$T[x,y]$, with $v$ being closer to $T[x,y]$ than~$u$ in~$T$, 
either $Q[u,x]$ or $Q[u,y]$ needs to use  a $T$-leap that contains~$uv$ in its shadow. 
Hence,  for the function~$w_Q$ obtained by the above procedure (for~$F = Q$) we get 
\begin{align*}
    w(Q) = w_Q(Q) & = \sum_{f \in Q \cap T[x,y]} w_Q(f) +
    \sum_{f \in Q \cap T \setminus T[x,y]} w_Q(f) +
    \sum_{f \in Q \setminus T} w_Q(f) \\[2pt]
    & \geq w(Q \cap T[x,y])+0+\sum_{\substack{\text{$L$: $L$ is a} \\ \text{$T$-leap on~$Q$}}} w_Q(L) \\ & \geq  w(T[x,y]).
\end{align*}
Observe also that equality is only possible if $w(Q \cap T[x,y])=w(T[x,y])$, which in turn happens only if $T[x,y] \subseteq Q$. 

To prove statement~(2), assume that $Q$ and $Q'$ are as in the statement.
Observe that any edge $f$ in~$(Q \cup Q') \cap T$ that is not contained in~$T[x,y] \Delta T[x',y']$ must be in the shadow of some $T$-leap in~$Q \cup Q'$: indeed,
 if $f$ is not in $T[x,y] \cup T[x',y']$, then this follows by the same arguments we used for statement~(1), and if $f \in T[x,y] \cap T[x',y']$, then either $f \in Q$ in which case it must be in the shadow of a $T$-leap on~$Q'$, or $f \in Q'$ in which case it is in the shadow of a $T$-leap on~$Q$, as $Q$ and~$Q'$ are vertex-disjoint. Considering the function~$w_{Q \cup Q'}$ obtained by the above procedure, we therefore know
\begin{align}
w(Q \cup Q') \, & = \,w_{Q \cup Q'}(Q \cup Q') \notag \\
    &= \sum_{\substack{f \in Q \cup Q' \text{ and} \\[2pt] f \in  T[x,y] \triangle T[x',y']}} w_Q(f) +
    \sum_{\substack{f \in Q \cup Q'\text{ and}  \notag \\[2pt] f \notin  T[x,y] \triangle T[x',y']}} w_Q(f) +
    \sum_{f \in Q \cup Q' \setminus T} w_Q(f) \notag \\[2pt]
&\geq w \Big( (Q \cup Q') \cap (T[x,y] \triangle T[x',y']) \Big) +0+\sum_{\substack{\text{$L$: $L$ is a } \\[1pt] \text{$T$-leap in~$Q \cup Q'$ }}} w_{Q \cup Q'}(L) \notag  \\
\, &\geq \,   w(T[x,y] \triangle T[x',y']). \label{eqn:Tminpath}
\end{align}To prove the last statement of the lemma, assume that $T[x,y]$ properly intersects $T[x',y']$. 
Since $Q$ and~$Q'$ are vertex-disjoint, 
this implies that there exists an edge $f^\star$ in $T[x,y] \cap  T[x',y']$ that is not contained in $Q \cup Q'$ 
(because neither~$Q$ nor~$Q'$ can entirely contain $T[x,y] \cap  T[x',y']$). 
Note that $Q$ must contain a leap~$L^\star$ for which the cycle induced by~$L$ contains~$f^\star$ (in fact, the same is true for~$Q'$).
Since $f^\star \notin Q_1 \cup Q_2$, we have that $f^\star$ is not in the shadow of~$L^\star$,  which implies that
$w_{Q \cup Q'}(L^\star) \geq  -w(f^\star) > 0$. Consequently,  inequality~(\ref{eqn:Tminpath}) is strict.
\end{proof}

\subsection{Solution Structure on Negative Trees}
\label{sec:init_paths}
We first observe a simple property of minimum-weight solutions.

\begin{definition}[\bf Locally cheapest path pairs]
Let $s_1,s_2,t_1,t_2$ be vertices in~$G$, and let $P_1$ and $P_2$ be two permissively disjoint $(\{s_1,s_2\},\{t_1,t_2\})$-paths. 
A path~$T[u,v]$ in some~$T \in \T$ is called a \emph{shortcut} for $P_1$ and $P_2$, if $u$ and $v$ both appear on the same path, either $P_1$ or~$P_2$, 
and there is no inner vertex or edge of~$T[u,v]$ contained in $P_1 \cup P_2$. 
We will call $P_1$ and~$P_2$ \emph{locally cheapest}, if there is no shortcut for them.
\end{definition}

The idea behind this concept is the following. 
Suppose that $P_1$ and $P_2$ are permissively disjoint $(\{s_1,s_2\},\{t_1,t_2\})$-paths, and 
$T[u,v]$ is a shortcut for $P_1$ and $P_2$.
Suppose that $u$ and~$v$ both lie on $P_i$ (for some $i \in [2]$), and let $P'_i$ be the path obtained by replacing $P_i[u,v]$ with $T[u,v]$; 
we refer to this operation as \emph{amending} the shortcut $T[u,v]$.
Then $P'_i$ is also permissively disjoint from~$P_{3-i}$ and, since Lemma~\ref{lem:closed-walk} implies $w(P_i[u,v]) \geq -w(T[u,v]) >0$,
has weight less than $w(P_i)$. 
Hence, we have the following observation.

\begin{observation}
\label{obs:locally-cheapest}
Let  $P_1$ and $P_2$ be two permissively disjoint $(\{s_1,s_2\},\{t_1,t_2\})$-paths
admitting a shortcut $T[z,z']$. 
Suppose that $z$ and $z'$ are on the path, say,~$P_1$.
Let $P'_1$ be the path obtained by amending~$T[z,z']$ on~$P_1$. Then 
$P'_1$ and $P_2$ are permissively disjoint $(\{s_1,s_2\},\{t_1,t_2\})$-paths and
    $w(P'_1) < w(P_1)$.
\end{observation}

\begin{corollary}
\label{cor:locally-cheapest}
Any minimum-weight solution for~$(G,w,s,t)$ is a pair of locally cheapest paths.
\end{corollary}


For two paths~$P_1$ and~$P_2$, we define $\Amend(P_1,P_2)$ as the result of amending all shortcuts on~$P_1$ and~$P_2$ in some arbitrarily fixed order. 
Observe that if shortcuts only appear on one of the paths among~$P_1$ and~$P_2$, then $\Amend(P_1,P_2)$ does not depend on the order in which we amend the shortcuts.
Shortcuts can be found simply by traversing all trees~$T \in \T$ and checking their intersection with $P_1$ and $P_2$. 
This way, $\Amend(P_1,P_2)$ can be computed in linear time.

\medskip

For convenience, for any $(s,t)$-path $P$ and vertices $u,v \in V(P)$ we say that $u$ \emph{precedes}~$v$  on~$P$, 
or equivalently, $v$ \emph{follows} $u$ on~$P$, if $u$ lies on $P[s,v]$. 
When defining a vertex as the ``first'' (or ``last'') vertex with some property on~$P$ or on a subpath~$P'$ of $P$ then, 
unless otherwise stated, we mean the vertex on~$P$ or on~$P'$ that is closest to~$s$ (or farthest from~$s$, respectively) that has the given property.

The following lemma shows that if two paths in a minimum-weight solution both use negative trees $T$ and~$T'$ for some $T,T' \in \mathcal{T}$
then, roughly speaking, they must traverse $T$ and~$T'$ in the same order; otherwise one would be able to 
replace the subpaths of the solution running between $T$ and~$T'$ by two paths, one within~$T$ and one within~$T'$, of smaller weight.

\begin{lemma}
\label{lem:twocomps-order}
Let $P_1$ and $P_2$ be two openly disjoint $(s,t)$-paths of minimum total weight, and let $T$ and $T'$ be distinct trees in~$\T$. 
Suppose that $v_1, v_2, v'_1$ and $v'_2$ are vertices such that $v_i \in V(T) \cap V(P_i)$ and $v'_i \in V(T') \cap V(P_i)$ for $i \in [2]$, 
with $v_1$ preceding $v'_1$ on~$P_1$. 
Then  $v_2$ precedes $v'_2$ on~$P_2$. 
\end{lemma}

\begin{proof}
Suppose for contradiction that the lemma does not hold, and assume that $v_1,v_2,v'_1$ and~$v'_2$ are 
as required by the lemma, but $v_2$ does not precede $v'_2$ on~$P_2$. 
We may also assume that $(v_1,v_2,v'_1,v'_2)$ forms a minimal counterexample in the sense that $|T[v_1,v_2]|+|T'[v'_1,v'_2]|$ is as small as possible.
Since $P_1$ and $P_2$ have minimum total weight, by Corollary~\ref{cor:locally-cheapest} we know that they are locally cheapest. 

We claim that $T[v_1,v_2]$ contains no vertex of~$P_1 \cup P_2$ as an inner vertex. Assuming the contrary,
let $u$ denote the first inner vertex on~$T[v_1,v_2]$, when traversed from~$v_1$ to~$v_2$, that lies on~$P_1 \cup P_2$. 
First, if $u$ lies on~$P_1[s,v'_1]$, then $(u,v_2,v'_1,v'_2)$ is a counterexample for the lemma with $|T[u,v_2]|+|T'[v'_1,v'_2]|<|T[v_1,v_2]|+|T'[v'_1,v'_2]|$, contradicting our assumption on~$(v_1,v_2,v'_1,v'_2)$.
Similarly, if $u$ lies on $P_2[v'_2,t]$, then $(v_1,u,v'_1,v'_2)$ is a counterexample for the lemma with $|T[v_1,u]|+|T'[v'_1,v'_2]|<|T[v_1,v_2]|+|T'[v'_1,v'_2]|$, again contradicting our assumption on~$(v_1,v_2,v'_1,v'_2)$.
Assume now $u$ lies on~$P_1[v'_1,t]$. Then $T[v_1,u]$ contains no inner vertex belonging to~$P_1 \cup P_2$ by the definition of~$u$, and is therefore a shortcut for~$P_1$ and~$P_2$, contradicting the fact that they are locally cheapest.
Hence, it must the case that $u$ lies on~$P_2[s,v'_2]$. Notice that $v'_2$ lies on~$P_2[u,v_2]$, and hence $T[u,v_2] \not \subseteq P_2$ by $v'_2 \notin V(T)$.
Since $P_1$ and~$P_2$ are locally cheapest, $T[u,v_2]$ does not contain a shortcut, and so $T[u,v_2]$ must contain a vertex~$z$ of $P_1$. 
On the one hand, if $z$ lies on~$P_1[s,v'_1]$, then 
$(z,v_2,v'_1,v'_2)$ is a counterexample for the lemma with $|T[z,v_2]|+|T'[v'_1,v'_2]|<|T[v_1,v_2]|+|T'[v'_1,v'_2]|$;
on the other hand, if $z$ lies on~$P_1[v'_1,t]$, then 
$(v'_1,v'_2,z,u)$ is a counterexample for the lemma with $|T'[v'_1,v'_2]|+|T[z,u]|<|T[v_1,v_2]|+|T'[v'_1,v'_2]|$.
In either case, we get a contradiction.

By symmetry, we also obtain that $T'[v'_1,v'_2]$ contains no vertex of~$P_1 \cup P_2$ as an inner vertex. Consider the closed walk 
\[W=P_1[v_1,v'_1] \cup T'[v'_1,v'_2] \cup P_2[v'_2,v_2] \cup T[v_2,v_1].\]
By the facts we just proved, $T[v_1,v_2] \cup T'[v'_1,v'_2]$ shares no vertices with~$P_1 \cup P_2$ except the endpoints, so by the disjointness of~$P_1$ and~$P_2$ we know that no edge of~$E^-$ is contained more than once in~$W$. Hence, Lemma~\ref{lem:closed-walk} yields $w(W) \geq 0$.

Define the paths $S_1$ and~$S_2$ as 
\begin{align*}
    S_1 &= P_1[s,v_1] \cup T[v_1,v_2] \cup P_2[v_2,t]; \\
    S_2 &= P_2[s,v'_2] \cup T'[v'_2,v'_1] \cup P_1[v'_1,t].
\end{align*}
Observe that $S_1$ and~$S_2$ are two openly disjoint $(s,t)$-paths with weight
\begin{align*}
w(S_1) +&  w(S_2) = w(P_1)+w( P_2) + w(T[v_1,v_2] \cup T'[v'_1,v'_2])  - w(P_1[v_1,v'_1] \cup P_2[v'_2,v_2])\\
&= w(P_1)+w( P_2) + w(T[v_1,v_2] \cup T'[v'_1,v'_2]) - w\left(W \setminus (T[v_1,v_2] \cup T'[v'_1,v'_2])\right) \\
&< w(P_1 )+ w( P_2) 
\end{align*}
where we used that $w(T[v_1,v_2] \cup T'[v'_1,v'_2]) <0$ and hence $w(W \setminus (T[v_1,v_2] \cup T'[v'_1,v'_2]))>0$.
This contradicts our assumption that $P_1$ and $P_2$ have minimum total weight.
\end{proof}

The following lemma is a consequence of Corollary~\ref{cor:locally-cheapest} and Lemma~\ref{lem:twocomps-order}, 
and considers a situation when one of the paths in an optimal solution visits a negative tree~$T \in \mathcal{T}$ at least twice, 
and visits some $T' \in \mathcal{T} \setminus \{T\}$ in between. 

\begin{lemma}
\label{lem:childcomp}
Let $P_1$ and $P_2$ be two openly disjoint $(s,t)$-paths of minimum total weight, and let $T$ and $T'$ be distinct trees in~$\T$. Suppose that $v_1, v'_2$, and~$v_3$ are vertices appearing in this order on~$P_1$ when traversed from~$s$ to~$t$, and suppose $v_1,v_3 \in V(T)$ while $v'_2 \in V(T')$.
Then 
\begin{itemize}
    \item $V(P_2) \cap V(T) \neq \emptyset$;
    \item $V(P_2) \cap V(T')=\emptyset$;
    \item no vertex of~$V(P_1) \cap V(T')$ precedes~$v_1$ or follows~$v_3$ on~$P_1$.
\end{itemize} 
\end{lemma}

\begin{proof}
First note that $P_1$ and~$P_2$ are locally cheapest by Corollary~\ref{cor:locally-cheapest}.
Since $v'_2 \notin V(T)$ but $v_2$ lies on~$P_1[v_1,v_3]$, we know $T[v_1,v_3] \not\subseteq P_1$, and thus $T[v_1,v_3]$ must contain a vertex of~$P_2$, say~$u$, as otherwise it admits a shortcut.

Second, assume for contradiction that there exists a vertex~$u' \in V(P_2) \cap V(T')$. 
On the one hand, if $u$ precedes~$u'$ on~$P_2$, then by Lemma~\ref{lem:twocomps-order} $v_3$ should precede~$v'_2$ on~$P_1$, contradicting our assumptions.
On the other hand, if $u'$ precedes~$u$ on~$P_2$, then by Lemma~\ref{lem:twocomps-order} $v'_2$ should precede~$v_1$ on~$P_1$, contradicting our assumptions.

Third, assume that $V(P_1) \cap V(T')$ contains a vertex~$u'$ that either precedes $v_1$ or follows~$v_3$ on~$P_1$. Note that either $v_1$ or $v_3$ lies between $v'_2$ and~$u'$ on~$P_1$, and thus $T[u',v'_2] \not \subseteq P_1$. However,  $T[u',v'_2]$ cannot contain any vertex of~$P_2$ (as we have just proved), and therefore $T[u',v'_2]$ contains a shortcut, contradicting the fact that $P_1$ and~$P_2$ are locally cheapest.
\end{proof}

We say that two paths $P_1$ and~$P_2$ are \emph{in contact} at~$T$, if there is a tree~$T \in \T$ and two distinct vertices~$v_1$ and~$v_2$ in~$T$
such that $v_1$ lies on~$P_1$, and $v_2$ lies on~$P_2$. 
Lemmas~\ref{lem:twocomps-order} and~\ref{lem:childcomp} imply the following fact 
that will enable us to use recursion in our algorithm to find solutions that consist of two paths in contact.
\begin{lemma}
\label{lem:partitioned-comps}
Let $P_1$ and $P_2$ be two openly disjoint $(s,t)$-paths of minimum total weight, and assume that they are in contact at some tree~$T \in \T$.
For $i \in [2]$, let $a_i$ and~$b_i$ denote the first and last vertices of~$P_i$ on~$T$ when traversed from~$s$ to~$t$.
Then we can partition~$\mathcal{T} \setminus \{T\}$ into~$(\T_s, \T_0,\T_t)$ such that for each $T' \in \T \setminus \{T\}$ and $i \in [2]$ it holds that
\begin{description}
    \item[(i)] if $P_i[s,a_i]$ contains a vertex of~$T'$, then $T' \in \T_s$
    \item[(ii)] if $P_i[b_i,t]$ contains a vertex of~$T'$, then $T' \in \T_t$
    \item[(iii)] if $P_i[a_i,b_i]$ contains a vertex of~$T'$, then $T' \in \T_0$.
\end{description}
\end{lemma}

\begin{proof}
Let $\T_s$ contain those trees in~$\T \setminus \{T\}$ that share a vertex with~$P_1[s,a_1]$ or~$P_2[s,a_2]$. 
Similarly, let $\T_t$ contain those trees in~$\T \setminus \{T\}$ that share a vertex with~$P_1[b_1,t]$ or~$P_2[b_2,t]$.
Then (i) and (ii) hold by definition.
We also let $\T_0=\T \setminus (\T_s \cup \T_t \cup \{T\})$.

Let us consider a tree $T' \in \T_s$; then there is a vertex~$v' \in V(T')$ contained on~$P_1[s,a_1]$ or on~$P_2[s,a_2]$.
We claim that $P_1[a_1,t] \cup P_1[a_2,t]$ shares no vertex with~$T'$.
For the sake of contradiction, let us assume that $u' \in V(T') \cap V(P_1[a_1,t] \cup P_1[a_2,t])$.
First, if $v'$ and~$u'$ are on the same path, $P_1$ or~$P_2$, then this contradicts Lemma~\ref{lem:childcomp}:
since both paths share vertices with~$T$, neither of them can visit $T$ in between~$v'$ and~$u'$.
Second, if $v'$ and~$u'$ are on different paths, then this contradicts Lemma~\ref{lem:twocomps-order}, 
since one of the paths visits~$T'$ before~$T$, and the other visits~$T$ before~$T'$.
This implies that $\T_s \cap \T_t=\emptyset$; in particular, $(\T_s, \T_0,\T_t)$ is a partitioning of~$\T$. 
Additionally, we get that no tree $T' \in \T_s$ may contain a vertex on~$P_1[a_1,b_1]$ or on~$P_2[a_2,b_2]$.
Observe that by symmetry, the same holds for any tree $T' \in T_t$; hence (iii) follows.
\end{proof}

Given a solution~$(P_1,P_2)$ whose paths are in contact at some tree $T \in \T$, 
a partition of~$\T \setminus \{T\}$ is \emph{$T$-valid} with respect to~$(P_1,P_2)$, 
if it satisfies the conditions of Lemma~\ref{lem:partitioned-comps}.

\subsection{Combining Path Pairs}
\label{sec:init_uncrossing}
The following lemma will be a crucial ingredient in our algorithm, as it enables us to combine ``partial solutions'' 
without violating the requirement of vertex-disjointness.

\begin{lemma}
\label{lem:uncrossing}
Let $p_1,p_2,q_1,q_2$ be vertices in~$G$, and let $T \in \T$ contain vertices~$v_1$ and~$v_2$ with $v_1 \neq v_2$.
Let $P_1$ and $P_2$ be two permissively disjoint $(\{p_1,p_2\},\{v_1,v_2\})$-paths in $G$, 
and let $Q_1$ and~$Q_2$ be two permissively disjoint $(\{v_1,v_2\},\{q_1,q_2\})$-paths in $G$ that are locally cheapest. 
Assume also that we can partition~$\T$ into two sets $\T_1$ and $\T_2$ with~$T \in \T_2$ such that 
\begin{description}
    \item[(i)]
    $V(T) \cap V(P_1 \cup P_2) =\{v_1,v_2\}$, and 
    \item[(ii)]
     $P_1 \cup P_2$ contains no edge of~$E(\T_2)$, and 
    $Q_1 \cup Q_2$  contains no edge of~$E(\T_1)$.
\end{description}
Then we can find in linear time two permissively disjoint $(\{p_1,p_2\},\{q_1,q_2\})$-paths $S_1$ and $S_2$ in~$G$ 
such that 
$w(S_1)+w(S_2) \leq w(P_1)+w(P_2) + w(Q_1)+w(Q_2)$.
\end{lemma}

\begin{proof}
W.l.o.g.\ we may assume that $P_i$ is a $(p_i,v_i)$-path 
and $Q_i$ is a $(v_i,q_i)$-path for both $i \in [2]$.
Let $y_i$ be the first vertex on~$P_i$ when traversed from~$p_i$ to~$v_i$ that is contained in~$V(Q_1 \cup Q_2)$.
We distinguish between three cases as follows.

{\bf Case A:} $y_1 \in V(Q_1)$ and $y_2 \in V(Q_2)$. 
In this case let $S_i=P_i[p_i,y_i] \cup Q_i[y_i,q_i]$
for $i \in [2]$. Observe that $S_1$ and $S_2$ are two permissively disjoint $(\{p_1,p_2\},\{q_1,q_2\})$-paths: 
this follows from the definition of~$y_1$ and $y_2$, and our assumptions on the disjointness of the paths~$P_1$ and $P_2$, as well as that of~$Q_1$ and~$Q_2$.
Note also that for any $i \in [2]$, the path~$S_i$ can be obtained from $P_i \cup Q_i$ by deleting the paths~$P_i[y_i,v_i]$ and $Q_i[v_i,y_i]$. Observe that $W_i=P_i[y_i,v_i] \cup Q_i[v_i,y_i]$ is a closed walk. 
Due to~(ii), no edge of~$E^-$ may appear both on~$P_i$ and on~$Q_i$. Hence, 
the walk~$W_i$ does not contain any edge of~$E^-$ more than once, and so by Lemma~\ref{lem:closed-walk} we get $w(W_i) \geq 0$.  This implies
$w(S_i) = w(P_i) +w(Q_i) -w(W_i) \leq w(P_i) + w(Q_i)$, proving the lemma for Case~A.

{\bf Case B:} $y_1 \in V(Q_2)$ and $y_2 \in V(Q_1)$. 
In this case let $S_i=P_i[s_i,y_i] \cup Q_{3-i}[y_i,q_{3-i}]$
for $i \in [2]$. One can observe as in the previous case that $S_1$ and $S_2$ are two permissively disjoint $(\{p_1,p_2\},\{q_1,q_2\})$-paths.
Note also that for any $i \in [2]$, the path~$S_i$ can be obtained from $P_i \cup Q_{3-i}$ by deleting the paths~$P_i[y_i,v_i]$ and $Q_{3-i}[v_{3-i},y_i]$. Hence we get
\[ S_1 \cup S_2= P_1 \cup P_2 \cup Q_1 \cup Q_2 \setminus
(P_1[y_1,v_1] \cup Q_1[v_1,y_2] \cup P_2[y_2,v_2] \cup Q_2[v_2,y_1]).
\]
Observe that $W=P_1[y_1,v_1] \cup Q_1[v_1,y_2] \cup P_2[y_2,v_2] \cup Q_2[v_2,y_1]$ is a closed walk.
Due to the disjointness of~$P_1$ and~$P_2$, we know
that $P_1$ shares no edge with~$P_2$; 
similarly, $Q_1$ shares no edge with~$Q_2$. 
Moreover, $P_1 \cup P_2$ shares no edge of~$E^-$ with $Q_1 \cup Q_2$ due to~(ii). Therefore, no edge of~$E^-$ may appear more than once on~$W$.
Thus by Lemma~\ref{lem:closed-walk} we get $w(W) \geq 0$.
 This implies
\[w(S_1)+w(S_2) = w(P_1)+w(P_2) + w(Q_1)+w(Q_2) -w(W) 
\leq w(P_1)+w(P_2) + w(Q_1)+w(Q_2),\] proving the lemma for Case~B.

{\bf Case C:} If $y_1$ and $y_2$ both lie on the same path from~$Q_1$ and~$Q_2$. By symmetry, we may assume that they both lie on~$Q_1$. Suppose that $y_1$ comes before~$y_2$ when $Q_1$ is traversed starting from~$v_1$;
the case when $y_2$ precedes $y_1$ can be dealt with in an analogous manner. 

\begin{claim}
\label{clm:connection-in-T}
We can find a path $T[u,u']$ in linear time such that~$u$ lies on~$Q_1 \setminus Q_1[y_2,q_1]$, $u'$ lies on~$Q_2$, and no other vertex of~$T[u,u']$ appears on~$Q_1 \cup Q_2$.  
\end{claim}

\begin{claimproof}
Recall that $y_1 \in V(P_1)$ implies that either $y_1=v_1$ or $y_1 \notin V(T)$. 

First assume $y_1=v_1$, and consider the path $T[v_1,v_2]$. Since $Q_1$ and~$Q_2$ may only share $q_1$ and~$q_2$ (in case they coincide), but $y_2 \notin V(T)$, it must be the case that 
when travelling through $T[v_1,v_2]$ from $v_1$, we exit $Q_1$ at some vertex~$u$, meaning that the edge following~$u$ on this path does not belong to~$Q_1$. Proceeding within~$T[v_1,v_2]$ towards $v_2$, at some point we enter a vertex of $Q_2$, let $u'$ denote this vertex. 
Observe that $T[u,u']$ contains no vertex of~$Q_2$ as an inner vertex by definition. Moreover, it also cannot contain a vertex of~$Q_1$ as an inner vertex. Indeed, supposing that some vertex $z \in V(Q_1)$ is an inner vertex of~$T[u,u']$, the path $T[u,z]$ would be a shortcut, contradict our assumption that $Q_1$ and~$Q_2$ are locally cheapest. 
Note also that by $y_2 \notin V(T)$ we also know that $u \notin Q_1[y_2,q_1]$.
Hence, the claim holds.
 
Assume now $y_1 \notin V(T)$. Then we start by walking on~$Q_1$ from~$y_1$ towards $v_1$ until we reach a vertex of~$T$ (note that, by our assumption, $y_2$ is not on this path, and note also $v_1 \in V(T)$), and from that point on, we travel towards $v_2 \in V(T)$ within~$T$. Following the procedure described in the previous case, we can define $u$ and $u'$ as the vertices where we exit $Q_1$ and enter $Q_2$, respectively, and we can argue $T[u,u']$ is disjoint from~$Q_1 \cup Q_2$ apart from its endpoints. Again, notice that $y_2$ cannot lie on the walk from~$y_1$ towards $u$, and so $u$ lies on $Q_1 \setminus Q_1[y_2,q_1]$ and the claim holds.
\end{claimproof}

\smallskip
Using Claim~\ref{clm:connection-in-T}, we can now define $S_1$ and $S_2$. Let $T[u,u']$ be the path guaranteed by the claim. We distinguish between two cases, depending on the place of~$u$ on~$Q_1$; see Figure~\ref{fig:uncrossing} for an illustration.

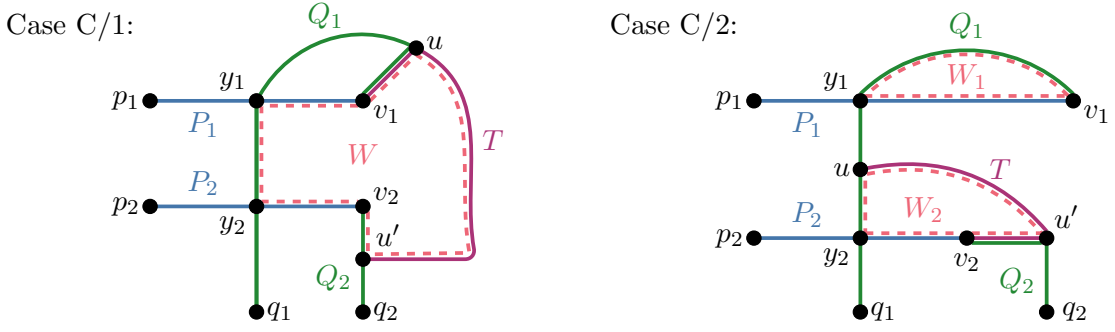
\begin{figure}
  \centering
  \begin{subfigure}[b]{0.5\textwidth}
    \centering  
    \begin{tikzpicture}[xscale=0.7, yscale=0.7]

  \node (C1) at (-3.5,3.4) {Case C/1:};
  \node[draw, circle, fill=black, inner sep=1.7pt] (y2) at (0, 0) {};
  \node[draw, circle, fill=black, inner sep=1.7pt] (y1) at (0, 2) {};  
  \node[draw, circle, fill=black, inner sep=1.7pt] (p2) at (-2, 0) {};
  \node[draw, circle, fill=black, inner sep=1.7pt] (p1) at (-2, 2) {};
  \node[draw, circle, fill=black, inner sep=1.7pt] (v2) at (2, 0) {};
  \node[draw, circle, fill=black, inner sep=1.7pt] (v1) at (2, 2) {};
  \node[draw, circle, fill=black, inner sep=1.7pt] (u) at (3, 3) {};
  \node[draw, circle, fill=black, inner sep=1.7pt] (q1) at (0, -2) {};
  \node[draw, circle, fill=black, inner sep=1.7pt] (q2) at (2, -2) {};
  \node[draw, circle, fill=black, inner sep=1.7pt] (u') at (2, -1) {};
  \coordinate (h) at (4, -0.88) {};
  \coordinate (h') at (4.12, -1) {};

  \node[left] at (p1) {$p_1$};
  \node[left] at (p2) {$p_2$};
  \node[above left] at (y1) {$y_1$};
  \node[below left] at (y2) {$y_2$};
  \node[right] at (q2) {$q_2$};  
  \node[right] at (q1) {$q_1$};    
  \node[right, yshift=-5pt] at (v1) {$v_1$};    
  \node[right, yshift=4pt] at (v2) {$v_2$};  
  \node[above right, yshift=1pt, xshift=1pt] at (u') {$u'$};    
  \node[right, yshift=2pt] at (u) {$u$};      

  \draw[line width=1.4pt, blue] (p1) to node[midway,below] {$P_1$} (y1);
  \draw[line width=1.4pt, blue] (y1) to (v1);  
  \draw[line width=1.4pt, blue] (p2) to node[midway,above] {$P_2$} (y2);  
  \draw[line width=1.4pt, blue] (y2) to (v2);
  \draw[line width=1.4pt, green] (v1.north) to (u.west);       
  \draw[line width=1.4pt, green] (q1) to (y2);    
  \draw[line width=1.4pt, green] (y2) to (y1);    
  \draw[line width=1.4pt, green] (q1) to (y2);    
  \draw[line width=1.4pt, green] (y2) to (y1);    
  \draw[line width=1.4pt, green] (y1) to[in=155, out=60] node[midway,above] {$Q_1$} (u);    
  \draw[line width=1.4pt, green] (v2) to (u');    
  \draw[line width=1.4pt, green] (u') to node[pos=0.3,left] {$Q_2$} (q2);      

  \draw[line width=1.4pt, treecolor, rounded corners] (v1) to (u) to[out=-30, in=100] node[midway,right] {$T$} (h') to (u');

  \draw[line width=1.4pt, loopcolor, dashed] (v1.east) to (u.south) to[out=-30, in=100] (h) to (u'.north east);    
  \draw[line width=1.4pt, loopcolor, dashed] (v1.south west) to (y1.south east) to (y2.north  east) to (v2.north west);    
  \draw[line width=1.4pt, loopcolor, dashed] (v2.south east) to (u'.north east);    
  
  \node[draw, circle, fill=black, inner sep=1.9pt] (y2) at (0, 0) {};
  \node[draw, circle, fill=black, inner sep=1.9pt] (y1) at (0, 2) {};  
  \node[draw, circle, fill=black, inner sep=1.9pt] (p2) at (-2, 0) {};
  \node[draw, circle, fill=black, inner sep=1.9pt] (p1) at (-2, 2) {};
  \node[draw, circle, fill=black, inner sep=1.9pt] (v2) at (2, 0) {};
  \node[draw, circle, fill=black, inner sep=1.9pt] (v1) at (2, 2) {};
  \node[draw, circle, fill=black, inner sep=1.9pt] (u) at (3, 3) {};
  \node[draw, circle, fill=black, inner sep=1.9pt] (q1) at (0, -2) {};
  \node[draw, circle, fill=black, inner sep=1.9pt] (q2) at (2, -2) {};
  \node[draw, circle, fill=black, inner sep=1.9pt] (u') at (2, -1) {};
  \node (W) at (2,1) [loopcolor] {$W$};

\end{tikzpicture}
  \end{subfigure} 
  \hspace{-5pt}
  \begin{subfigure}[b]{0.49\textwidth}
    \centering  
    \begin{tikzpicture}[xscale=0.7, yscale=0.7]

  \node (C2) at (-3.5,3.4) {Case C/2:};
  \node[draw, circle, fill=black, inner sep=1.7pt] (y2) at (0, -0.6) {};
  \node[draw, circle, fill=black, inner sep=1.7pt] (y1) at (0, 2) {};  
  \node[draw, circle, fill=black, inner sep=1.7pt] (p2) at (-2, -0.6) {};
  \node[draw, circle, fill=black, inner sep=1.7pt] (u)  at (0, 0.7) {};
  \node[draw, circle, fill=black, inner sep=1.7pt] (p1) at (-2, 2) {};
  \node[draw, circle, fill=black, inner sep=1.7pt] (v2) at (2, -0.6) {};
  \node[draw, circle, fill=black, inner sep=1.7pt] (q2) at (3.5, -2) {};
  \node[draw, circle, fill=black, inner sep=1.7pt] (u') at (3.5, -0.6) {};
  \node[draw, circle, fill=black, inner sep=1.7pt] (v1) at (4, 2) {};
  \node[draw, circle, fill=black, inner sep=1.7pt] (q1) at (0, -2) {};

  \node[left] at (p1) {$p_1$};
  \node[left] at (p2) {$p_2$};
  \node[above left] at (y1) {$y_1$};
  \node[below left] at (y2) {$y_2$};
  \node[right, xshift=2pt] at (q2) {$q_2$};  
  \node[right] at (q1) {$q_1$};    
  \node[right, yshift=-5pt] at (v1) {$v_1$};    
  \node[below, yshift=-2pt] at (v2) {$v_2$};  
  \node[above right, yshift=0pt, xshift=-1pt] at (u') {$u'$};    
  \node[left] at (u) {$u$};      

  \draw[line width=1.4pt, blue] (p1) to node[midway,below] {$P_1$} (y1);
  \draw[line width=1.4pt, blue] (y1) to (v1);  
  \draw[line width=1.4pt, blue] (p2) to node[midway,above] {$P_2$} (y2);  
  \draw[line width=1.4pt, blue] (y2) to (v2);
  \draw[line width=1.4pt, green] (v1.north) to [bend right=45] node[midway,above] {$Q_1$} (y1.north);
  \draw[line width=1.4pt, green] (y1) to (u) to (y2) to (q1);    
  \draw[line width=1.4pt, green] (v2.south east) to (u'.south west);
  \draw[line width=1.4pt, green]  (u') to node[pos=0.6,left] {$Q_2$} (q2);      

  \draw[line width=1.4pt, treecolor] (v2) to (u');
  \draw[line width=1.4pt, treecolor] (u'.north) to [bend right=30] node [pos=0.3,above] {$T$} (u);

  \draw[line width=1.4pt, loopcolor, dashed] (v1.north west) to [bend right=45] (y1.north east) to node [midway, above, yshift=-1.5pt] {$W_1$} (v1.north west);    
  \draw[line width=1.4pt, loopcolor, dashed] (y2.north east) to node [pos=0.6, above] {$W_2$} (v2.north west);
  \draw[line width=1.4pt, loopcolor, dashed] (v2.north east) to (u'.north west) to [bend right=30] (u.south east) to (y2.north east);

  \node[draw, circle, fill=black, inner sep=1.9pt] (y2) at (0, -0.6) {};
  \node[draw, circle, fill=black, inner sep=1.9pt] (y1) at (0, 2) {};  
  \node[draw, circle, fill=black, inner sep=1.9pt] (p2) at (-2, -0.6) {};
  \node[draw, circle, fill=black, inner sep=1.9pt] (u)  at (0, 0.7) {};
  \node[draw, circle, fill=black, inner sep=1.9pt] (p1) at (-2, 2) {};
  \node[draw, circle, fill=black, inner sep=1.9pt] (v2) at (2, -0.6) {};
  \node[draw, circle, fill=black, inner sep=1.9pt] (q2) at (3.5, -2) {};
  \node[draw, circle, fill=black, inner sep=1.9pt] (u') at (3.5, -0.6) {};
  \node[draw, circle, fill=black, inner sep=1.9pt] (v1) at (4, 2) {};
  \node[draw, circle, fill=black, inner sep=1.9pt] (q1) at (0, -2) {};

\end{tikzpicture}  
  \end{subfigure}
  \caption{Illustration for Case C in the proof of Lemma~\ref{lem:uncrossing}. 
  Paths $P_1$ and~$P_2$ are shown in \textcolor{blue}{\textbf{blue}}, paths $Q_1$ and~$Q_2$ in \textcolor{green}{\textbf{green}}, and edges of~$T$ in \textcolor{purple}{\textbf{purple}}.  Closed walks~$W$ (in Case C/1), $W_1$ and~$W_2$ (in Case C/2) are depicted with \textcolor{coralred}{\textbf{coral red}}, dashed lines. 
  }
  \label{fig:uncrossing}
\end{figure}

{\bf Case C/1:}
Assume that $u$ lies on $Q_1[v_1,y_1]$. Define
\begin{align}
\label{eq:uncrossing-defS1}
S_1 &=P_1[p_1,y_1] \cup Q_1[y_1,u] \cup T[u,u'] \cup Q_2[u',q_2], \\
\label{eq:uncrossing-defS2}
S_2 &=P_2[p_2,y_2] \cup Q_1[y_2,q_1].
\end{align} 
Observe that $S_1$ and $S_2$ are two permissively disjoint $(\{p_1,p_2\},\{q_1,q_2\})$-paths: 
this follows from the definition of~$y_1$ and $y_2$, our assumptions on the disjointness of the paths~$P_1$ and $P_2$, as well as that of~$Q_1$ and~$Q_2$, from the properties of~$T[u,u']$, and from our assumption~(i).

Consider the closed walk 
\[W:=P_1[y_1,v_1] \cup Q_1[v_1,u] \cup T[u,u'] \cup Q_2[u',v_2] \cup P_2[v_2,y_2] \cup
Q_1[y_2,y_1]. \]
Note that $W$ does not contain any edge of~$E^-$ more than once due to our assumption~(ii), the properties of~$T[u,u']$, and the disjointness of~$P_1$ and $P_2$, as well as that of $Q_1$ and~$Q_2$.
Hence we get $w(W)\geq 0$ by Lemma~\ref{lem:closed-walk}. 
Observe that Equations~\ref{eq:uncrossing-defS1} and~\ref{eq:uncrossing-defS2} yield
\begin{align*}
S_1 \cup S_2 &=  (P_1 \cup P_2 \cup Q_1 \cup Q_2) \setminus  (W \setminus T[u,u']) \cup T[u,u'].
\end{align*}
By $w(W) \geq 0$ and $w(T[u,u'])<0$ this implies 
$w(S_1)+w(S_2) <(P_1)+w(P_2) + w(Q_1)+w(Q_2)$, proving the lemma for Case C/1.

{\bf Case C/2:}
Assume now that $u$ lies on $Q_1[y_1,y_2]$.\footnote{The reader might observe that the proof of Claim~\ref{clm:connection-in-T} implies that in this case only $v_1=y_1$ is possible, but we do not use this fact in the proof of Lemma~\ref{lem:uncrossing}.}
Define $S_1$ and $S_2$ exactly as in Equations~\ref{eq:uncrossing-defS1} and~\ref{eq:uncrossing-defS2}; again, $S_1$ and $S_2$ are two permissively disjoint $(\{p_1,p_2\},\{q_1,q_2\})$-paths.
Consider the closed walks 
\begin{align*}
W_1 &=P_1[y_1,v_1] \cup Q_1[v_1,y_1], \\
W_2 &= Q_1[u,y_2] \cup P_2[y_2,v_2] \cup Q_2[v_2,u'] \cup T[u',u],
\end{align*}
and observe that, owing to the same reasons as in the previous case, neither~$W_1$ nor~$W_2$ contains an edge of~$E^-$ more than once.
Hence, we get $w(W_1 \cup W_2)\geq 0$ by Lemma~\ref{lem:closed-walk}.
Observe that Equations~\ref{eq:uncrossing-defS1} and~\ref{eq:uncrossing-defS2} yield
\begin{align*}
S_1 \cup S_2 &=  (P_1 \cup P_2 \cup Q_1 \cup Q_2) \setminus  (W_1 \cup W_2 \setminus T[u,u']) \cup T[u,u'].
\end{align*}
By $w(W_1 \cup W_2) \geq 0$ and $w(T[u,u'])<0$ this implies 
$w(S_1)+w(S_2) < w(P_1)+w(P_2) + w(Q_1)+w(Q_2)$, proving the lemma for Case C/2.
\end{proof}

\section{Polynomial-Time Algorithm for Constant \texorpdfstring{$|\mathcal{T}|$}{T}}
\label{sec:negtree}
This section contains the algorithm proving our main result, Theorem~\ref{thm:DISP-main}.
Let $(G,w,s,t)$ be our instance of \DISP{} with input graph~$G=(V,E)$, 
and assume that the set~$E^-$ of negative edges spans $c$ trees in~$G$ for some constant~$c$.
We present a polynomial-time algorithm that computes a solution for $(G,w,s,t)$ with minimum total weight, 
or correctly concludes that no solution exists for~$(G,w,s,t)$. 
The running time of our algorithm is $O(n^{2c+9})$ where $n=|V|$, 
so in the language of parameterized complexity, our algorithm is in~$\mathsf{XP}$ with respect to the parameter~$c$.

In Section~\ref{sec:mainAlg} we present the main ideas and definitions necessary for our algorithm, 
and provide its high-level description together with some further details. 
We will distinguish between so-called \emph{separable} and \emph{non-separable} solutions. 
Finding an optimal and separable solution will be relatively easy, requiring extensive guessing but 
only standard techniques for computing minimum-cost flows.
By contrast, finding an optimal but non-separable solution is much more difficult. 
Therefore, in Section~\ref{sec:properties} we collect useful properties of optimal, non-separable solutions. 
The observations of Section~\ref{sec:properties} form the basis for an 
important subroutine necessary for finding optimal, non-separable solutions;
this subroutine is presented in Section~\ref{sec:partsol}. 
We close the section by providing a detailed pseudo-code description of our algorithm, 
and proving Theorem~\ref{thm:DISP-main} in Section~\ref{sec:permdis}.

\subsection{The Algorithm}
\label{sec:mainAlg}
We distinguish between two types of solutions for our instance $(G,w,s,t)$ of \DISP{}. 

\begin{definition}[{\bf Separable solution}]
\label{def:separable}
Let $(P_1,P_2)$ be a solution for $(G,w,s,t)$. We say that $P_1$ and $P_2$ are  \emph{separable}, if either
\begin{itemize}
\item they are \emph{not in contact}, i.e., there is no tree~$T \in \T$ that shares distinct vertices with both~$P_1$ and~$P_2$, or
\item there is a \emph{unique} tree~$T \in \T$ such that $P_1$ and~$P_2$ are in contact at~$T$, 
but the intersection of~$T$ with both $P_1$ and~$P_2$ is a path, possibly containing only a single vertex;
\end{itemize}
otherwise they are \emph{non-separable}.
\end{definition}

In Section~\ref{sec:sepsol} we show how to find an optimal, separable solution, whenever such a solution exists for~$(G,w,s,t)$.
Section~\ref{sec:nonsep-highlevel} deals with the case when we need to find an optimal, non-separable solution.
Our algorithm for the latter case is more involved, and relies on a subroutine that is based on dynamic programming and is developed 
throughout Sections~\ref{sec:properties} and~\ref{sec:partsol}.
The existence of this subroutine  (namely, Algorithm~\ref{alg:PermDisj}) is stated in Corollary~\ref{cor:perm-disjoint-paths}.

\subsubsection{Finding Separable Solutions}
\label{sec:sepsol}
Suppose that $(P_1,P_2)$ is a minimum-weight solution for $(G,w,s,t)$ that is separable.
The following definition establishes conditions when we are able to apply a simple strategy 
for finding a minimum-weight solution using well-known flow techniques.

\begin{definition}[{\bf Strongly separable solution}]
\label{def:stronglyseparable}
Let $(P_1,P_2)$ be a solution for $(G,w,s,t)$. We say that $P_1$ and $P_2$ are \emph{strongly separable}, 
if they are separable and they are either not in contact at any tree of~$\T$, or 
they are contact at a tree of~$\T$ that contains~$s$ or~$t$.
\end{definition}

Suppose that $P_1$ and $P_2$ are either not in contact, 
or they are in contact at a tree of~$\T$ that contains~$s$ or~$t$.
Let $\T^{\not\ni s,t}$ denote the set of all trees in~$\T$ that contain neither~$s$ nor $t$.
For each $T \in \T^{\not\ni s,t}$ that shares a vertex with~$P_i$ for some~$i \in [2]$, 
we define $a_T$ as the first vertex on~$P_i$ (when traversed from~$s$ to~$t$) that is contained in~$T$;
note that $T$ cannot share vertices with both~$P_1$ and $P_2$ as they are strongly separable, so $P_i$ is uniquely defined.

Our approach is the following: we guess the vertex~$a_T$ for each~$T \in \T^{\not\ni s,t}$, 
and then compute a minimum-cost flow in an appropriately defined network. 
More precisely, for each possible choice of vertices $Z=\{z_T \in V(T): T \in \T^{\not\ni s,t}\}$,
we build a network~$N_Z$ as follows. 

\begin{definition}[{\bf Flow network $N_Z$ for strongly separable solutions}]
\label{def:N_z} Given a set~$Z \subseteq V$ such that $Z \cap V(T)\{z_T\}$ for each $T \in \T^{\not\ni s,t}$, we create~$N_Z$ as follows.
We direct each non-negative edge in~$G$ in both directions. 
Then for each $T \in \T^{\not\ni s,t}$, we direct the edges of~$T$ away from~$z_T$.\footnote{Directing a tree~$T$ away from a vertex~$z \in V(T)$ 
means that an edge~$uv$ in~$T$ becomes an arc~$(u,v)$ if and only if $T[z,u]$ has fewer edges than $T[z,v]$;
directing $T$ towards~$z$ is defined analogously.}
If some~$T \in \T$ contains~$s$, then we direct all edges of~$T$ away from~$s$; 
similarly, if some~$T \in \T$ contains~$t$, then we direct all edges of~$T$ towards~$t$.
We assign a capacity of~$1$ to each arc and to each vertex\footnote{The standard network flow model can be adjusted by well-known techniques to allow for vertex capacities.} in the network except for~$s$ and~$t$, 
and we retain the cost function~$w$ (meaning that we define $w(\overrightarrow{e})$ as $w(e)$ for any arc~$\overrightarrow{e}$ obtained by directing some edge~$e$). 
We let $s$ and~$t$ be the source and the sink in~$N_Z$, respectively.
\end{definition}


\begin{lemma}
\label{lem:separablesol-noncontact}
If there exists a strongly separable solution for~$(G,w,s,t)$ with weight~$k$, 
then there exists a flow of value~2 having cost~$k$ in the network~$N_Z$ for some choice of~$Z \subseteq V$ 
containing exactly one vertex from each tree in~$\T^{\not\ni s,t}$. 
Conversely, a flow of value~2 and cost~$k$ in the constructed network~$N_Z$ for some $Z$ 
yields a solution for~$(G,w,s,t)$ with weight at most~$k$. 
\end{lemma}

\begin{proof}
Suppose first that $(P_1,P_2)$ is a minimum-weight solution for our instance, and suppose that
either $P_1$ and~$P_2$ are not in contact, or they are in contact at a tree of~$\T$ that contains~$s$ or~$t$.
Let $\overrightarrow{P_1}$ and $\overrightarrow{P_2}$ denote the paths obtained by directing the edges $P_1$ and $P_2$, respectively, 
as they are traversed from~$s$ to~$t$. 

First, since $P_1$ and $P_2$ are separable,
we know that the intersection of each of these paths with any tree~$T \in \T$ either consists of at most one vertex, 
or it is a single path within~$T$.
If some $T \in \T$ contains~$s$ (or $t$), then it is clear that the intersection of~$P_1$ and~$P_2$ with~$T$ consists of two subpaths of~$T$ 
starting at~$s$ (ending at~$t$, respectively);
hence, they are contained as directed paths in~$N_Z$ for any choice of~$Z$, because all edges of~$T$ are directed away from~$s$
(towards~$t$, respectively).
Now, consider a tree~$T \in \T^{\not\ni s,t}$ that shares a vertex with $P_1$ or~$P_2$. 
As $P_1$ and~$P_2$ cannot be in contact at~$T$ as they are strongly separable, only one of~$P_1$ and~$P_2$ may share a vertex with~$T$.
Let $P_i$ denote this path, and $a_T$ the first vertex on~$P_i$ contained in~$T$.
Then directing the edges of~$T$ away from~$a_T$ yields an orientation in which the intersection of~$P_i$ with $T$ 
is present as a directed path.
Summarizing all this, it follows that 
$\overrightarrow{P_1}$ and $\overrightarrow{P_2}$ are paths contained in the network~$N_Z$ defined by~$Z^\star=\{a_T: T \in \T^{\not\ni s,t}\}$.
Thus, there exists a flow of value~2 with cost~$w(P_1)+w(P_2)$ in the network~$N_{Z^\star}$.

For the other direction, assume that there is a flow of value~2  in $N_Z$ for some vertex set~$Z$ 
containing a vertex from each tree in~$\T^{\not\ni s,t}$. 
Such a flow must flow through two openly disjoint path from~$s$ to~$t$, due to the unit capacities in~$N_Z$.
Moreover, by the conservativeness of~$w$ (and since each edge of~$T$ is directed in only one direction), 
there are no cycles with negative total cost in the network.
This implies that the cost of these two paths is at most the cost of our flow. 
Hence, the lemma follows.  
\end{proof}

Next, we show how to deal with the case when a minimum-weight solution $(P_1,P_2)$ is separable, but not strongly separable;
then there is a unique tree~$T \in \T$ at which $P_1$ and~$P_2$ are in contact.
In such a case, we can simply delete an edge from~$T$ in a way that paths~$P_1$ and~$P_2$ 
cease to be in contact in the resulting instance.
This way, we can reduce our problem to the case when there is a strongly separable optimal solution;
note, however, that the number of trees spanned by the negative edges (our parameter~$c$) increases by~$1$.

\begin{lemma}
\label{lem:separablesol-contact}
If there exists a separable, but not strongly separable solution for~$(G,w,s,t)$ with weight~$k$,
then there exists an edge~$e \in E^-$ such that setting $G'=G-e$ and $E'=E \setminus \{e\}$, 
the instance~$(G',w_{|E'},s,t)$ admits a strongly separable solution with weight at most~$k$.
Conversely, a solution for~$(G',w_{|E'},s,t)$ is also a solution for~$(G,w,s,t)$ with the same weight. 
\end{lemma}

\begin{proof}
If $P_1$ and~$P_2$ are separable but not strongly separable, then there exists a tree~$T \in \T$ 
that contains a vertex from both paths but neither~$s$ nor~$t$. 
Since $P_1$ and~$P_2$ are separable, 
they both intersect~$T$ in a path (containing possibly only a single vertex); let $P^T_1$ and~$P^T_2$ denote these paths.
Due to the permissive disjointness of~$P_1$ and~$P_2$ and since $s,t \notin V(T)$, 
the paths $P^T_1$ and~$P^T_2$ are vertex-disjoint. 
Consequently, there exists an edge~$e \in T \setminus (P_1 \cup P_2)$ such that $P^T_1$ and~$P^T_2$ are in different components of~$T-e$.
Thus, setting $G'=G-e$ and $E'=E \setminus \{e\}$,
we get that $P_1$ and $P_2$ are not in contact in the instance $(G',w_{|E'},s,t)$, 
since they are not in contact in $(G,w,s,t)$ at any negative tree other than~$T$ (by Definition~\ref{def:separable}).
Hence, they form a strongly separable solution for~$(G',w_{|E'},s,t)$, proving the first statement of the lemma. 

The second statement of the lemma is trivial.
\end{proof}

Thanks to Lemmas~\ref{lem:separablesol-noncontact} and~\ref{lem:separablesol-contact}, 
if there exists a minimum-weight solution for~$(G,w,s,t)$ that is separable, then we can find some minimum-weight solution 
using standard algorithms for computing minimum-cost flows.
In Section~\ref{sec:nonsep-highlevel} we explain how we can find a minimum-weight non-separable solution for~$(G,w,s,t)$.

\subsubsection{Finding Non-separable Solutions}
\label{sec:nonsep-highlevel}
To find a minimum-weight solution for~$(G,w,s,t)$ that is not separable, we need a more involved approach.
We now provide a high-level presentation of our algorithm for finding a non-separable solution of minimum weight. 
We remark that Algorithm~\ref{alg:DISP} contains a pseudocode; however, we believe that it is best to read the following description first.

\begin{description}
\item[Step 1.] We guess certain properties of a minimum-weight non-separable solution $(P_1,P_2)$:
First, we guess a tree~$T \in \T$ that shares distinct vertices both with~$P_1$ and~$P_2$, i.e., a tree at which $P_1$ and~$P_2$ are in contact.
Second, we guess a partition $(\T_s,\T_0,\T_t)$ of~$\T \setminus \{T\}$ that is $T$-valid with respect to~$(P_1,P_2)$. 
\item[Step 2.] If $\T_s \neq \emptyset$, then 
we guess the first vertex of~$P_1$ and of~$P_2$ contained in~$V(T)$, denoted by~$a_1$ and~$a_2$, respectively.
We use recursion to compute two permissively disjoint $(s,\{a_1,a_2\})$-paths using only the negative trees in~$\T_s$, 
and to compute two permissively disjoint $(\{a_1,a_2\},t)$-paths using only the negative trees in~$\T \setminus \T_s$. 
Observe that by~$\T_s \neq \emptyset$ and $T \in \T \setminus \T_s$, we search for these paths in graphs that contain 
only a strict subset of the negative trees in~$\mathcal{T}$, 
meaning that our parameter~$c$ strictly decreases in both constructed sub-instances. 
We combine the obtained pairs of paths into a solution by using Lemma~\ref{lem:uncrossing}.
\\
We proceed in a similar fashion when $\T_t \neq \emptyset$. 
\item[Step 3.] If $\T_s=\T_t=\emptyset$, then for both $i \in [2]$
we guess the first and last vertex of~$P_i$ contained in~$V(T)$, denoted by~$a_i$ and~$b_i$, respectively.
We apply standard flow techniques to compute two  $(s,\{a_1,a_2\})$-paths and 
two $(\{b_1,b_2\},t)$-paths  with no inner vertices in~$V(\T)$ that are pairwise permissively disjoint.
Then, we apply the polynomial-time algorithm we devise for computing a pair of permissively disjoint $(\{a_1,a_2\},\{b_1,b_2\})$-paths in~$G$. This algorithm is the cornerstone of our method, 
and is based on important structural observations that allow for efficient dynamic programming.
We combine the obtained pairs of paths into a solution by applying Lemma~\ref{lem:uncrossing}.
\item[Step 4.] We output a solution of minimum weight among all solutions found in Steps~2 and~3.
\end{description}

Let us now provide more details about these steps; see also Algorithm~\ref{alg:DISP}.

\paragraph*{Step 1: Initial guesses on~$\T$.}
There are $c=|\T|$ possibilities to choose~$T$ from~$\T$, and there 
are $3^{c-1}$ further possibilities to partition~$\T \setminus \{T\}$ into $(\T_s,\T_0,\T_t)$, 
yielding~$c 3^{c-1}$ possibilities in total for our guesses. 

\paragraph*{Step 2: Applying recursion.}
To apply recursion in Step~2 when $\T_s \neq \emptyset$, we guess two distinct vertices $a_1$ and~$a_2$ in~$T$, 
and define the following sub-instances.
\begin{definition}[{\bf Sub-instances for $\T_s \neq \emptyset$}]
\label{def:sub-instances-Ts}
For some $\emptyset \neq \T_s \subseteq \T$ and distinct vertices~$a_1$ and~$a_2$, let us first create a graph~$G_s$ by
adding a new vertex~$a^\star$, and connecting it to each of~$a_1$ and~$a_2$ with an edge of weight~$w_{a^\star}=|w(T[a_1,a_2])|/2$.  
We create instances~$I_s^1$ and~$I_s^2$ 
 as follows:
\begin{itemize}
\item to get~$I_s^1$,  we delete $E(\T \setminus \T_s)$ and $V(\T \setminus \T_s) \setminus \{a_1,a_2\}$ from~$G_s$,
 and designate~$s$ and~$a^\star$ as our two terminals; 
\item to get~$I_s^2$,  we delete $V(\T_s)$ from~$G_s$,
 and designate~$a^\star$ and~$t$ as our two terminals. 
\end{itemize}
\end{definition}

We use recursion to solve \DISP{} on sub-instances~$I_s^1$ and~$I_s^2$. 
Note that conservativeness is maintained for both $I_s^1$ and $I_s^2$, due to our choice of $w_{a^\star}$ and 
statement~(1) of \cref{lem:T-min-pathlength}.
If both sub-instances admit solutions, we obtain two permissively disjoint~$(s,\{a_1,a_2\})$-paths 
$Q^{\mynearrow}_1$ and~$Q^{\mynearrow}_2$ from our solution to~$I_s^1$ by deleting the vertex~$a^\star$.
Similarly, we obtain two permissively disjoint~$(\{a_1,a_2\},t)$-paths $Q^{\mysearrow}_1$ and~$Q^{\mysearrow}_2$ 
from our solution to~$I_s^2$ by deleting the vertex~$a^\star$.
Next we use Lemma~\ref{lem:uncrossing} to create a solution for our original instance~$(G,w,s,t)$ 
using paths~$Q^{\mynearrow}_1,Q^{\mynearrow}_2,Q^{\mysearrow}_1$, 
and $Q^{\mysearrow}_2$.

\smallskip

If $\T_t \neq \emptyset$, then we proceed similarly:
after guessing two distinct vertices~$b_1$ and~$b_2$ in~$T$, we use a construction analogous to Definition~\ref{def:sub-instances-Ts}.

\begin{definition}[{\bf Sub-instances for $\T_t \neq \emptyset$}]
\label{def:sub-instances-Tt}
For some $\emptyset \neq \T_s \subseteq \T$ and distinct vertices~$b_1$ and~$b_2$, let us first create a graph~$G_t$ by
adding a new vertex~$b^\star$, and connecting it to each of~$b_1$ and~$b_2$ with an edge of weight~$w_{b^\star}=|w(T[b_1,b_2])|/2$.  
We create instances~$I_t^1$ and~$I_t^2$ 
 as follows:
\begin{itemize}
\item to get~$I_t^1$,  we delete $V(\T_t)$ from~$G_t$,
 and designate~$s$ and~$b^\star$ as our two terminals;
\item to get~$I_t^2$,  we delete $E(\T \setminus \T_t)$ and $V(\T \setminus \T_t) \setminus \{b_1,b_2\}$ from~$G_t$
 and designate~$b^\star$ and~$t$ as our two terminals.
\end{itemize}
\end{definition}

We use recursion to solve \DISP{} on sub-instances~$I_t^1$ and~$I_t^2$;
again, conservativeness is ensured for $I_t^1$ and $I_t^2$ by our choice of $w_{b^\star}$ and 
statement~(1) of \cref{lem:T-min-pathlength}. 
If both sub-instances admit solutions, we obtain two permissively disjoint~$(s,\{b_1,b_2\})$-paths 
$Q^{\mynearrow}_1$ and~$Q^{\mynearrow}_2$ from our solution to~$I_t^1$ 
by deleting the vertex~$b^\star$.
Similarly, we obtain two permissively disjoint~$(\{b_1,b_2\},t)$-paths $Q^{\mysearrow}_1$ and~$Q^{\mysearrow}_2$ 
from our solution to~$I_t^2$ by deleting the vertex~$b^\star$.
Again, we use Lemma~\ref{lem:uncrossing} to create a solution for our original instance~$(G,w,s,t)$ 
using paths~$Q^{\mynearrow}_1,Q^{\mynearrow}_2,Q^{\mysearrow}_1$, 
and $Q^{\mysearrow}_2$.

\smallskip
We state the correctness of Step~2 in the following lemma: 
\begin{lemma}
\label{lem:subinstances-correct}
Suppose that $(P_1,P_2)$ is a minimum-weight solution for~$(G,w,s,t)$ such that 
\begin{itemize}
\item $P_1$ and $P_2$ are in contact at some~$T \in \T$,
\item $a_i$ and $b_i$ are the first and last vertices of~$P_i$ contained in~$T$, respectively, for $i \in [2]$, and
\item $(\T_s,\T_0,\T_s)$ is a $T$-valid partition w.r.t.~$(P_1,P_2)$.
\end{itemize}
Then instances $I_s^1$ and $I_s^2$ admit solutions whose total weight (summed over all four paths) is~$w(P_1)+w(P_2)+4w_{a^\star}$.
Furthermore, given a solution $\mathcal{S}_i$ for~$I_s^i$ for both $i \in [2]$, 
we can compute  in linear time a solution for~$(G,w,s,t)$ 
of weight at most 
$w(\mathcal{S}_1)+w(\mathcal{S}_2)-4w_{a^\star}$. \\
The same holds when substituting~$I_s^1, I_s^2,$ and $w_{a^\star}$ with~$I_t^1,I_t^2$, and~$w_{b^\star}$  in these claims.
\end{lemma}

\begin{proof}
Since $(\T_s,\T_0,\T_t)$ is a $T$-valid partition for $(P_1,P_2)$, as defined by the conditions in Lemma~\ref{lem:partitioned-comps},
we know that  $P_i[s,a_i]$ contains no vertices from $V(\T_t \cup \T_0 \setminus \{T\})$.
By the definition of~$a_i$, the only vertex on~$P_i[s,a_i]$ contained in~$V(T)$ is~$a_i$. 
Thus, $P_1[s,a_1]$ and $P_2[s,a_2]$ are paths in~$I_s^1$; adding the edges~$a_1 a^\star$ and $a_2 a^\star$, each of weight~$w_{a^\star}$, to these paths 
yields two permissively disjoint $(s,a^\star)$-paths in~$I_s^1$. 
Similarly, $P_1[a_1,t]$ and $P_2[a_2,t]$ are paths in~$I_s^2$ due to the definition of a $T$-valid partition;
adding the edges~$a_1 a^\star$ and $a_2 a^\star$ to them yields 
two permissively disjoint $(a^\star,t)$-paths in~$I_s^2$. 
Clearly, the total weight of these four paths (two in~$I_s^1$ and two in~$I_s^2$) is exactly $w(P_1)+w(P_2)+4w_{a^\star}$.

As described in Step~2, a solution~$\mathcal{S}_1$ for~$I_s^1$ yields 
two permissively disjoint $(s,\{a_1,a_2\})$-paths $Q^{\mynearrow}_1$ and~$Q^{\mynearrow}_2$ of
weight~$w(\mathcal{S}_1)-2w_{a^\star}$, where the term $-2w_{a^\star}$  is due to the deletion of edges~$a_1 a^\star$ 
and~$a_2 a^\star$. 
Similarly, a solution~$\mathcal{S}_2$ for~$I_s^2$ yields 
 two permissively disjoint $(\{a_1,a_2\},t)$-paths $Q^{\mysearrow}_1$ and~$Q^{\mysearrow}_2$
 of weight~$w(\mathcal{S}_2)-2w_{a^\star}$.

We can assume that both these path pairs are locally cheapest, 
as otherwise we can amend all shortcuts on them while preserving permissive disjointness and without increasing their weight.
Thus we can apply Lemma~\ref{lem:uncrossing} by setting 
$p_1=p_2=s$, $\{v_1,v_2\}=\{a_1,a_2\}$, $q_1=q_2=t$, $P_i=Q^{\mynearrow}_i$ and $Q_i=Q^{\mysearrow}_i$ for $i \in [2]$,
$\T_1=\T_s$ and $\T_2=\T \setminus \T_s$. 
Note that the conditions Lemma~\ref{lem:uncrossing} hold, since (i) $Q^{\mynearrow}_1 \cup Q^{\mynearrow}_2$ 
contains no vertices of~$T$ other than~$a_1$ and~$a_2$, and (ii) $Q^{\mynearrow}_1 \cup Q^{\mynearrow}_2$ 
contains no edge of~$\T_t \cup \T_0 \cup \{T\}$, while $Q^{\mysearrow}_1 \cup Q^{\mysearrow}_2$ 
contains no edge of~$\T_s$. Therefore,  Lemma~\ref{lem:uncrossing} yields two permissively disjoint $(s,t)$-paths with weight at most
$\sum_{i=1}^2 w(Q^{\mynearrow}_i)+w(Q^{\mysearrow}_i)=w(\mathcal{S}_1)+w(\mathcal{S}_1)-4w_{a^\star}$.

To see the first claim for subinstances~$I_t^1$ and~$I_t^2$, one can use analogous arguments in a straightforward manner
to prove that $I_t^1$ and $I_t^2$ admits solutions whose total weight is at most~$w(P_1)+w(P_2)+4w_{b^\star}$.
To show the second statement, it is clear that a solution for~$I_t^1$ and for~$I_t^2$
yields two permissively disjoint $(s,\{b_1,b_2\})$-paths $Q^{\mynearrow}_1$ and~$Q^{\mynearrow}_2$
and two permissively disjoint $(\{b_1,b_2\},t)$-paths $Q^{\mysearrow}_1$ and~$Q^{\mysearrow}_2$, 
and we can assume that these path pairs are locally cheapest. 
Thus we can we apply Lemma~\ref{lem:uncrossing}, but note that we need to set $p_1=p_2=t$, $q_1=q_2=s$,
$\{v_1,v_2\}=\{b_1,b_2\}$,  $P_i=Q^{\mysearrow}_i$ and $Q_i=Q^{\mynearrow}_i$ for $i \in [2]$,
$\T_1=\T_t$ and $\T_2=\T \setminus \T_t$. It is straightforward to verify using the arguments of the previous paragraph
that Lemma~\ref{lem:uncrossing} then yields a solution for~$(G,w,s,t)$ with the desired bound on the weight.
\end{proof}

\paragraph*{Step 3: Applying flow techniques and dynamic programming.}
We now describe Step~3 in more detail, which concerns the case when $\T_s=\T_t=\emptyset$.

First, we guess vertices~$a_1,b_1, a_2, $ and~$b_2$;
the intended meaning of these vertices is that $a_i$ and~$b_i$ are the first and last vertices of~$P_i$ contained in~$V(T)$, for both~$i \in [2]$. 
We only consider guesses that are \emph{reasonable}, meaning that they satisfy the following conditions: 
\begin{itemize}
\item  if $s \in V(T)$, then $s=a_1=a_2$, otherwise $a_1 \neq a_2$;
\item  if $t \in V(T)$, then $t=b_1=b_2$, otherwise $b_1 \neq b_2$;
\item $T[a_1,b_1]$ and $T[a_2,b_2]$ share at least one edge.
\end{itemize}

Then we compute four paths from~$\{s,t\}$ to~$\{a_1,b_1,a_2,b_2\}$ in the graph $G-E(\T)-(V(\T) \setminus \{a_1,a_2,b_1,b_2\})$ with minimum weight 
such that two paths have~$s$ as an endpoint, the other two have~$t$ as an endpoint, and no other vertex appears on more than one path. 
To this end, we define the following network and compute a minimum-cost flow of value~4 in it.

\begin{definition}[{\bf Flow network $N_{(a_1,b_1,a_2,b_2)}$ for non-separable solutions.}]
\label{def:N_aabb} Given vertices $a_1, b_1,a_2,$ and $b_2$, we create~$N_{(a_1,b_1,a_2,b_2)}$ as follows.
First, we delete all edges in~$E^-$ and all vertices in~$V(\T) \setminus \{a_1,a_2,b_1,b_2\}$ from~$G$, 
and direct each edge in~$G$ in both directions. 
We then add new vertices~$s^\star$ and~$t^\star$, along with arcs~$(s^\star,s)$ and $(s^\star,t)$ of capacity~$2$, 
and arcs $(a_i,t^\star)$ and~$(b_i,t^\star)$ for $i=1,2$  with capacity~1.\footnote{In the degenerate case when $s=a_1=a_2$ or $t=b_1=b_2$ this yields two parallel arcs from~$s$ or $t$ to $t^\star$.
} We assign capacity~1 to all other arcs, and also to each vertex of~$V(G)\setminus \{s,t\}$. 
All newly added arcs will have cost~$0$, otherwise we retain the cost function~$w$.
We let $s^\star$ and~$t^\star$ be the source and the sink in~$N_{(a_1,b_1,a_2,b_2)}$, respectively.
\end{definition}

Next, we compute 
two $(\{a_1,a_2\},\{b_1,b_2\})$-paths $Q_1$ and~$Q_2$ in $G$ with minimum total weight  that are permissively disjoint. 
A polynomial-time computation for this problem, 
building on structural observations from Section~\ref{sec:properties}, is provided in Section~\ref{sec:partsol} 
(see Corollary~\ref{cor:perm-disjoint-paths}).
Finally, we apply Lemma~\ref{lem:uncrossing} (in fact, twice) to obtain a solution to our instance~$(G,w,s,t)$
based on a minimum-cost flow of value~4 in~$N_{(a_1,b_1,a_2,b_2)}$ and paths~$Q_1$ and~$Q_2$. 
We finish this section with Lemma~\ref{lem:type2bsol}, stating the correctness of Step~3.

\begin{lemma}
\label{lem:type2bsol}
Suppose that $(P_1,P_2)$ is a minimum-weight solution for~$(G,w,s,t)$ such that 
\begin{itemize}
\item $P_1$ and $P_2$ are non-separable, and in contact at some~$T \in \T$, 
\item $a_i$ and $b_i$ are the first and last vertices of~$P_i$ contained in~$T$, respectively, for $i \in [2]$, and
\item $(\emptyset,\T  \setminus \{T\},\emptyset)$ is a $T$-valid partition w.r.t.~$(P_1,P_2)$.
\end{itemize}
Let $Q_1$ and $Q_2$ be two permissively disjoint $(\{a_1,a_2\},\{b_1,b_2\})$-paths in~$G$ with minimum total weight.	
Then the following holds:\\
If $w^\star$ is the minimum cost of a flow of value~4 in the network~$N_{(a_1,b_1,a_2,b_2)}$, then 
$w^\star+w(Q_1)+w(Q_2) \leq w(P_1)+w(P_2)$. 
Conversely, given a flow of value~4 in the network~$N_{(a_1,b_1,a_2,b_2)}$ with cost~$w^\star$, 
together with paths~$Q_1$ and~$Q_2$, 
we can find a solution for~$(G,w,s,t)$ with cost at most $w^\star+w(Q_1)+w(Q_2)$ in linear time. 
\end{lemma}

\begin{proof}
Suppose that $(P_1,P_2)$ is a non-separable solution for our instance fulfilling the conditions of the lemma, 
and $Q_1$ and~$Q_2$ are two permissively disjoint $(\{a_1,a_2\},\{b_1,b_2\})$-paths in~$G$.
Clearly, $P_1[a_1,b_1]$ and $P_2[a_2,b_2]$ are permissively disjoint, so we have
\begin{equation}
\label{eqn:nonsepsol-e1}
w(Q_1)+w(Q_2) \leq w(P_1[a_1,b_1])+w(P_2[a_2,b_2]).
\end{equation}
The four paths $P_i[s,a_i]$ and $P_i[b_i,t]$, $i \in [2]$, do not share vertices other than~$s$ and~$t$,
and since $(\emptyset,\T \setminus \{T\},\emptyset)$ is a $T$-valid partition w.r.t.~$(P_1,P_2)$, 
we know that they do not contain any vertex of $V(T) \setminus \{a_1,a_2,b_1,b_2\}$ (recall condition~(iii) of Lemma~\ref{lem:partitioned-comps}). 
Therefore, these four paths are present (as directed paths) in~$N_{(a_1,b_1,a_2,b_2)}$, and 
together with the arcs incident to~$s^\star$ and~$t^\star$ yield a flow of value~4 in~$N_{(a_1,b_1,a_2,b_2)}$.
This implies
$w^\star \leq \sum_{i \in [2]} w(P_i[s,a_i])+w(P_i[b_i,t])$, which together with~(\ref{eqn:nonsepsol-e1}) yields 
\begin{align*}
w^\star + w(Q_1)+w(Q_2) & \leq  \sum_{i \in [2]} \bigg( w(P_i[s,a_i])+w(P_i[b_i,t]) \bigg) + w(P_1[a_1,b_1])+w(P_2[a_2,b_2]) \\
& = w(P_1)+w(P_2),
\end{align*}
as required. This proves the first statement of the lemma.

\smallskip
For the other direction, we first show the following claim. 

\begin{claim} 
\label{clm:nonsep-intersection}
$T[a_1,b_1]$ and $T[a_2,b_2]$ share at least one edge.
\end{claim}
\begin{claimproof}
Suppose for the sake of contradiction that the claim does not hold. 
Note that $T[a_1,b_1]$ and $T[a_2,b_2]$ may only share an endpoint if $a_1=a_2=s$ or if $b_1=b_2=t$.
Thus, if $T[a_1,b_1]$ and $T[a_2,b_2]$ share no edge, then 
replacing the subpaths~$P_1[a_1,b_1]$ and~$P_2[a_2,b_2]$ with $T[a_1,b_1]$ and $T[a_2,b_2]$, 
respectively, yields a solution~$(P'_1,P'_2)$ for~$(G,w,s,t)$.
where $P'_1=P_1 \setminus P_1[a_1,b_1] \cup T[a_1,b_1]$ and $P'_2=P_1 \setminus P_2[a_2,b_2] \cup T[a_2,b_2]$.
By statement~(1) of Lemma~\ref{lem:T-min-pathlength} we know that $w(P_1[a_1,b_1])\geq  w(T[a_1,b_1])$ 
and $w(P_2[a_2,b_2])\geq  w(T[a_2,b_2])$, which implies 
$w(P'_1)+w(P'_2)\leq w(P_1)+w(P_2)$. 
Since $(P_1,P_2)$ is a minimum-weight solution for~$(G,w,s,t)$, we know that equality must hold.
Using the last sentence in statement~(1) of Lemma~\ref{lem:T-min-pathlength}, this implies
$P_1[a_1,b_1] \supseteq T[a_1,b_1]$ and $P_2[a_2,b_2] \supseteq T[a_2,b_2]$. 
As $P_1$ and $P_2$ are paths, we immediately get $P_1[a_1,b_1] = T[a_1,b_1]$ and $P_2[a_2,b_2] = T[a_2,b_2]$. 

Recall that $(\emptyset, \T \setminus \{T\}, \emptyset)$ is a $T$-valid partition with respect to~$(P_1,P_2)$.
Thus, $P_1 \cap V(T')=\emptyset$ for any $T' \in \T \setminus \{\T\}$ follows from
$P_1[a_1,b_1] = T[a_1,b_1]$ and $P_2[a_2,b_2] = T[a_2,b_2]$. However, in this case $P_1$ and~$P_2$ are separable, a contradiction proving the claim.
\end{claimproof}

\medskip
Consider now a flow~$f$ of value~4 in~$N_{(a_1,b_1,a_2,b_2)}$.
Let $P^f_{s,1}$ and $P^f_{s,2}$ ($P^f_{t,1}$ and $P^f_{t,2}$) be the two paths
indicated by~$f$ from~$s$ (from~$t$, respectively) to two vertices  of $\{a_1,b_1,a_2,b_2\}$.
Let~$X=\T[a_1,b_1] \cap T[a_2,b_2]$.
If $a_1$ and~$a_2$ are not in the same connected component of~$T \setminus X$, then switch the names of~$a_2$ and~$b_2$. 
This way, we ensure that $a_1$ and~$a_2$ are in the same connected component of~$T \setminus X$.
In particular, $T[a_1,a_2]$ and $T[b_1,b_2]$ share no edges. 
We now distinguish between two cases.

\smallskip
{\bf Case A:}
 if the set of the endpoints of $P^f_{s,1}$ and $P^f_{s,2}$ is $\{s,a_i,b_j\}$ for some $i,j \in [2]$. 
Then, clearly, the set of endpoints of $P^f_{t,1}$ and $P^f_{t,2}$ is $\{t,a_{3-i},b_{3-j}\}$. Hence, the union of the paths $T[a_1,a_2]$ and $T[b_1,b_2]$ (using only edges of~$E(T)$) as well as the four paths $P^f_{s,i}$ and $P^f_{t,i}$ for $i \in [2]$ (using only edges of~$E(G) \setminus E^-$) yields two openly disjoint $(s,t)$-paths $S_1$ and $S_2$. Moreover,
by the second statement of Lemma~\ref{lem:T-min-pathlength}
we get
\[
w(S_1)+w(S_2) =w^\star + w(T[a_1,a_2])+w(T[b_1,b_2]) \leq w^\star + w(Q_1)+w(Q_2).
\]
Note that the paths~$S_1$ and~$S_2$ can be computed in linear time, given our flow of value~4. 
Hence, the second statement of the lemma holds in this case.

\smallskip
{\bf Case B:}
if the set of the endpoints of $P^f_{s,1}$ and $P^f_{s,2}$ is $\{s,a_1,a_2\}$ 
(the case when the two paths starting from~$s$ lead to~$b_1$ and~$b_2$ is symmetric). 
Then the set of the endpoints of $P^f_{t,1}$ and $P^f_{t,2}$ is $\{t,b_1,b_2\}$. 
We need to apply Lemma~\ref{lem:uncrossing} twice. 

First, if $a_1 \neq a_2$, then we apply Lemma~\ref{lem:uncrossing} for paths $P^f_{s,1}$, $P^f_{s,2}$, $Q_1$ and $Q_2$, 
i.e., setting $s=p_1=p_2$, $\{v_1,v_2\}=\{a_1,a_2\}$, $\{q_1,q_2\}=\{b_1,b_2\}$, $\T_1=\emptyset$ and~$\T_2=\T$. 
Note that the conditions of Lemma~\ref{lem:uncrossing} hold: 
by the definition of the network~$N_{(a_1,b_1,a_2,b_2)}$ we know that $P^f_{s,1} \cup P^f_{s,2}$ contains no edge of~$E^-$ and no vertices of~$T$ other than~$a_1$ and~$a_2$; 
moreover, $Q_1$ and~$Q_2$ are locally cheapest $(\{a_1,a_2\},\{b_1,b_2\})$-paths, as they have minimum total weight (recall Observation~\ref{obs:locally-cheapest}).
Hence, we obtain in linear time two permissively disjoint $(s,\{b_1,b_2\})$-paths $Q'_1$ and~$Q'_2$ in~$G$ whose weight is at most $\sum_{i \in [2]} w(P^f_{s,i})+w(Q_i)$. 
If $a_1=a_2$, then we also know $s=a_1=a_2$, and thus 
$P^f_{s,1}$ and $P^f_{s,2}$ are paths of length~0. In this case we simply define $Q'_1=Q_1$ and $Q'_2=Q_2$.

We next check whether $Q'_1$ and~$Q'_2$ are locally cheapest, and if not, we apply Observation~\ref{obs:locally-cheapest}  to compute $\Amend(Q'_1,Q'_2)$,
obtaining two locally cheapest, permissively disjoint $(s,\{b_1,b_2\})$-paths $Q''_1$ and~$Q''_2$ with weight at most $w(Q'_1)+w(Q'_2)$.
By our remarks following Observation~\ref{obs:locally-cheapest}, this can be done in linear time.

If $b_1\neq b_2$, then we apply Lemma~\ref{lem:uncrossing} once again, for paths $P^f_{t,1}$, $P^f_{t,2}$, $Q''_1$ and $Q''_2$, 
i.e., setting $t=p_1=p_2$, $\{v_1,v_2\}=\{b_1,b_2\}$, $s=q_1=q_2$, $\T_1=\emptyset$ and~$\T_2=\T$. 
Observe again that the conditions of Lemma~\ref{lem:uncrossing} hold:
by the definition of the network~$N_{(a_1,b_1,a_2,b_2)}$ we know that $P^f_{t,1} \cup P^f_{t,2}$ contains no edge of~$E^-$ and no vertices of~$T$ other than~$b_1$ and~$b_2$; 
moreover, $Q''_1$ and~$Q''_2$ are locally cheapest $(\{b_1,b_2\},s)$-paths.
Hence, we obtain in linear time two permissively disjoint $(t,s)$-paths $S_1$ and $S_2$ in~$G$ such that
\begin{align}
\label{eq:sol2b-weights}
w(S_1)+w(S_2) & \leq \sum_{i \in [2]} w(P^f_{t,i})+w(Q''_i) \leq 
\sum_{i \in [2]} w(P^f_{t,i})+w(Q'_i) \\ \notag
& \leq 
\sum_{i \in [2]} w(P^f_{t,i})+w(Q_i) + w(P^f_{s,i}) 
\\ \notag
&= w^\star + w(Q_1)+w(Q_2).   
\end{align}
Thus, $S_1$ and $S_2$ form a solution for our instance of \DISP{} with weight at most $w^\star + w(Q_1)+w(Q_2)$, as promised.

If $b_1=b_2$, then we also know $t=b_1=b_2$, and thus 
$P^f_{t,1}$ and $P^f_{t,2}$ are paths of length~0. In this case $Q''_1$ and~$Q''_2$ are two permissively disjoint $(s,t)$-paths, so we simply define $S_1=Q''_1$ and $S_2=Q''_2$.
Note that Inequality~\ref{eq:sol2b-weights} still holds, and thus $S_1$ and $S_2$ form a solution for~$(G,w,s,t)$with weight at most $w^\star + w(Q_1)+w(Q_2)$ as promised.

Note that applying Lemma~\ref{lem:uncrossing} (twice) and amending all shortcuts in a pair of paths can be performed in linear time.
This finishes the proof of the lemma.
\end{proof}

\subsection{Properties of a Non-separable Solution}
\label{sec:properties}

Let us now turn our attention to the subroutine lying at the heart of our algorithm for \DISP{}: an algorithm that, 
given two source terminals and two sink terminals on some tree~$T \in \T$,
computes two permissively disjoint paths from the two source terminals to the two sink terminals, with minimum total weight. 
It is straightforward to see that any non-separable solution whose paths are in contact at~$T$ contains such a pair of paths. 
Therefore, as described in Section~\ref{sec:nonsep-highlevel}, finding such paths is a necessary step to computing an optimal, 
non-separable solution for our instance~$(G,w,s,t)$.
This section contains observations about the properties of such paths.

Let us now formalize our setting.
Let $a_1, a_2, b_1,$ and $b_2$ be vertices on a fixed tree~$T \in \T$ such that $T[a_1,b_1]$ and $T[a_2,b_2]$ intersect in a path~$X$ (with at least one edge), 
with one component of~$T \setminus X$ containing $a_1$ and $a_2$, and the other containing $b_1$ and $b_2$. 
Let the vertices on~$X$ be $x_1,\dots,x_r$ with $x_1$ being the closest to~$a_1$ and~$a_2$. 
We will use the notation $A_i=T[a_i,x_1]$ and $B_i=T[b_i,x_r]$ for each~$i \in [2]$.
For each $i \in [r]$, let $T_i$ be the maximal subtree of~$T$ containing~$x_i$ but no other vertex of~$X$.
We also define $T_{(i,j)}=\bigcup_{i \leq h \leq j} T_h$ for some $i$ and~$j$ with~$1 \leq i\leq j \leq r$.

For convenience, for any path $Q$ that has $a_i \in \{a_1,a_2\}$ as its endpoint,  we will say that $Q$ \emph{starts} at $a_i$ and \emph{ends} at its other endpoint. 
Accordingly, for vertices $u,v \in V(Q)$ we say that $u$ \emph{precedes} $v$  on~$Q$, or equivalently, $v$ \emph{follows}~$u$ on~$Q$, if $u$ lies on $Q[a_i,v]$. 
When defining a vertex as the ``first'' (or ``last'') vertex with some property on~$Q$ or on a subpath~$Q'$ of $Q$ then, unless otherwise stated, we mean the vertex on~$Q$ or on~$Q'$ that is closest to~$a_i$ (or farthest from~$a_i$, respectively) that has the given property.

We begin with the following observation about the paths $A_1$, $A_2$, $B_1$, and $B_2$.
\begin{lemma}
\label{lem:starting}
Let $Q_1$ and $Q_2$ be two permissively disjoint $(\{a_1,a_2\},\{b_1,b_2\})$-paths in~$G$ with minimum total weight.
For each $i \in [2]$, neither $V(A_i) \setminus \{x_1\}$ nor
  $V(B_i) \setminus \{x_r\}$ can contain vertices both from~$Q_1$ and~$Q_2$.
\end{lemma}

\begin{proof}
W.l.o.g.\ we suppose $a_i,b_i \in V(Q_i)$ for $i \in [2]$. See Figure~\ref{fig:starting} for an illustration of the proof.
For the sake of contradiction, assume first that there is a vertex on~$A_i$ other than~$x_1$ that is contained in~$Q_{3-i}$; 
by symmetry, we may also assume that $i=1$.
Note that in this case $a_1 \neq a_2$.
Let $z$ be the vertex of~$Q_2$ that is closest to~$a_1$ on~$T[a_1,x_1]$; then $V(T[a_1,z])\cap V(Q_2)=\{z\}$. 
Moreover, for each vertex $z'$ on~$T[a_1,z]$ that is on~$Q_1$, we must have $T[a_1,z'] \subseteq Q_1$, because $Q_1$ and~$Q_2$ are locally cheapest.

\begin{figure}
  \centering
    \begin{tikzpicture}[xscale=0.7, yscale=0.7]

  \node[draw, circle, fill=black, inner sep=1.7pt] (a1) at (-3, 1.5) {};
  \node[draw, circle, fill=black, inner sep=1.7pt] (z)  at (-1, 0.5) {};
  \node[draw, circle, fill=black, inner sep=1.7pt] (a2) at (-3, -1.5) {};
  \node[draw, circle, fill=black, inner sep=1.7pt] (x1) at (0, 0) {};
  \node[draw, circle, fill=black, inner sep=1.7pt] (y)  at (1, 0) {};
  \node[draw, circle, fill=black, inner sep=1.7pt] (xr) at (3, 0) {};
  \node[draw, circle, fill=black, inner sep=1.7pt] (b1) at (6, 1.5) {};
  \node[draw, circle, fill=black, inner sep=1.7pt] (b2) at (6, -1.5) {};
  \coordinate[inner sep=0pt] (h0) at (-2,1) {};
  \coordinate[inner sep=0pt] (h1) at (1.9,0) {};
  \coordinate[inner sep=0pt] (h2) at (2.4,0) {};  
  \coordinate[inner sep=0pt] (h3) at (4,0.5) {};  
  \coordinate[inner sep=0pt] (a1') at (-3,1.4) {};
  \coordinate[inner sep=0pt] (a2') at (-3,-1.6) {};
  \coordinate[inner sep=0pt] (z') at (-1,0.4) {};  
  \coordinate[inner sep=0pt] (x1') at (0,-0.1) {};
  \coordinate[inner sep=0pt] (x1'') at (0,-0.07) {};
  \coordinate[inner sep=0pt] (y') at (1,-0.07) {};  

  \node[left] at (a1) {$a_1$};
  \node[left] at (a2) {$a_2$};
  \node[right] at (b1) {$b_1$};
  \node[right] at (b2) {$b_2$};
  \node[above] at (z) {$z$};
  \node[below right, xshift=-2pt] at (y) {$y$};
  \node[below right, xshift=-2pt] at (x1) {$x_1$};
  \node[below, yshift=-2pt] at (xr) {$x_r$};

  \draw[line width=0.7pt] (a1) to node [pos=0.3, below, yshift=-2pt] {$A_1$} (x1) to (xr) to (b1);
  \draw[line width=0.7pt] (a2) to (x1);
  \draw[line width=0.7pt] (xr) to (b2);
  \draw[line width=1.4pt, blue] (a1) to (h0); 
  \draw[line width=1.4pt, blue, dashed] (h0) to [bend left=30] node[midway, above] {$Q_1$} (y); 
  \draw[line width=1.4pt, blue] (y) to (h1); 
  \draw[line width=1.4pt, blue, dashed] (h1) to [bend left=30] (h3); 
  \draw[line width=1.4pt, blue] (h3) to (b1);

  \draw[line width=1.4pt, green] (a2) to (x1) to (z);
  \draw[line width=1.4pt, green, dashed] (z) to [bend right=80] node[pos=0.6, below] {$Q_2$} (h2); 
  \draw[line width=1.4pt, green]  (h2) to (xr) to (b2);
  
  \draw[line width=1.0pt, loopcolor] (a1') to (z');
  \draw[line width=1.0pt, loopcolor] (a2') to (x1');
  \draw[line width=1.0pt, loopcolor] (x1'') to (y');

  \node[draw, circle, fill=black, inner sep=1.9pt] (a1) at (-3, 1.5) {};
  \node[draw, circle, fill=black, inner sep=1.9pt] (z)  at (-1, 0.5) {};
  \node[draw, circle, fill=black, inner sep=1.9pt] (a2) at (-3, -1.5) {};
  \node[draw, circle, fill=black, inner sep=1.9pt] (x1) at (0, 0) {};
  \node[draw, circle, fill=black, inner sep=1.9pt] (y)  at (1, 0) {};
  \node[draw, circle, fill=black, inner sep=1.9pt] (xr) at (3, 0) {};
  \node[draw, circle, fill=black, inner sep=1.9pt] (b1) at (6, 1.5) {};
  \node[draw, circle, fill=black, inner sep=1.9pt] (b2) at (6, -1.5) {};

\end{tikzpicture}
  \hspace{-5pt}
  \caption{Illustration for the proof of Claim~\ref{clm:Xmonotone}. 
  Edges within~$T$ are depicted using solid lines, edges not in~$T$ using dashed lines. 
  Paths $Q_1$ and~$Q_2$ are shown in \textcolor{blue}{\textbf{blue}} and in \textcolor{green}{\textbf{green}}, respectively. We highlighted the closed walk $W$ in \textcolor{coralred}{\textbf{coral red}}. 
  }
  \label{fig:starting}
\end{figure}
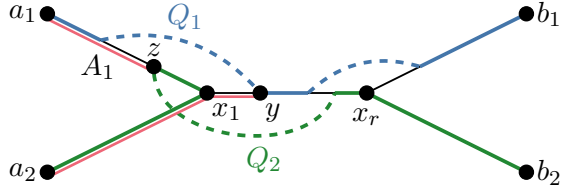

Since $z \in V(Q_2)$ and $z$ is not on~$T[a_2,b_2]$, using that $Q_1$ and~$Q_2$ are locally cheapest, 
we know that there exists a vertex of~$Q_1$ on~$T[a_2,b_2]$. 
Let $y$ be the vertex of~$Q_1$ on~$T[a_2,b_2]$ that is closest to~$a_2$; then $V(T[a_2,y]) \cap V(Q_1)=\{y\}$.
As before, for each vertex $y'$ on~$T[a_2,y]$ that is on~$Q_2$, we must have $T[a_2,y'] \subseteq Q_2$, because $Q_1$ and~$Q_2$ are locally cheapest.

These facts imply that $Q_1[y,b_1]$ shares no vertices with~$T[a_1,z]$, and similarly, $Q_2[z,b_2]$ shares no vertices with~$T[a_2,y]$. 
Define paths 
\begin{align*}
    S_1 &=T[a_1,z] \cup Q_2[z,b_2], \\
    S_2 &=T[a_2,y] \cup Q_1[y,b_1]. 
\end{align*}
Observe that $S_1$ and $S_2$ are permissively disjoint $(\{a_1,a_2\},\{b_1,b_2\})$-paths due to the permissive disjointness of $Q_1$ and $Q_2$, and 
our observation above. 

Using statement~(2) of Lemma~\ref{lem:T-min-pathlength} for paths~$Q_1[a_1,y]$ and~$Q_2[a_2,z]$, we obtain
\begin{align*}
 w(S_1)+w(S_2) &= w(T[a_1,z])+ w(Q_2[z,b_2])+w(T[a_2,y])+ w(Q_1[y,b_1])  \\
  & < w(Q_1[a_1,y])+w(Q_2[a_2,z]) + w(Q_2[z,b_2])+ w(Q_1[y,b_1])  \\
  & = w(Q_1)+w(Q_2)
\end{align*}
where the inequality follows from the fact that $T[a_1,y]$ and $T[a_2,z]$ properly intersect each other (since both contain~$x_1$).
This contradicts our definition of $Q_1$ and~$Q_2$. 
Thus, for each $i \in [2]$, we have proved that $Q_1$ and~$Q_2$ cannot both contain vertices of $V(A_i) \setminus \{x_1\}$. 

The second claim of the lemma follows by symmetry (i.e., by 
switching the roles of $A_i$ and~$B_i$ as well as that of~$x_1$ and~$x_r$).
\end{proof}

Using Lemmas~\ref{lem:closed-walk},~\ref{lem:T-min-pathlength} and~\ref{lem:starting}, we can establish the following properties of
a non-separable minimum-weight solution.
\begin{definition}[\bf $X$-monotone path]
A path~$Q$ starting at $a_1$ or $a_2$ is \emph{$X$-monotone} 
if for any vertices $u_1$ and $u_2$ on~$Q$ such that $u_1 \in V(T_{j_1})$ and $u_2 \in V(T_{j_2})$ for some $j_1<j_2$ 
it holds that $u_1$ precedes $u_2$ on~$Q$.
\end{definition}

\begin{definition}[\bf Plain path]
\label{def:plain}
A path~$Q$ is \emph{plain}, if whenever $Q$ contains some $x_i \in V(X)$, 
then the vertices of~$Q$ in~$T_i$ induce a path in~$T_i$. 
In other words, if vertices $x_j \in V(X)$ and $u \in V(T_j)$  
both appear on~$Q$, then $T[u,x_j] \subseteq Q$.
\end{definition}

\begin{lemma}
\label{lem:plain+X-monotone}
Let $Q_1$ and $Q_2$ be two permissively disjoint $(\{a_1,a_2\},\{b_1,b_2\})$-paths in~$G$ with minimum total weight.
Then both~$Q_1$ and~$Q_2$ are $X$-monotone and plain.
\end{lemma}

\begin{proof}
W.l.o.g.\ we suppose $a_i,b_i \in V(Q_i)$ for $i \in [2]$.
First we prove that $Q_1$ and~$Q_2$ are $X$-monotone, and then prove their plainness.
\begin{claim}
\label{clm:Xmonotone}
Both $Q_1$ and~$Q_2$ are $X$-monotone.
\end{claim}
\begin{claimproof}
For vertices $u_1 \in V(T_{j_1})$ and $u_2 \in V(T_{j_2})$ for some indices $j_1  < j_2$, we  
say that $(u_1,u_2)$ is a  \emph{reversed pair}, if $u_1$ and $u_2$ both appear on~$Q_i$  for some $i \in [2]$, but $u_2$ precedes $u_1$ on~$Q_i$. The statement of the lemma is that no reversed pair exists for $Q_1$ or $Q_2$. 
Assume the contrary, and choose $u_1$ and $u_2$ so that 
$u_1$ is as close to~$x_1$ as possible (i.e., it minimizes $|T[u_1,x_1]|$),
and subject to that, $u_2$ is as close to~$x_r$ as possible. By symmetry, we may assume that $u_1$ and $u_2$ are both on~$Q_1$. 
See Figure~\ref{fig:monotone} for an illustration.

First note that $T[a_1,u_1] \not\subseteq Q_1$ (because $u_2$ precedes~$u_1$), and therefore
$T[a_1,u_1]$ must contain a vertex of $Q_2$, since $Q_1$ and $Q_2$ are locally cheapest.
Let $v$ be a vertex on $T[a_1,u_1]$ closest to~$u_1$ among those contained in~$V(Q_2)$.  
Similarly, note that $T[u_2,b_1] \not\subseteq Q_2$, and therefore 
$T[u_2,b_1]$ must contain a vertex of $Q_2$, since $Q_1$ and $Q_2$ are locally cheapest.
Let $v'$ be a vertex on $T[u_2,b_1]$ closest to~$u_2$ among those contained in~$V(Q_1)$.  

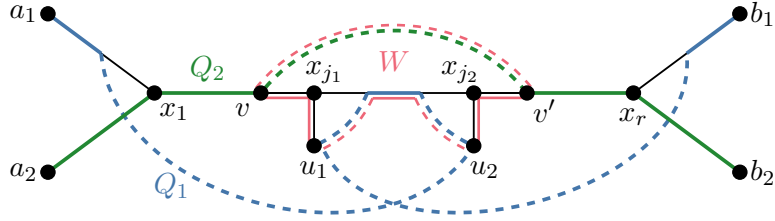
\begin{figure}
  \centering
    \begin{tikzpicture}[xscale=0.7, yscale=0.7]

  \node[draw, circle, fill=black, inner sep=1.7pt] (a1) at (-2, 1.5) {};
  \node[draw, circle, fill=black, inner sep=1.7pt] (a2) at (-2, -1.5) {};
  \node[draw, circle, fill=black, inner sep=1.7pt] (x1) at (0, 0) {};
  \node[draw, circle, fill=black, inner sep=1.7pt] (v)  at (2, 0) {};
  \node[draw, circle, fill=black, inner sep=1.7pt] (xj1) at (3, 0) {};
  \node[draw, circle, fill=black, inner sep=1.7pt] (u1) at (3, -1) {};
  \node[draw, circle, fill=black, inner sep=1.7pt] (xj2) at (6, 0) {};  
  \node[draw, circle, fill=black, inner sep=1.7pt] (u2) at (6, -1) {};
  \node[draw, circle, fill=black, inner sep=1.7pt] (v') at (7, 0) {};
  \node[draw, circle, fill=black, inner sep=1.7pt] (xr) at (9, 0) {};
  \node[draw, circle, fill=black, inner sep=1.7pt] (b1) at (11, 1.5) {};
  \node[draw, circle, fill=black, inner sep=1.7pt] (b2) at (11, -1.5) {};
  \coordinate[inner sep=0pt] (h0) at (-1,0.75) {};
  \coordinate[inner sep=0pt] (h1) at (4,0) {};
  \coordinate[inner sep=0pt] (h2) at (5,0) {};
  \coordinate[inner sep=0pt] (h3) at (10,0.75) {};  
  \coordinate[inner sep=0pt] (h1') at (4.1,-0.1) {};  
  \coordinate[inner sep=0pt] (h2') at (4.9,-0.1) {};    

  \node[left] at (a1) {$a_1$};
  \node[left] at (a2) {$a_2$};
  \node[right] at (b1) {$b_1$};
  \node[right] at (b2) {$b_2$};
  \node[below, yshift=-2pt] at (u1) {$u_1$};
  \node[below, yshift=-2pt, xshift=4pt] at (u2) {$u_2$};  
  \node[below left] at (v) {$v$};
  \node[below right, xshift=-1.5pt, yshift=1.5pt] at (v') {$v'$};  
  \node[above, xshift=4pt] at (xj1) {$x_{j_1}$};    
  \node[above left, xshift=6pt] at (xj2) {$x_{j_2}$};    
  \node[below right, xshift=-2pt] at (x1) {$x_1$};
  \node[below, yshift=-2pt] at (xr) {$x_r$};

  \draw[line width=0.7pt] (a1) to  (x1) to (xr) to (b1);
  \draw[line width=0.7pt] (a2) to (x1) to (xj1) to (u1);
  \draw[line width=0.7pt] (u2) to (xj2) to (xr) to (b2);
  
  \draw[line width=1.4pt, blue] (a1) to (h0); 
  \draw[line width=1.4pt, blue, dashed] (h0) to [bend right=70] node[pos=0.4, left, xshift=-4pt] {$Q_1$} (u2);
  \draw[line width=1.4pt, blue, dashed] (u2) to [bend left=20] (h2); 
  \draw[line width=1.4pt, blue] (h2) to (h1); 
  \draw[line width=1.4pt, blue, dashed] (h1) to [bend left=20] (u1) to [bend right=82] (h3); 
  \draw[line width=1.4pt, blue] (h3) to (b1);

  \draw[line width=1.4pt, green] (a2) to (x1) to node[pos=0.5, above] {$Q_2$}(v);
  \draw[line width=1.4pt, green, dashed] (v) to [bend left=50]  (v'); 
  \draw[line width=1.4pt, green]  (v') to (xr) to (b2);
  
  \draw[line width=1.0pt, loopcolor] (u1.north west) to (xj1.south west);
  \draw[line width=1.0pt, loopcolor] (xj1.south west) to (v.south east);
  \draw[line width=1.0pt, loopcolor, dashed] (v.north west) to [bend left=52] (v'.north east);
  \draw[line width=1.0pt, loopcolor, dashed] (h1') to [bend left=20] (u1.south east);
  \draw[line width=1.0pt, loopcolor] (h1') to node[pos=0.5, above, yshift=6pt] {$W$}  (h2');
  \draw[line width=1.0pt, loopcolor, dashed] (h2') to [bend right=20] (u2.south west);
  \draw[line width=1.0pt, loopcolor] (u2.north east) to (xj2.south east);
  \draw[line width=1.0pt, loopcolor] (xj2.south east) to (v'.south west);  

  \node[draw, circle, fill=black, inner sep=1.9pt] (a1) at (-2, 1.5) {};
  \node[draw, circle, fill=black, inner sep=1.9pt] (a2) at (-2, -1.5) {};
  \node[draw, circle, fill=black, inner sep=1.9pt] (x1) at (0, 0) {};
  \node[draw, circle, fill=black, inner sep=1.9pt] (v)  at (2, 0) {};
  \node[draw, circle, fill=black, inner sep=1.9pt] (xj1) at (3, 0) {};
  \node[draw, circle, fill=black, inner sep=1.9pt] (u1) at (3, -1) {};
  \node[draw, circle, fill=black, inner sep=1.9pt] (xj2) at (6, 0) {};  
  \node[draw, circle, fill=black, inner sep=1.9pt] (u2) at (6, -1) {};
  \node[draw, circle, fill=black, inner sep=1.9pt] (v') at (7, 0) {};
  \node[draw, circle, fill=black, inner sep=1.9pt] (xr) at (9, 0) {};
  \node[draw, circle, fill=black, inner sep=1.9pt] (b1) at (11, 1.5) {};
  \node[draw, circle, fill=black, inner sep=1.9pt] (b2) at (11, -1.5) {};

\end{tikzpicture}
  \hspace{-5pt}
  \caption{Illustration for the proof of Lemma~\ref{lem:starting}. 
  Edges within~$T$ are depicted using solid lines, edges not in~$T$ using dashed lines. 
  Paths $Q_1$ and~$Q_2$ are shown in \textcolor{blue}{\textbf{blue}} and in \textcolor{green}{\textbf{green}}, respectively. We highlighted paths $T[a_1,z]$ and $T[a_2,y]$ in \textcolor{coralred}{\textbf{coral red}}. 
  }
  \label{fig:monotone}
\end{figure}

By Lemma~\ref{lem:starting} we know that $v$ lies on~$T[x_1,u_1]$, and $v'$ lies on~$T[u_2,x_r]$.
Note that no inner vertex of~$T[v,u_1]$ is contained in $V(Q_2)$ by the definition of~$v$. 
Let $e_1$ denote the edge $T[x_1,u_1]$ incident to~$u_1$, and let $u'_1$ be the other endpoint of~$e_1$. 
Observe that $e_1 \in Q_1$ is not possible, because then $(u'_1,u_2)$ would also be a reversed pair, which would contradict our choice of~$(u_1,u_2)$
(since $u'_1$ is closer to~$x_1$ than~$u_1$). 
Using that $Q_1$ and $Q_2$ are locally cheapest, we have 
that no inner vertex of $T[v,u_1]$ is contained in $V(Q_1 \cup Q_2)$.
Reasoning analogously, it also follows that no inner vertex of $T[u_2,v']$ is contained in $V(Q_1 \cup Q_2)$.

Note that $v \in V(T_j)$ for some $j\leq j_1$, and $v' \in V(T_{j'})$ for some $j' \geq j_2$, implying $j<j'$. 
Moreover, $v$ is closer to~$x_1$ than~$u_1$. Hence, by our choice of $(u_1,u_2)$, the pair $(v,v')$ is not a reversed pair. 
This implies that $v$ is on $Q_2[a_2,v']$.
In particular, $Q_2[a_2,v]$ and $Q_2[v',b_2]$ are disjoint (because $v \neq v'$).
Define paths 
\begin{align*}
    S_1 &=Q_1[a_1,u_2] \cup T[u_2,v'] \cup Q_2[v',b_2], \\
    S_2 &=Q_2[a_2,v] \cup T[v,u_1] \cup Q_1[u_1,b_1]. 
\end{align*}
Observe that $S_1$ and $S_2$ are permissively disjoint $(\{a_1,a_2\},\{b_1,b_2\})$-paths due to the permissive disjointness of $Q_1$ and $Q_2$ 
and our observations on $T[v,u_1]$ and $T[u_2,v']$. 
Moreover, the walk 
\[ W=T[v,u_1] \cup Q_1[u_1,u_2] \cup T[u_2,v'] \cup Q_2[v',v]
\]
is closed and does not contain any edge of~$E^-$ more than once, since $Q_1$ and $Q_2$ share no edges, and no edge of $T[v,u_1] \cup T[u_2,v']$ 
is contained in~$Q_1 \cup Q_2$. Thus, $w(W)\geq 0$ by Lemma~\ref{lem:closed-walk}. This implies
\begin{align*}
\!\!w(S_1)+w(S_2) &= w(Q_1)+w(Q_2)-w(W \setminus (T[v,u_1] \cup T[u_2,v'])) + w(T[v,u_1] \cup T[u_2,v']) \\
&< w(Q_1)+w(Q_2),
\end{align*} 
because $w(T[v,u_1] \cup T[u_2,v']) <0$. This contradicts our definition of $Q_1$ and $Q_2$.
Hence, we have proved that no reversed pair exists, and thus both~$Q_1$ and~$Q_2$ are $X$-monotone.
\end{claimproof}

\begin{claim}
\label{clm:plain}
Both $Q_1$ and~$Q_2$ are plain.
\end{claim}
\begin{claimproof}
Suppose for contradiction that $x_j \in V(Q_1)$ and $u \in V(Q_1) \cap V(T_j)$, but $T[u,x_j]  \not \subseteq Q_1$;
the case when $u$ and $x_j$ are on~$Q_2$ is symmetric. 
Clearly, $u \neq x_j$.
Since $Q_1$ and $Q_2$ are locally cheapest, by $T[u,x_j] \not \subseteq Q_1$ we know that
$T[u,x_j]$ must contain a vertex of~$Q_2$. 

Let $z$ be the closest vertex on $T[u,x_j]$ to~$u$ that appears on~$Q_2$, and let $z'$ be the closest vertex on~$T[u,z]$ to~$z$  that appears on~$Q_1$. Then no inner vertex of $T[z,z']$ contained in $V(Q_1 \cup Q_2)$. 
Since $Q_1$ and $Q_2$ are locally cheapest, it follows from the definition of~$z$ and $z'$ that $T[u,z'] \subseteq Q_1$. 

We now distinguish between two cases; see Figure~\ref{fig:plain} for an illustration. 

\begin{figure}
  \centering
  \begin{subfigure}[b]{0.5\textwidth}
    \centering  
    \begin{tikzpicture}[xscale=0.7, yscale=0.7]

  \node[] (A) at (-0.9,2) {Case A:};
  \node[draw, circle, fill=black, inner sep=1.7pt] (a1) at (-1, 1) {};
  \node[draw, circle, fill=black, inner sep=1.7pt] (a2) at (-1, -1) {};
  \node[draw, circle, fill=black, inner sep=1.7pt] (x1) at (0, 0) {};
  \node[draw, circle, fill=black, inner sep=1.7pt] (xh)  at (2, 0) {};
  \node[draw, circle, fill=black, inner sep=1.7pt] (xj) at (4, 0) {};
  \node[draw, circle, fill=black, inner sep=1.7pt] (z) at (4, -1.5) {};  
  \node[draw, circle, fill=black, inner sep=1.7pt] (z')  at (4, -2.25) {};    
  \node[draw, circle, fill=black, inner sep=1.7pt] (u)  at (4, -3) {};      
  \node[draw, circle, fill=black, inner sep=1.7pt] (xr) at (6, 0) {};
  \node[draw, circle, fill=black, inner sep=1.7pt] (b1) at (7, 1) {};
  \node[draw, circle, fill=black, inner sep=1.7pt] (b2) at (7, -1) {};
  \coordinate[inner sep=0pt] (h0) at (-0.5,0.5) {};
  \coordinate[inner sep=0pt] (h1) at (6.5,-0.5) {};  
  \coordinate[inner sep=0pt] (h2) at (4,-0.75) {};    


  \node[left] at (a1) {$a_1$};
  \node[left] at (a2) {$a_2$};
  \node[right] at (b1) {$b_1$};
  \node[right] at (b2) {$b_2$};
  \node[right] at (u) {$u$};
  \node[right, yshift=3pt] at (z) {$z$};  
  \node[left] at (z') {$z'$};    
  \node[above, yshift=1.7pt, xshift=4pt] at (xh) {$x_h$};
  \node[above, xshift=4pt] at (xj) {$x_j$};    
  \node[below right, xshift=-2pt] at (x1) {$x_1$};
  \node[above, yshift=1.7pt, xshift=-6pt] at (xr) {$x_r$};

  \draw[line width=0.7pt] (a1) to  (x1) to (xr) to (b1);
  \draw[line width=0.7pt] (xh) to (xj) to (z');
  \draw[line width=0.7pt] (xr) to (b2);
  
  \draw[line width=1.4pt, blue] (a1) to (h0); 
  \draw[line width=1.4pt, blue, dashed] (h0) to [bend right=50] node[pos=0.5, above right, xshift=0pt] {$Q_1$} (u);
  \draw[line width=1.4pt, blue] (u) to (z'); 
  \draw[line width=1.4pt, blue, dashed] (z'.east) to [bend right=70] (xj.east); 
  \draw[line width=1.4pt, blue] (xj) to (xr) to (b1);

  \draw[line width=1.4pt, green] (a2) to (x1) to node[pos=0.5, above] {$Q_2$}(xh);
  \draw[line width=1.4pt, green, dashed] (xh.south) to [bend right=40]  (z); 
  \draw[line width=1.4pt, green]  (z) to (h2);
  \draw[line width=1.4pt, green, dashed] (h2) to [bend right=10]  (h1);   
  \draw[line width=1.4pt, green]  (h1) to (b2);
  
  \draw[line width=1.0pt, loopcolor] (xh.south east) to (xj.south west);
  \draw[line width=1.0pt, loopcolor, dashed] (z'.north east) to  [bend right=70] (xj.south east);
  \draw[line width=1.0pt, loopcolor] (z'.north east) to (z.south east);
  \draw[line width=1.4pt, loopcolor, dashed] (xh.south east) to [bend right=40] node[pos=0.7,above,yshift=4pt,xshift=4pt] {$W$} ( z.north west);   

  \node[draw, circle, fill=black, inner sep=1.9pt] (a1) at (-1, 1) {};
  \node[draw, circle, fill=black, inner sep=1.9pt] (a2) at (-1, -1) {};
  \node[draw, circle, fill=black, inner sep=1.9pt] (x1) at (0, 0) {};
  \node[draw, circle, fill=black, inner sep=1.9pt] (xh)  at (2, 0) {};
  \node[draw, circle, fill=black, inner sep=1.9pt] (xj) at (4, 0) {};
  \node[draw, circle, fill=black, inner sep=1.9pt] (z) at (4, -1.5) {};  
  \node[draw, circle, fill=black, inner sep=1.9pt] (z')  at (4, -2.25) {};    
  \node[draw, circle, fill=black, inner sep=1.9pt] (u)  at (4, -3) {};      
  \node[draw, circle, fill=black, inner sep=1.9pt] (xr) at (6, 0) {};
  \node[draw, circle, fill=black, inner sep=1.9pt] (b1) at (7, 1) {};
  \node[draw, circle, fill=black, inner sep=1.9pt] (b2) at (7, -1) {};

\end{tikzpicture}
  \end{subfigure} 
  \hspace{-5pt}
  \begin{subfigure}[b]{0.49\textwidth}
    \centering  
    \begin{tikzpicture}[xscale=0.7, yscale=0.7]

  \node[] (B) at (-0.9,2) {Case B:};
  \node[draw, circle, fill=black, inner sep=1.7pt] (a1) at (-1, 1) {};
  \node[draw, circle, fill=black, inner sep=1.7pt] (a2) at (-1, -1) {};
  \node[draw, circle, fill=black, inner sep=1.7pt] (x1) at (0, 0) {};
  \node[draw, circle, fill=black, inner sep=1.7pt] (xh)  at (4, 0) {};
  \node[draw, circle, fill=black, inner sep=1.7pt] (xj) at (2, 0) {};
  \node[draw, circle, fill=black, inner sep=1.7pt] (z) at (2, -1.5) {};  
  \node[draw, circle, fill=black, inner sep=1.7pt] (z')  at (2, -2.25) {};    
  \node[draw, circle, fill=black, inner sep=1.7pt] (u)  at (2, -3) {};      
  \node[draw, circle, fill=black, inner sep=1.7pt] (xr) at (6, 0) {};
  \node[draw, circle, fill=black, inner sep=1.7pt] (b1) at (7, 1) {};
  \node[draw, circle, fill=black, inner sep=1.7pt] (b2) at (7, -1) {};
  \coordinate[inner sep=0pt] (h0) at (-0.5,-0.5) {};
  \coordinate[inner sep=0pt] (h1) at (6.5,0.5) {};  
  \coordinate[inner sep=0pt] (h2) at (2,-0.75) {};    


  \node[left] at (a1) {$a_1$};
  \node[left] at (a2) {$a_2$};
  \node[right] at (b1) {$b_1$};
  \node[right] at (b2) {$b_2$};
  \node[below, yshift=-3pt] at (u) {$u$};
  \node[right, yshift=-3pt] at (z) {$z$};  
  \node[below right, yshift=4pt] at (z') {$z'$};    
  \node[above, yshift=1.7pt, xshift=4pt] at (xh) {$x_h$};
  \node[above, xshift=4pt] at (xj) {$x_j$};    
  \node[below right, xshift=-2pt] at (x1) {$x_1$};
  \node[above, yshift=1.7pt, xshift=-6pt] at (xr) {$x_r$};

  \draw[line width=0.7pt] (a2) to  (x1) to (xr) to (b1);
  \draw[line width=0.7pt] (xh) to (xj) to (z');
  \draw[line width=0.7pt] (xr) to (b2);

  \draw[line width=1.4pt, blue] (a1) to (x1) to node[pos=0.5, above] {$Q_1$}(xj);
  \draw[line width=1.4pt, blue, dashed] (xj) to [bend right=70] (u);
  \draw[line width=1.4pt, blue] (u) to (z'); 
  \draw[line width=1.4pt, blue, dashed] (z') to [in=-70, out=20] (h1); 
  \draw[line width=1.4pt, blue] (h1) to (b1);

  \draw[line width=1.4pt, green] (a2) to (h0); 
  \draw[line width=1.4pt, green, dashed] (h0) to [bend right=40] node [pos=0.3, below] {$Q_2$} (h2); 
  \draw[line width=1.4pt, green]  (h2) to (z);
  \draw[line width=1.4pt, green, dashed] (z.east) to [bend right=10]  (xh.south);   
  \draw[line width=1.4pt, green]  (xh) to (xr) to (b2);

  \draw[line width=1.4pt, loopcolor, dashed] (xj.south) to [bend right=76] (u.north);  
  \draw[line width=1.1pt, loopcolor] (xj.south east) to (xh.south west);
  \draw[line width=1.4pt, loopcolor, dashed] (z.north east) to [bend right=10] node[pos=0.5,above left, xshift=2pt, yshift=-2pt] {$W$} (xh);   
  \draw[line width=1.4pt, loopcolor] (z.south west) to (u.north west);  

  \node[draw, circle, fill=black, inner sep=1.7pt] (a1) at (-1, 1) {};
  \node[draw, circle, fill=black, inner sep=1.7pt] (a2) at (-1, -1) {};
  \node[draw, circle, fill=black, inner sep=1.7pt] (x1) at (0, 0) {};
  \node[draw, circle, fill=black, inner sep=1.7pt] (xh)  at (4, 0) {};
  \node[draw, circle, fill=black, inner sep=1.7pt] (xj) at (2, 0) {};
  \node[draw, circle, fill=black, inner sep=1.7pt] (z) at (2, -1.5) {};  
  \node[draw, circle, fill=black, inner sep=1.7pt] (z')  at (2, -2.25) {};    
  \node[draw, circle, fill=black, inner sep=1.7pt] (u)  at (2, -3) {};      
  \node[draw, circle, fill=black, inner sep=1.7pt] (xr) at (6, 0) {};
  \node[draw, circle, fill=black, inner sep=1.7pt] (b1) at (7, 1) {};
  \node[draw, circle, fill=black, inner sep=1.7pt] (b2) at (7, -1) {};

\end{tikzpicture}  
  \end{subfigure}
  \caption{Illustration for the proof of Claim~\ref{clm:plain}. 
  Edges within~$T$ are depicted using solid lines, edges not in~$T$ using dashed lines. 
  Paths $Q_1$ and~$Q_2$ are shown in \textcolor{blue}{\textbf{blue}} and in \textcolor{green}{\textbf{green}}, respectively. We highlighted the closed walk $W$ in \textcolor{coralred}{\textbf{coral red}}. 
  }
  \label{fig:plain}
\end{figure}
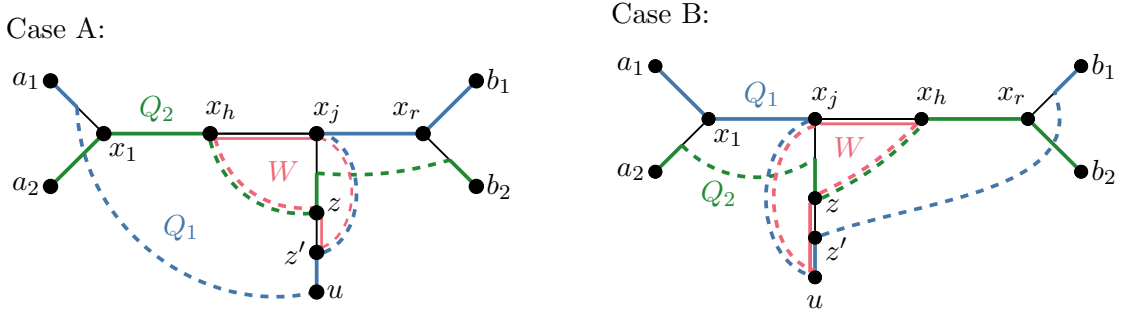

{\bf Case A.} First assume that $u$ precedes $x_j$ on~$Q_1$.
By $T[u,z'] \subseteq Q_1$ this implies that $z'$ also precedes $x_j$ on~$Q_1$. 

Using again that $Q_1$ and $Q_2$ are locally cheapest, by $T[a_1,x_j] \notin Q_1$ we know that $[Tx_1,x_j]$ must contain a vertex of~$Q_2$. Let $x_h$ be the vertex on~$T[x_1,x_j]$ closest to~$x_j$. We show that neither $Q_1$ nor $Q_2$ contains an inner vertex of $T[x_h,x_j]$. For~$Q_2$, this follows from the definition of~$x_h$. For $Q_1$, this follows from  Claim~\ref{clm:Xmonotone} 
since $u$ precedes~$x_j$ on~$Q_1$, the $X$-monotonicity of~$Q_1$ implies that $Q_1$ enters $x_j$ through an edge with both endpoints in~$T_j$, 
so in particular, $x_j x_{j-1} \notin Q_1$; using that $Q_1$ and $Q_2$ are locally cheapest, it follows that $Q_1$ contains no inner vertices of~$T[x_h,x_j]$.
By the $X$-monotonicity of~$Q_2$, we also know that $z$ follows $x_h$ on~$Q_2$.

Define the paths
\begin{align*}
S_1 &= Q_1[a_1,z'] \cup T[z',z] \cup Q_2[z,b_2], \\
S_2 &= Q_2[a_2,x_h] \cup T[x_h,x_j] \cup Q_1[x_j,b_1].
\end{align*}
Observe that $S_1$ and $S_2$ are permissively disjoint $(\{a_1,a_2\},\{b_1,b_2\})$-paths due to the permissive disjointness of $Q_1$ and $Q_2$, and our observations on $T[z,z']$ and $T[x_h,x_j]$. 
Moreover, the walk 
\[ W=Q_1[z',x_j] \cup T[x_j,x_h] \cup Q_2[x_h,z] \cup T[z,z']
\]
is closed and does not contain any edge of~$E^-$ more than once, since $Q_1$ and $Q_2$ share no edges, and no edge of $T[z,z'] \cup T[x_j,x_h]$ 
is contained in~$Q_1 \cup Q_2$. Thus, $w(W)\geq 0$ by Lemma~\ref{lem:closed-walk}. This implies
\begin{align*}
\!\!w(S_1)+w(S_2) &= w(Q_1)+w(Q_2)-w(W \setminus (T[z,z'] \cup T[x_j,x_h])) + w(T[z,z'] \cup T[x_j,x_h]) \\
&< w(Q_1)+w(Q_2),
\end{align*} 
because $w(T[z,z'] \cup T[x_j,x_h]) <0$. This contradicts our definition of $Q_1$ and $Q_2$.

{\bf Case B.} 
Assume now that $u$ follows $x_j$ on~$Q_1$; our reasoning is very similar to the previous case.
By $T[u,z'] \subseteq Q_1$ this implies that $z'$ also follows $x_j$ on~$Q_1$. 

Using again that $Q_1$ and $Q_2$ are locally cheapest, by $T[x_j,b_1] \notin Q_1$ we know that $T[x_j,x_r]$ must contain a vertex of~$Q_2$. Let $x_h$ be the vertex on~$T[x_j,x_r]$ closest to~$x_j$. We show that neither $Q_1$ nor $Q_2$ contains an inner vertex of $T[x_j,x_h]$. For~$Q_2$, this follows from the definition of~$x_h$. For $Q_1$, this follows from Lemma~\ref{lem:plain+X-monotone} (implying $x_j x_{j+1} \notin Q_1$) 
and the fact that $Q_1$ and $Q_2$ are locally cheapest.
Note also that by the $X$-monotonicity of~$Q_2$, we know that $z$ precedes $x_h$ on~$Q_2$.

Define the paths
\begin{align*}
S_1 &= Q_1[a_1,x_j] \cup T[x_j,x_h] \cup Q_2[x_h,b_2], \\
S_2 &= Q_2[a_2,z] \cup T[z,z'] \cup Q_1[z,b_1].
\end{align*}
Observe that $S_1$ and $S_2$ are permissively disjoint $(\{a_1,a_2\},\{b_1,b_2\})$-paths due to the permissive disjointness of $Q_1$ and $Q_2$, and our observations on $T[z',z]$ and $T[x_j,x_h]$. 
Moreover, the walk 
\[ W=Q_1[x_j,z'] \cup T[z',z] \cup Q_2[z,x_h] \cup T[x_h,x_j] 
\]
is closed and does not contain any edge of~$E^-$ more than once, since $Q_1$ and $Q_2$ share no edges, and no edge of $T[z',z] \cup T[x_h,x_j]$ 
is contained in~$Q_1 \cup Q_2$. Thus, $w(W)\geq 0$ by Lemma~\ref{lem:closed-walk}. This implies
\begin{align*}
\!\! w(S_1)+w(S_2) &= w(Q_1)+w(Q_2)-w(W \setminus (T[z',z] \cup T[x_h,x_j])) + w(T[z',z] \cup T[x_h,x_j]) \\
&< w(Q_1)+w(Q_2),
\end{align*} 
because $w(T[z,z'] \cup T[x_j,x_h]) <0$. This contradicts our definition of $Q_1$ and $Q_2$.
\end{claimproof}
Having proved Claims~\ref{clm:Xmonotone} and~\ref{clm:plain}, 
Lemma~\ref{lem:plain+X-monotone} follows.
\end{proof}

The following observation summarizes our understanding on how an optimal solution uses paths $A_1$, $A_2$, $B_1$, and~$B_2$.

\begin{lemma}
\label{lem:contains-A1orA2}
If $Q_1$ and $Q_2$ are two permissively disjoint $(\{a_1,a_2\},\{b_1,b_2\})$-paths that are locally cheapest and also plain, then one of them contains $A_1$ or~$A_2$, and one of them contains~$B_1$ or~$B_2$.    
\end{lemma}

\begin{proof}
W.l.o.g.\ we suppose $a_i,b_i \in V(Q_i)$ for $i \in [2]$.
First, we show that both $x_1$ and $x_r$ are contained in $Q_1 \cup Q_2$. For each $i \in [2]$, let $u_i$ denote the vertex closest to~$x_1$ on~$A_i$ that is contained in~$Q_1 \cup Q_2$.
Similarly, let $v_i$ denote the vertex closest to~$x_r$ on~$B_i$ that is contained in~$Q_1 \cup Q_2$. 
If $V(X) \cap V(Q_1 \cup Q_2) =  \emptyset$, then $u_1, u_2, v_1$ and $v_2$ are four distinct vertices, with at least two of them belonging to the same path~$Q_1$ or $Q_2$, contradicting our assumption that $Q_1$ and $Q_2$ are locally cheapest. 
Consider the vertex $z$ of $V(X) \cap V(Q_1 \cup Q_2)$ closest to~$x_1$ on~$X$.
If $x_1 \notin V(Q_1 \cup Q_2)$, then $u_1, u_2$ and $z$ are three distinct vertices, so at least two of them belong to the same path, $Q_1$ or $Q_2$
which contradicts our assumption that $Q_1$ and $Q_2$ are locally cheapest. 
Hence, $x_1 \in V(Q_1 \cup Q_2)$, and an analogous argument shows $x_r \in V(Q_1 \cup Q_2)$.

Now, let $Q_i$ contain $x_1$. As $Q_i$ is plain and $a_i \in V(Q_i)$ we get that $V(A_i) \subseteq V(Q_i)$. As $Q_1$ and $Q_2$ are locally cheapest, we get that $A_i \subseteq Q_i$. The analogous argument shows that $B_i$ is contained in $Q_1$ or $Q_2$. 
\end{proof}

\subsection{Computing Partial Solutions}
\label{sec:partsol}

In this section we design a dynamic programming algorithm that computes two permissively disjoint $(\{a_1,a_2\},\{b_1,b_2\})$-paths of minimum total weight 
(we keep all definitions introduced in Section~\ref{sec:properties}, including our assumptions on vertices~$a_1,b_1,a_2,$ and $b_2$).
In Section~\ref{sec:properties} we have established that two permissively disjoint $(\{a_1,a_2\},\{b_1,b_2\})$-paths of minimum total weight 
are necessarily $X$-monotone, plain, and they form a locally cheapest pair.
A natural approach would be to require these same properties from a partial solution that we aim to compute. 
However, it turns out that the property of $X$-monotonicity is quite hard to ensure when building subpaths of a solution. 
The following relaxed version of monotonicity can be satisfied much easier, and still suffices for our purposes:

\begin{definition}[\bf Quasi-monotone path]
\label{def:quasimonotone}
A path~$P$ starting at~$a_1$ or $a_2$ is \emph{quasi-monotone}, if the following holds:
if $x_i \in V(P)$ for some $i \in [r]$, then all vertices in $\bigcup_{h \in [i-1]} V(T_h) \cap V(P)$ precede $x_i$ on~$P$,
and all vertices in $\bigcup_{h \in [r] \setminus [i]} V(T_h) \cap V(P)$ follow $x_i$ on~$P$.
\end{definition}

\begin{definition}[\bf Well-formed path pair]
Two paths $P_1$ and $P_2$ form a \emph{well-formed pair}, if they are locally cheapest, and both are plain and quasi-monotone.
\end{definition}

We are now ready to define partial solutions, the central notion that our dynamic programming algorithm relies on. 

\begin{definition}[\bf Partial solution]
\label{def:partsol}
Given vertices $u \in  V(T_i)$ and $v \in V(T_j)$ for some~$i \leq j$ and a set $\tau \subseteq \T \setminus \{T\}$,
two paths~$Q_1$ and $Q_2$ form a \emph{partial solution $(Q_1,Q_2)$ for}~$(u,v,\tau)$, if 
\begin{description}
    \item[(a)] $Q_1$ and $Q_2$ are permissively disjoint $(\{a_1,a_2\},\{u,v\})$-paths; 
    \item[(b)] $Q_1$ and $Q_2$ are a well-formed pair;
    \item[(c)] $Q_1$ ends with the subpath $T[x_i,u]$;
    \item[(d)] $V(T_{(i+1,r)}) \cap V(Q_2)\subseteq \{v\}$;
    \item[(e)] if $Q_1 \cup Q_2$ contains a vertex of some~$T' \in \T \setminus \{T\}$, then $T' \in \tau$;
    \item[(f)] there exists no tree $T' \in \T \setminus \{T\}$ such that $Q_1$ and $Q_2$ are in contact at~$T'$. 
\end{description}
\end{definition}

We will say that the vertices of $V(\T \setminus (\tau \cup \{T\}))$ are \emph{forbidden} for $(\tau,T)$; then
condition~(e) asks for $Q_1 \cup Q_2$ not to contain vertices forbidden for~$(\tau,T)$. 

Before turning our attention to the problem of computing partial solutions, let us first show 
how partial solutions enable us to find two permissively disjoint $(\{a_1,a_2\},\{b_1,b_2\})$-paths.%
\begin{lemma}
\label{lem:partsol-to-permdisj}
Paths $P_1$ and~$P_2$ are permissively disjoint $(\{a_1,a_2\},\{b_1,b_2\})$-paths of minimum weight in~$G$ if and only if
they form a partial solution for~$(b_h,b_{3-h},\T \setminus \{T\})$ of minimum weight for some \mbox{$h \in [2]$}.
\end{lemma}

\begin{proof}
First assume that $P_1$ and~$P_2$ are permissively disjoint $(\{a_1,a_2\},\{b_1,b_2\})$-paths in~$G$. 
By Lemma~\ref{lem:contains-A1orA2} we know that one of them contains~$B_1$ or~$B_2$, so define~$h$ such that $B_h \subseteq P_1 \cup P_2$.
By Observation~\ref{obs:locally-cheapest} and Lemmas\ref{lem:plain+X-monotone} we know that $P_1$ and~$P_2$ are a well-formed pair.
Note also that condition~(d) holds vacuously, since $b_1,b_2 \in V(T_r)$. 
Condition~(e) holds trivially, since there are no vertices forbidden for~$(\T \setminus \{T\},T)$. 
Therefore, if $P_1$ ends with~$B_h$, then $(P_1,P_2)$ is a partial solution for~$(b_h,b_{3-h},\T)$,
and if $P_2$ ends with~$B_h$, then $(P_2,P_1)$ is a partial solution for~$(b_h,b_{3-h},\T)$.

Since a partial solution for $(b_h,b_{3-h},\T \setminus \{T\})$ for some~$h \in [2]$ is by 
definition a pair of two permissively disjoint $(\{a_1,a_2\},\{b_1,b_2\})$-paths,
the lemma follows.
\end{proof}

We now present our approach for computing partial solutions using dynamic programming.

\paragraph*{Computing partial solutions: high-level view.}
For each $u \in  V(T_i)$ and $v \in V(T_j)$ for some $i\leq j$, and each $\tau \subseteq \T \setminus \{T\}$,
we are going to compute a partial solution for $(u,v,\tau)$ of minimum weight, 
denoted by $F(u,v,\tau)$, using dynamic programming;
if there exists no partial solution for $(u,v,\tau)$, 
we set $F(u,v,\tau)=\varnothing$. 

To apply dynamic programming, 
we fix an ordering~$\prec$ over~$V(T)$ fulfilling the condition that for each $i'<i$, $u \in V(T_i)$ and $u' \in V(T_{i'})$ we have $u' \prec u$. 
We compute the values $F(u,v,\tau)$ based on the ordering~$\prec$ in the sense that 
$F(u',v',\tau')$ is computed before $F(u,v,\tau)$ whenever $u' \prec u$.
This computation is performed by Algorithm~\ref{alg:PartSol} which determines a partial solution~$F(u,v,\tau)$
based on partial solutions already computed.

To compute $F(u,v,\tau)$ in a recursive manner, we use an observation that
either the partial solution has a fairly simple structure, 
or it strictly contains a partial solution for~$(u',v',\tau')$ for some vertices~$u'$ and~$v'$ with $u' \in V(T_{i'})$ and $i'<i$, and some set~$\tau' \subseteq \tau$. 
We can thus try all possible values for~$u', v'$ and~$\tau'$, and use the partial solution~$(Q'_1,Q'_2)$ 
we have already computed and stored in $F(u',v',\tau')$.
To obtain a partial solution for~$(u,v,\tau)$ based on~$Q'_1$ and~$Q'_2$, we append paths to~$Q'_1$ and to~$Q'_2$ so that 
they fulfill the requirements of Definition~\ref{def:partsol} -- most importantly, that 
$Q_1$ ends with~$T[x_i,u]$, that $Q_2$ ends at~$v$, and that $Q_1 \cup Q_2$ contains no vertex of~$V(\T \setminus (\tau \cup \{T\}))$. 
To this end, we create a path $P_1=Q'_2 \cup T[v',u]$ 
and a path $P_2=Q'_1 \cup R$ where
$R$ is a shortest $(u',v)$-path in a certain auxiliary graph.
Essentially, we use the tree~$T$ for getting from~$v'$ to~$u$, and 
we use the ``remainder'' of the graph for getting from~$u'$ to~$v$; 
note that we need to avoid the forbidden vertices and ensure condition~(f) as well.
The precise definition of the auxiliary subgraph of~$G$ that we use for this purpose is provided in Definition~\ref{def:auxgraph}.
If the obtained path pair $(P_1,P_2)$ is indeed a partial solution for~$(u,v,\tau)$, then we store it. 
After trying all possible values for~$u',v',$ and~$\tau'$, we select a partial solution that has minimum weight among those we computed.

\begin{definition}[\bf Auxiliary graph]
\label{def:auxgraph}
For some $T \in \T$, let $P \subseteq T$  be a path within~$T$, let $u$ and $v$ be two vertices on~$T$, and let $\tau \subseteq \T \setminus \{T\}$. 
Then the \emph{auxiliary graph} $G \langle P,u,v,\tau \rangle$ 
denotes the graph defined as
\[
G \langle P,u,v,\tau \rangle = G - \left(\bigcup\{ V(T_h): V(T_h) \cap V(P)=\emptyset\}  \cup  V(P) \setminus \{u,v\} \cup 
V(\T \setminus (\tau \cup \{T\})) \right).
\]
In other words, we obtain $G \langle P,u,v,\tau \rangle$ from~$G$ by deleting all trees~$T_h$ that do not intersect~$P$, and deleting $P$ itself as well,
while taking care not to delete $u$ or $v$, and additionally deleting all vertices forbidden for~$(\tau,T)$.
\end{definition}

Working towards a proof for the correctness of Algorithm~\ref{alg:PartSol},
we start with two simple observations. 
The first one, stated by Lemma~\ref{lem:paths-vs-X} below, essentially says that a path in a partial solution that uses a subtree~$T_h$ of~$T$ for some~$h \in [r]$
should also go through the vertex~$x_h$ whenever possible, that is, unless the other path uses~$x_h$.

\begin{lemma}
\label{lem:paths-vs-X}
Let $Q_1$ and $Q_2$ be two permissively disjoint, locally cheapest $(\{a_1,a_2\},\{x_i,v\})$-paths for some $v \in V(T_j)$ where $1 \leq i\leq j \leq r$.
Let $z \in V(T_h)$ for some $h \leq j$ such that $h<j$ or $x_h \in V(T[z,v])$. 
If $z\in V(Q_1 \cup Q_2)$, then $x_h \in V(Q_1 \cup Q_2)$.
\end{lemma}

\begin{proof}
W.l.o.g.\ we may assume $z \in V(Q_1)$; let $a$ denote the starting vertex ($a_1$ or $a_2$) of $Q_1$.
Suppose for contradiction that $x_h \notin V(Q_1 \cup Q_2)$; then either $1 \leq h<i$ or $i<h\leq j$. 
In the former case, define $P_1$, $P_2$, and $P_3$ as the three openly disjoint paths leading within~$T$ from $x_h$ to~$a$, to~$x_i$, and to~$z$, respectively. 
In the latter case, define $P_1$, $P_2$, and $P_3$ as the three openly disjoint paths leading within~$T$ from $x_h$ to~$x_i$, to~$v$, and to~$z$, respectively;
observe that such paths exist due to the condition that either $h<j$ or $x_h \in V(T[z,v])$. 
Let $z_1$, $z_2$, and $z_3$ denote the vertex closest to~$x_h$ on $T_1$, $T_2$, and $T_3$, respectively, that is contained in~$V(Q_1 \cup Q_2)$; 
note that $z_1$, $z_2$ and $z_3$ are three distinct vertices, with no vertex of~$V(Q_1 \cup Q_2)$ lying between any two of them on~$T$. 
Since at least two vertices from $\{z_1,z_2,z_3\}$ belong to the same path $Q_1$ or~$Q_2$, this contradicts the assumption that $Q_1$ and $Q_2$ are locally cheapest.
\end{proof}

As a consequence of Lemma~\ref{lem:paths-vs-X}, applied with $a_1$ taking the role of~$z$ and $x_1$ taking the role of~$x_h$, 
we get that every partial solution for some~$(u,v,\tau)$ must contain~$x_1$; using that both paths in a partial solution must be plain, 
we get the following fact.
\begin{observation}
\label{obs:contains-A1orA2}
Let $(Q_1,Q_2)$ be a partial solution for~$(u,v,\tau)$ for
vertices $u \in  V(T_i)$ and $v \in V(T_j)$ for some $i\leq j$ and for $\tau \subseteq \T \setminus \{T\}$.
Then either $Q_1$ or $Q_2$ contains $A_1$ or $A_2$.
\end{observation}

Let us now give some insight on Algorithm~\ref{alg:PartSol} that computes a minimum-weight partial solution $(Q_1,Q_2)$ for $(u,v,\tau)$, if it exists, for vertices $u \in V(T_i)$ and $v \in V(T_j)$ with $i \leq j$ and trees $\tau \subseteq \T \setminus \{T\}$.
It distinguishes between two cases based on whether $Q_2$ 
contains a vertex of~$T[x_1,x_i]$ or not.
In both cases it constructs candidates for a partial solution, and then chooses the one among these with minimum weight.

\begin{description}
\item[Case A:] $Q_2$ does not contain any vertices from~$T[x_1,x_i]$. See Figure~\ref{fig:partsolA} for an illustration.
In this case, due to Observation~\ref{obs:contains-A1orA2}, 
we know that $Q_1$ contains $A_h$ for some $h \in [2]$; let us fix this value of~$h$.
Since $Q_1$ and~$Q_2$ are locally cheapest and~$Q_1$ ends with~$T[x_i,u]$, we also know that $Q_1$ must contain~$T[x_1,x_i]$.
Therefore, we obtain $Q_1=A_h \cup T[x_1,u]$. 
In this case, we can also prove that $Q_2$ is a shortest $(a_{3-h},v)$-path
in the auxiliary graph~$G \langle A_h \cup T[x_1,u],a_{3-h},v,\tau \rangle$.
Hence, Algorithm~\ref{alg:PartSol} computes such a path~$R$ and 
constructs the pair~$(A_h \cup T[x_1,u],R)$ as a candidate for a partial solution for~$(u,v,\tau)$.
\item[Case B:] $Q_2$ contains a vertex from~$T[x_1,x_i]$. See Figure~\ref{fig:partsolB} for an illustration.
In this case, 
let $x_{i'}$ be the vertex on $T[x_1,x_{i-1}]$ closest to~$x_i$ that appears on~$Q_2$, 
and let $u'$ be the last vertex of $Q_2$ in~$T_{i'}$;
since $Q_2$ is plain, we know $T[x_{i'},u'] \subseteq Q_2$.
Let $x_{j'}$ denote the vertex on~$T[x_{i'},x_i]$ closest to $x_{i'}$ that appears on~$Q_1$; then $i' <j' \leq i$.
As $Q_1$ and $Q_2$ are locally cheapest, $T[x_{j'},x_i] \subseteq Q_1$ follows. 
Let $v'$ denote the first vertex of~$Q_1$ in~$T_{j'}$. Since $Q_1$ is plain, we know $T[v',x_{j'}] \subseteq Q_1$.

Define $\Q_1=Q_2 \setminus Q_2[u',v]$ and $\Q_2=Q_1 \setminus Q_1[v',u]$.
Let also~$\tau'$ denote those trees in~$\T \setminus \{T\}$ that share a vertex with~$\Q_1 \cup \Q_2$. 
We can then prove that $(\Q_1,\Q_2)$ is a partial solution for~$(u',v',\tau')$; moreover,
$Q_2[u',v]$ is a path in the auxiliary graph \hbox{$G \langle T[v',u],u',v, \tau \setminus \tau'\rangle$}.
Thus, Algorithm \ref{alg:PartSol} takes a partial solution $(Q'_1,Q'_2)$ for~$(u',v',\tau')$, already computed,
and computes a shortest $(u',v)$-path~$R$ in $G \langle T[v',u],u',v, \tau \setminus \tau'\rangle$. 
It then creates the path pair  $(Q'_2 \cup T[v',u],Q'_1 \cup R)$ as a candidate for a partial solution for~$(u,v,\tau)$. 
\end{description}

\begin{varalgorithm}{{\textsc{PartSol}}}
\caption{Computes a partial solution $F(u,v,\tau)$ of minimum weight for $(u,v,\tau)$ where $u \in V(T_i)$ and $v \in V(T_j)$ with $i\leq j$, and $\tau \subseteq \T \setminus \{T\}$. 
}
\label{alg:PartSol}
\begin{algorithmic}[1]
\Require{Vertices $u$ and $v$ where $u \in V(T_i)$ and $v \in V(T_j)$ for some $i\leq j$, and a set $\tau \subseteq \T \setminus \{T\}$.}
\Ensure{A partial solution $F(u,v,\tau)$ for~$(u,v,\tau)$ of minimum weight, or $\varnothing$ if not existent.}
\State Let $\mathcal{S}=\emptyset$.
\ForAll{$h \in [2]$}				\label{line:PS-case1-start}
	\If{$A_h \cup T[x_1,u]$ is a path}
		\If{$v$ is reachable from~$a_{3-h}$ in~$G \langle A_h \cup T[x_1,u],a_{3-h},v,\tau \rangle$} 
			\State Compute a shortest $(a_{3-h},v)$-path $R$ in~$G \langle A_h \cup T[x_1,u],u,a_{3-h},v,\tau \rangle$.
			\If{$(A_h \cup T[x_1,u], R )$ is a partial solution for~$(u,v,\tau)$}				
				\State $\mathcal{S} \leftarrow (A_h \cup T[x_1,u], R)$.			\label{line:PS-comp1}
			\EndIf				
		\EndIf
	\EndIf 
\EndFor 
\ForAll{$i' \in [i-1]$ and $u' \in V(T_{i'})$}		\label{line:PS-case-s2tart}
	\ForAll{$j' \in [i] \setminus [i']$ and $v' \in V(T_{j'})$ such that $T[x_{j'},v'] \cap T[x_i,u]=\emptyset$} \label{line:choosev'}
		\ForAll{$\tau' \subseteq \tau$}
			\If{$F(u',v',\tau') = \varnothing$}  {\bf continue;} 	\label{line:PS-callF}
			\EndIf
			\State Let $(Q'_1,Q'_2)=F(u',v',\tau')$.					\label{line:PS-defQ'}
			\If{$v$ is not reachable from~$u'$ in~$G \langle T[v',u],u',v, \tau \setminus \tau' \rangle$}  {\bf continue;}  \label{line:PS-reach}
			\EndIf
			\State Compute a shortest $(u',v)$-path $R$ in~$G \langle  T[v',u],u',v, \tau \setminus \tau' \rangle$.			\label{line:PS-shortest}
			\State Let $P_1=Q'_2 \cup T[v',u]$ and $P_2=Q'_1 \cup R$. 
			\If{$(P_1,P_2)$ is a partial solution for~$(u,v,\tau)$}
				\State $\mathcal{S} \leftarrow (P_1,P_2)$. 						\label{line:PS-case2-end}
			\EndIf
		\EndFor					
	\EndFor		
\EndFor
\If{$\mathcal{S}=\emptyset$} {\bf return} $\varnothing$.
\Else{ Let $S^\star$ be the cheapest pair among those in $\mathcal{S}$, and {\bf return} $F(u,v,\tau):=S^\star$.} \label{line:SP-final}
\EndIf
\end{algorithmic}
\end{varalgorithm}


The following lemma guarantees the correctness of Algorithm~\ref{alg:PartSol}.
Formally, we say that $F(u,v,\tau)$ is \emph{correctly computed}
if either it contains a minimum-weight partial solution for~$(u,v,\tau)$, or 
no partial solution for~$(u,v,\tau)$ exists and $F(u,v,\tau)=\varnothing$.
\begin{lemma}
\label{lem:partsol-alg}
Let $i,j \in [r]$ with $i\leq j$, $u \in V(T_i)$, $v \in V(T_j)$ and $\tau \subseteq \T \setminus \{T\}$.
Assuming that the values $F(u',v',\tau')$ are correctly computed for each $u' \in V(T_{i'})$ with~$i'<i$, 
Algorithm~\ref{alg:PartSol} correctly computes~$F(u,v,\tau)$.
\end{lemma}

\begin{proof} 
Let $(Q_1,Q_2)$ be a partial solution for~$(u,v,\tau)$ with minimum weight. 
We distinguish between two cases. 

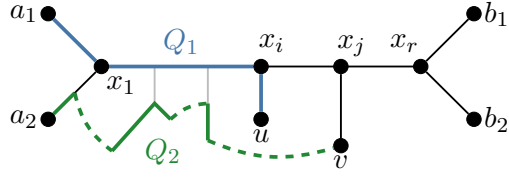
\begin{figure}
    \centering  
    \begin{tikzpicture}[xscale=0.7, yscale=0.7]

  \node[draw, circle, fill=black, inner sep=1.9pt] (a1) at (-1, 1) {};
  \node[draw, circle, fill=black, inner sep=1.9pt] (a2) at (-1, -1) {};
  \node[draw, circle, fill=black, inner sep=1.9pt] (x1) at (0, 0) {};
  \node[draw, circle, fill=black, inner sep=1.9pt] (xi)  at (3, 0) {};
  \node[draw, circle, fill=black, inner sep=1.9pt] (xj) at (4.5, 0) {};
  \node[draw, circle, fill=black, inner sep=1.9pt] (u)  at (3, -1) {};      
  \node[draw, circle, fill=black, inner sep=1.9pt] (v) at (4.5, -1.5) {};
  \node[draw, circle, fill=black, inner sep=1.9pt] (xr) at (6, 0) {};
  \node[draw, circle, fill=black, inner sep=1.9pt] (b1) at (7, 1) {};
  \node[draw, circle, fill=black, inner sep=1.9pt] (b2) at (7, -1) {};
  \coordinate[inner sep=0pt] (h0) at (-0.5,-0.5) {};
  \coordinate[inner sep=0pt] (h6) at (1,0) {};  
  \coordinate[inner sep=0pt] (h7) at (2,0) {};    
  \coordinate[inner sep=0pt] (h4) at (2,-1.4) {};    
  \coordinate[inner sep=0pt] (h5) at (2,-0.7) {};    
  \coordinate[inner sep=0pt] (h2) at (1,-0.7) {};    
  \coordinate[inner sep=0pt] (h1) at (0.2,-1.6) {};      
  \coordinate[inner sep=0pt] (h3) at (1.3,-1) {};      


  \node[left] at (a1) {$a_1$};
  \node[left] at (a2) {$a_2$};
  \node[right] at (b1) {$b_1$};
  \node[right] at (b2) {$b_2$};
  \node[below] at (u) {$u$};
  \node[below] at (v) {$v$};  
  \node[above, yshift=1.7pt, xshift=4pt] at (xi) {$x_i$};
  \node[above, xshift=4pt] at (xj) {$x_j$};    
  \node[below right, xshift=-2pt] at (x1) {$x_1$};
  \node[above, yshift=1.7pt, xshift=-6pt] at (xr) {$x_r$};

  \draw[line width=0.7pt] (a2) to  (x1) to (xr) to (b1);
  \draw[line width=0.7pt,grey] (h2) to (h6);
  \draw[line width=0.7pt,grey] (h5) to (h7);  
  \draw[line width=0.7pt] (xr) to (b2);
  \draw[line width=0.7pt] (v) to (xj);
  
  \draw[line width=1.4pt, blue] (a1) to (x1) to node[pos=0.5, above] {$Q_1$}(xi) to (u);

  \draw[line width=1.4pt, green] (a2) to (h0); 
  \draw[line width=1.4pt, green, dashed] (h0) to [bend right=20]  (h1); 
  \draw[line width=1.4pt, green] (h1) to node [midway, below right] {$Q_2$} (h2) to (h3);
  \draw[line width=1.4pt, green, dashed] (h3) to [bend left=20]  (h5);   
  \draw[line width=1.4pt, green]  (h4) to (h5);
  \draw[line width=1.4pt, green, dashed] (h4) to [bend right=20]  (v);

\end{tikzpicture}
  \caption{Illustration for Case~A in Lemma~\ref{lem:partsol-alg}. 
  Edges within~$T$ are depicted using solid lines, edges not in~$T$ using dashed lines. 
  The figure assumes $Q_1=A_1 \cup T[x_1,u]$ (so $h=1$). Paths~$Q_1$ and~$Q_2$
  are shown in \textcolor{blue}{\textbf{blue}} and in \textcolor{green}{\textbf{green}}, respectively. 
  }
  \label{fig:partsolA}
\end{figure}

\noindent
{\bf Case A:} $Q_2$ contains no vertex of~$T[x_1,x_i]$. See Figure~\ref{fig:partsolA} for an illustration.
Due to Observation~\ref{obs:contains-A1orA2}, 
we know that $Q_1$ contains $A_h$ for some $h \in [2]$; let us fix this value of~$h$.
Since $Q_1$ and~$Q_2$ are locally cheapest and~$Q_1$ ends with~$T[x_i,u]$, we also know that $Q_1$ must contain~$T[x_1,x_i]$.
Therefore, we obtain $Q_1=A_h \cup T[x_1,u]$. 

First, $Q_1$ and $Q_2$ may share a vertex only if $x_1=a_1=a_2$.
We also know that $Q_2$ may only contain vertices of~$T_{(1,i)}$ that are not on the path~$A_h \cup T[x_1,u]$, 
except possibly for the starting vertex~$a_{3-h}$ (in case $a_1=a_2$). 
Second, by condition~(d) in the definition of a partial solution, we also have that $Q_2$ contains no vertices 
from~$T_h$ for any $i<h \leq r$ other than its endpoint~$v$.
Third, by condition~(e) in the definition of a partial solution, we know that $Q_1 \cup Q_2$ contains no vertex that is forbidden for~$(\tau,T)$.
This means that $Q_2$ is a path in the auxiliary graph~$G \langle A_h \cup T[x_1,u],a_{3-h},v,\tau \rangle$.

Conversely, we claim that $(A_h \cup T[x_1,u], R)$ is a partial solution for~$(u,v,\tau)$
for any shortest $(a_{3-h},v)$-path~$R$ in~$G \langle A_h \cup T[x_1,u],a_{3-h},v,\tau \rangle$,
supposing that $A_h \cup T[x_1,u]$ is an $(a_h,u)$-path; 
note that this condition is explicitly checked by the algorithm (and trivial whenever $i>1$).
By the definition of $G \langle A_h \cup T[x_1,u],a_{3-h},v,\tau \rangle$, 
we know that  $A_h \cup T[x_1,u]$ and $R$ are two permissively disjoint $(\{a_1,a_2\},\{u,v\})$-paths.

Let us show that $A_h \cup T[x_1,u]$ and $R$ are also locally cheapest: 
assuming that there exists a shortcut $T[q,q']$ witnessing the opposite,
 we get that both~$q$ and~$q'$ must be vertices of~$R$ on~$T_{(1,i)}$, and 
the path~$T[q,q']$ must share no vertices with~$A_h \cup T[x_1,u]$.
Hence, replacing $R[q,q']$ with $T[q,q']$ would yield a path in~$G \langle A_h \cup T[x_1,u],a_{3-h},v \rangle$ whose weight is less than~$w(R)$
by Lemma~\ref{lem:T-min-pathlength},
contradicting the fact that $R$ is a shortest path. 
Thus, $A_h \cup T[x_1,u]$ and $R$ are locally cheapest.
They also form a well-formed pair, since 
their quasi-monotonicity and plainness is obvious.
It is straightforward to see that they also satisfy conditions~(c)--(f) in Definition~\ref{def:partsol}, 
and thus form a partial solution for~$(u,v,\tau)$. 

Therefore, for the right choice of~$h$, 
Algorithm~\ref{alg:PartSol} will put $(A_h \cup T[x_1,u], R)$ into~$\mathcal{S}$ on line~\ref{line:PS-comp1}
where $R$ is a shortest $(a_{3-h},v)$-path in~$G \langle A_h \cup T[x_1,u],a_{3-h},v,\tau \rangle$.
Observe that 
\begin{align*}
w(A_h \cup T[x_1,u])+w( R) &=w(Q_1) + w(R)  
 \leq w(Q_1)+w(Q_2),
\end{align*}
implying that $(A_h \cup T[x_1,u], R)$ is indeed a partial solution for $(u,v,\tau)$ with minimum total weight.
This proves the lemma for Case~A.

\medskip
\noindent
{\bf Case B:} $Q_2$ contains a vertex of~$T[x_1,x_i]$.
Consider the vertex~$x_{i'}$ on $T[x_1,x_{i-1}]$ closest to~$x_i$ that appears on~$Q_2$, and let $u'$ be the last vertex of $Q_2$ in~$T_{i'}$;
since $Q_2$ is plain, we know $T[x_{i'},u'] \subseteq Q_2$.
Let $x_{j'}$ denote the vertex on~$T[x_{i'},x_i]$ closest to $x_{i'}$ that is contained in~$Q_1$; then $i' <j' \leq i$.
As $Q_1$ and $Q_2$ are locally cheapest, we have $T[x_{j'},x_i] \subseteq Q_1$. 
Let $v'$ denote the first vertex of~$Q_1$ in~$T_{j'}$. Since $Q_1$ is plain, we know $T[v',x_{j'}] \subseteq Q_1$.
Note also that $T[x_{j'},v'] \cap  T[x_i,u]=\emptyset$: 
this is trivial if $x_{j'}\neq x_i$, and it follows from the plainness of~$Q_1$ if $x_{j'}= x_i$, because 
then $T[v',x_{j'}]$ and~$T[x_{j'},u]$ are both subpaths of~$Q_1$ and therefore can share no edge.

Let $\Q_1=Q_2 \setminus Q_2[u',v]$ and $\Q_2=Q_1 \setminus Q_1[v',u]$.
Let also~$\tau'$ denote those trees in~$\T \setminus \{T\}$ that share a vertex with~$\Q_1 \cup \Q_2$. 

\begin{figure}
    \centering  
    \begin{tikzpicture}[xscale=0.7, yscale=0.7]

  \node[draw, circle, fill=black, inner sep=1.7pt] (a1) at (-1, 1) {};
  \node[draw, circle, fill=black, inner sep=1.7pt] (a2) at (-1, -1) {};
  \node[draw, circle, fill=black, inner sep=1.7pt] (x1) at (0, 0) {};
  \node[draw, circle, fill=black, inner sep=1.7pt] (xi')  at (1.5, 0) {};
  \node[draw, circle, fill=black, inner sep=1.7pt] (xj')  at (3, 0) {};
  \node[draw, circle, fill=black, inner sep=1.7pt] (xi) at (5.5, 0) {};
  \node[draw, circle, fill=black, inner sep=1.7pt] (xj) at (7, 0) {};  
  \node[draw, circle, fill=black, inner sep=1.7pt] (u')  at (1.5, -0.8) {};        
  \node[draw, circle, fill=black, inner sep=1.7pt] (v')  at (3, -1.3) {};          
  \node[draw, circle, fill=black, inner sep=1.7pt] (u)  at (5.5, -1) {};      
  \node[draw, circle, fill=black, inner sep=1.7pt] (v) at (7, -1.5) {};
  \node[draw, circle, fill=black, inner sep=1.7pt] (xr) at (8, 0) {};
  \node[draw, circle, fill=black, inner sep=1.7pt] (b1) at (9, 1) {};
  \node[draw, circle, fill=black, inner sep=1.7pt] (b2) at (9, -1) {};
  \coordinate[inner sep=0pt] (h0) at (-0.5,-0.5) {};
  \coordinate[inner sep=0pt] (h0') at (-0.5,-0.37) {};  

  \coordinate[inner sep=0pt] (h4) at (5,-1.7) {};    
  \coordinate[inner sep=0pt] (h5) at (5,-1) {};    
  \coordinate[inner sep=0pt] (h2) at (4,-1.5) {};    
  \coordinate[inner sep=0pt] (h1) at (3.4,-2.1) {};      
  \coordinate[inner sep=0pt] (h3) at (4.3,-1.8) {};      

  \coordinate[inner sep=0pt] (h6) at (4,0) {};  
  \coordinate[inner sep=0pt] (h7) at (5,0) {};

  \node[left] at (a1) {$a_1$};
  \node[left] at (a2) {$a_2$};
  \node[right] at (b1) {$b_1$};
  \node[right] at (b2) {$b_2$};
  \node[below, yshift=-3pt] at (u) {$u$};
  \node[below, yshift=-3pt] at (v) {$v$};  
  \node[right] at (u') {$u'$};
  \node[right] at (v') {$v'$};  
  \node[above, yshift=1.7pt, xshift=4pt] at (xi) {$x_i$};
  \node[above, yshift=1.7pt, xshift=4pt] at (xi') {$x_{i'}$};
  \node[above, xshift=4pt] at (xj) {$x_j$};    
  \node[above, xshift=4pt] at (xj') {$x_{j'}$};      
  \node[below right, xshift=-2pt] at (x1) {$x_1$};

  \draw[line width=0.7pt] (a2) to  (x1) to (xr) to (b1);
  \draw[line width=0.7pt] (xr) to (b2);
  \draw[line width=0.7pt, grey] (h2) to (h6);
  \draw[line width=0.7pt, grey] (h5) to (h7);  
  \draw[line width=0.7pt, grey] (xj) to (v);  
  
  \draw[line width=2.6pt, green] (a1) to (x1) to  node[pos=0.4, above] {$\Q_1$} (xi') to (u');
  \draw[line width=1.4pt, green, arrows = {-Parenthesis[]}] (xi') to (u');
  \draw[line width=1.4pt, green, dashed] (u') to [bend right=40] (h1);
  \draw[line width=1.4pt, green] (h1) to node [pos=0.5, below right] {$\R \subseteq Q_2$}  (h2) to (h3);
  \draw[line width=1.4pt, green, dashed] (h3) to [bend left=20]  (h5);   
  \draw[line width=1.4pt, green]  (h4) to (h5);
  \draw[line width=1.4pt, green, dashed] (h4) to [bend right=20]  (v);  
  
  \draw[line width=2.6pt, blue] (a2) to (h0); 
  \draw[line width=2.6pt, blue, dashed] (h0) to [bend right=50]  node[pos=0.5, below] {$\Q_2$}  (v'); 
  \draw[line width=1.4pt, blue, dashed, arrows = {-Parenthesis[]}] (h0) to [bend right=50]  node[pos=0.5, below] {}  (v'); 
  \draw[line width=1.4pt, blue] (v') to (xj') to  node[pos=0.5, above] {$Q_1$} (xi) to (u);

\end{tikzpicture}
  \caption{Illustration for Case~B in Lemma~\ref{lem:partsol-alg}. 
  Edges within~$T$ are depicted using solid lines, edges not in~$T$ using dashed lines. 
  Paths~$Q_1$ and~$Q_2$ are shown in \textcolor{blue}{\textbf{blue}} and in \textcolor{green}{\textbf{green}}, respectively. 
  Their subpaths $\Q_2$ and $\Q_1$, forming a partial solution for~$(u',v',\tau')$, are depicted in bold, with their endings 
  marked by a parenthesis-shaped delimiter.
 }
  \label{fig:partsolB}
\end{figure}
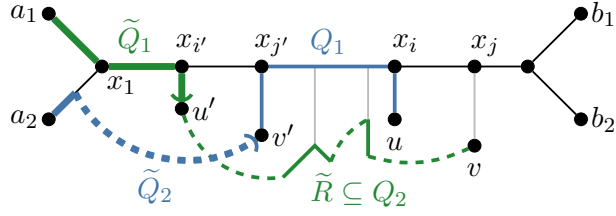

\begin{claim} 
\label{clm:partsol-recurse} $(\Q_1,\Q_2)$ is a partial solution for~$(u',v',\tau')$.
\end{claim}
\begin{claimproof}
Observe that $\Q_1$ and $\Q_2$ are permissively disjoint $(\{a_1,a_2\},\{u',v'\})$-paths.
It is straightforward to verify that they form a well-formed pair and satisfy condition~(c)
in the definition of partial solutions.
Observe that neither~$Q_1$ nor~$Q_2$ passes through an inner vertex of~$T[x_{i'},x_{j'}]$, by the definition of~$x_{i'}$ and $x_{j'}$.
By Lemma~\ref{lem:paths-vs-X} 
this implies that $V(\Q_2) \cap V(T_h)=\emptyset$ for any $i' < h <j'$.
Recall also that $V(\Q_2) \cap V(T_{j'})=\{v'\}$ by the definition of~$v'$. 
Furthermore, since $Q_1$ is quasi-monotone, we also have that $\Q_2$ contains no vertices from~$T_{(j'+1,r)}$,
so $\Q_2$ satisfies condition~(d) as well.
By our definition of~$\tau'$, we also know that $\Q_1 \cup \Q_2$ contains no vertex of~$V(\T \setminus (\tau' \cup \{T\})$, 
that is, no vertex that is forbidden for~$(\tau',T)$. Since condition~(f) holds for~$Q_1$ and $Q_2$, it also holds for their subpaths, $\Q_1$ and~$\Q_2$.
Thus $(\Q_1,\Q_2)$ is a partial solution for~$(u',v',\tau')$.
\end{claimproof}

\begin{claim} 
\label{clm:partsol-path}
$\R=Q_2[u',v]$ is a $(u',v)$-path in~$G \langle T[v',u], u',v,\tau \setminus \tau' \rangle$. 
\end{claim}
\begin{claimproof}
First we show that $\R$ cannot contain any vertex that is forbidden for~$(\tau \setminus \tau',T)$. 
For the sake of contradiction assume that $\R$ contains some vertex~$v_f$ forbidden for~$(\tau \setminus \tau',T)$.
By condition~(e) in the definition of a partial solution for~$(u,v,\tau)$, we know that $Q_1 \cup Q_2$ 
does not contain vertices forbidden for~$(\tau,T)$, and by $\R \subseteq Q_2$, neither does~$\R$.
Note that the vertices that are forbidden for~$(\tau \setminus \tau',T)$ but not for~$(\tau',T)$ are exactly the vertices
\[
\underbrace{V(\T \setminus (\tau \setminus \tau' \cup \{T\}))}_{\textrm{forbidden for }(\tau \setminus \tau',T)} \,\, \setminus  \,\,  
\underbrace{V(\T \setminus (\tau                 \cup \{T\}))}_{\textrm{forbidden for }(\tau,T)} = V(\tau') \]
Consequently, $v_f$ must be contained in $\tau'$. Let $T'$ be the negative tree in~$\tau'$ that contains~$v_f$. 

Recall that by the definition of~$\tau'$, $\Q_1 \cup \Q_2$ must contain a vertex of~$T'$. 
In this case we can show that $Q_1$ contains a vertex of~$T'$:
if $V(\Q_2) \cap V(T') \neq \emptyset$, then this follows from by $\Q_2 \subseteq Q_1$;
if $V(\Q_1) \cap V(T') \neq \emptyset$, then by $\Q_1=Q_2 \setminus Q_2[u',v]$ we get that $Q_2$ visits $T'$ both 
before and after the vertex $u' \in V(T)$, and 
since $Q_1$ and~$Q_2$ are locally cheapest, this implies that $Q_1$ must contain a vertex of~$T'$.
Hence, in either case we get that $Q_1$ and $Q_2$ are in contact at~$T'$, 
a contradiction to our assumption that $(Q_1,Q_2)$ is a partial solution for $(u,v,\tau)$ 
and thus satisfies condition~(f). 
This proves that $\R$ contains no vertex forbidden for~$(\tau \setminus \tau',T)$. 

\smallskip
It is clear that $\R \subseteq Q_2$ has no common vertices with~$T[v',u] \subseteq Q_1$ by the (permissive) disjointness of~$Q_1$ and~$Q_2$.
Hence it suffices to prove the following:
if $\R$ contains a vertex in~$T_\ell$ other than~$u'$ or~$v$, then $j' \leq \ell \leq i$.
First, $\ell<i'$ is not possible by the quasi-monotonicity of~$Q_2$ (using that $x_{i'}$ is on $Q_2$). 
Second, $\ell>i$ is also not possible, due to condition~(d) for~$Q_2$ in the definition of a partial solution.
Third, since $Q_2$ is plain, we get $\ell \neq i'$ from the definition of~$u'$.
Fourth, recall that 
neither~$Q_1$ nor~$Q_2$ passes through an inner vertex of~$T[x_{i'},x_{j'}]$, by the definition of~$x_{i'}$ and~$x_{j'}$.
Hence, by Lemma~\ref{lem:paths-vs-X}, no vertex of~$Q_1$ or $Q_2$
can be contained in~$\bigcup_{i' < h <j'} V(T_h)$, so $\ell$ is not between $i'$ and $j'$.
Thus, $j' \leq \ell \leq i$ as claimed, and so 
$\R$ is a path in~$G \langle T[v',u],u',v,\tau \setminus \tau' \rangle$. 
\end{claimproof}

\smallskip
Consider the iteration when Algorithm~\ref{alg:PartSol} picks the values for $i'$, $j'$, $u'$, $v'$, and $\tau'$ as defined above 
(note that the condition on line~\ref{line:choosev'} holds for these values);
we call this the \emph{lucky} iteration.
In this iteration, the algorithm will find on line~\ref{line:PS-callF} that $F(u',v',\tau') \neq \varnothing$:
by Claim~\ref{clm:partsol-recurse} and our assumption that $F(u',v',\tau')$ is already correctly computed due to~$i'<i$, 
we know that $F(u',v',\tau')=(Q'_1,Q'_2)$ is a partial solution for~$(u',v',\tau')$ with minimum total weight. 
Moreover, by Claim~\ref{clm:partsol-path} Algorithm~\ref{alg:PartSol} will find on line~\ref{line:PS-reach} 
that $v$ is reachable from~$u'$ in~$G \langle T[v',u],u',v,\tau \setminus \tau' \rangle$, 
and so the algorithm will not continue with the next iteration but will proceed with computing a pair~$(P_1,P_2)$
where $P_1=Q'_2 \cup T[v',u]$ and $P_2=Q'_1 \cup R$ for some shortest $(u',v)$-path~$R$ in~$G \langle T[v',u],u',v,\tau \setminus \tau' \rangle$.

In the remainder of the proof we show that $(P_1,P_2)$ is a  partial solution for $(u,v,\tau)$, and moreover, 
its weight is at most~$w(Q_1)+w(Q_2)$.
Clearly, this finishes the proof of the lemma for Case B.
We start with the following claim.
\begin{claim}
\label{clm:caseB-almostpartsol}
$(P_1,P_2)$ satisfies conditions~(b)--(f) of being a partial solution for $(u,v,\tau)$. 
\end{claim}
\begin{claimproof}
Recall that $P_1=Q'_2 \cup T[v',u]$ and $P_2=Q'_1 \cup R$. 

Observe that both~$P_1$ and~$P_2$ are plain and quasi-monotone by the definitions.
Let us show now that $(P_1,P_2)$ is a locally cheapest and, hence, a well-formed pair.
Suppose for contradiction that there is a shortcut for $P_1$ and $P_2$, with endpoints $q$ and~$q'$. 
Since $Q'_1$ and~$Q'_2$ are locally cheapest, it follows that $q$ and $q'$ must lie on~$R$. 
Note first that the auxiliary graph~$G \langle T[v',u],u',v,\tau \setminus \tau' \rangle$ has the property that for each $T' \in \T \setminus \{T\}$ it 
either contains $T'$ as a subgraph, or contains no vertex of~$T'$. Since $R$ is a shortest path in this graph, 
it follows that $R$ intersects each tree in $\T \setminus \{T\}$ in a path (if at all), and hence any shortcut on~$R$ must be within~$T$, 
i.e., it has the form~$T[q,q']$. 
On the one hand, if $T[q,q']$ is contained in the auxiliary graph~$G \langle T[v',u],u',v,\tau \setminus \tau' \rangle$, 
then this contradicts the optimality of~$R$. On the other hand, if the auxiliary graph does not contain $T[q,q']$ 
even though it contains both $q$ and~$q'$, 
then $q$ and~$q'$ must be contained in different components of~$T_{(j',i)} - V(T[v',u])$: in this case, however, by $T[v',u] \subseteq P_1$ 
we obtain that there is a vertex of~$P_1$ on~$T[q,q']$, contradicting our assumption that $T[q,q']$ is a shortcut for~$(P_1,P_2)$.
This proves that $(P_1,P_2)$ is a well-formed pair.

Note that condition~(c) for being a partial solution for~$(u,v,\tau)$ holds, because $T[v',x_{j'}] \cap T[x_i,u] = \emptyset$  
is guaranteed by the algorithm's choice for~$v'$, and so we have $T[x_i,u] \subseteq T[v',u]$.
Condition~(d) holds by the definition of~$G \langle T[v',u],u',v,\tau \setminus \tau' \rangle$ and of~$R$. 

To see that $(P_1,P_2)$ satisfies condition~(e), note that $Q'_1 \cup Q'_2$ contains no vertex that is forbidden for~$(\tau',T)$, 
and $R$ contains no vertex that is forbidden for~$(\tau \setminus \tau',T)$. Hence, neither of them contains a vertex forbidden for~$(\tau,T)$.
Consequently, $P_1 \cup P_2=Q'_2 \cup T[v',u] \cup Q'_1 \cup R$ does not contain vertices forbidden for~$(\tau,T)$, 
i.e., condition~(e) holds for $(P_1,P_2)$.

Finally, let us show that condition~(f) also holds for~$(P_1,P_2)$. Consider some $T' \in \T \setminus \{T\}$.
First, if $T' \notin \tau'$, then the vertices of~$T'$ are forbidden for $(\tau',T)$ and hence cannot appear on $Q'_1 \cup Q'_2$. 
This implies $V(P_1) \cap V(T') = \emptyset$, so $P_1$ and~$P_2$ are certainly not in contact at~$T'$. 
Second, if $T' \in \tau'$, then the vertices of~$T'$ are forbidden for $(\tau \setminus \tau',T)$ and hence cannot appear on $R$,
by the definition of $G \langle T[v',u],u',v,\tau \setminus \tau',\rangle$. 
Moreover, since $(Q'_1,Q'_2)$ is a partial solution for~$(u',v',\tau')$, we know that $Q'_1$ and~$Q'_2$ are not in contact at~$T'$, due to condition~(f).
By $V(T') \cap V(R)=\emptyset$, this implies that $P_1$ and~$P_2$ cannot be in contact at~$T'$ either. 
This proves that $(P_1,P_2)$ satisfies condition~(f).
\end{claimproof}

\smallskip
By Claim~\ref{clm:caseB-almostpartsol}, to show that $(P_1,P_2)$ is a  partial solution for $(u,v,\tau)$,
it remains to prove that $P_1$ and $P_2$ are permissively disjoint $(\{a_1,a_2\},\{u,v\})$-paths. 
We show that this holds whenever $R$ is vertex-disjoint from~$Q'_1 \cup Q'_2$, as stated in Claim~\ref{clm:caseB-Rdisjoint}.
Then we will prove in Claims~\ref{clm:caseB-Q1intersectsR} and~\ref{clm:caseB-Q2intersectsR}, that this is always the case.

\begin{claim}
\label{clm:caseB-Rdisjoint}
If $V(R) \cap V(Q'_1 \cup Q'_2)=\emptyset$, then $(P_1,P_2)$ is a  partial solution for $(u,v,\tau)$. 
\end{claim}
\begin{claimproof}
As a consequence of Claim~\ref{clm:caseB-almostpartsol}, it remains to prove that $P_1$ and $P_2$ are permissively disjoint $(\{a_1,a_2\},\{u,v\})$-paths. 

By definition, $Q'_1$ and $Q'_2$ are permissively disjoint $(\{a_1,a_2\},\{u',v'\})$-paths.
Moreover, they do not share any vertex with~$T[v',u] $ except for $v'$:
for $Q'_1$ this follows from its quasi-monotonicity and $i'<j'$;
for $Q'_2$ this follows from condition~(d) on~$Q'_2$ for being a partial solution for~$(u',v').$
Note that $R$ cannot share a vertex with $T[v',u]$ either,
by the definition of~$G \langle T[v',u],u',v,\tau \setminus \tau' \rangle$. 
Hence, our assumption $V(R) \cap V(Q'_1 \cup Q'_2)=\emptyset$ implies that 
$P_1$ and $P_2$ are permissively disjoint $(\{a_1,a_2\},\{u,v\})$-paths, as claimed.
\end{claimproof}

\smallskip 
Using Claims~\ref{clm:partsol-recurse} and~\ref{clm:partsol-path}, the optimality of~$(Q'_1,Q'_2)$ and that of~$R$, we get
\begin{align}
\begin{split}
\label{eq:qpqp}
w(Q_1) + w(Q_2) & = w(\Q_2) +  w(T[v',u]) + w(\Q_1) + w(\R)  \\
& \geq w(Q'_1)+w(Q'_2) + w(T[v',u]) + w(R) \geq  w(P_1) + w(P_2)
\end{split}
\end{align}
where the last inequality follows from the fact that $Q'_1$, $Q'_2$, and~$T[v',u]$ are 
pairwise edge-disjoint, as we have proved in Claim~\ref{clm:caseB-Rdisjoint},
and moreover, $R$ cannot share a negative-weight edge with any of these four paths.
This latter fact follows from our definition of the auxiliary graph~$G \langle T[v',u],u',v,\tau \setminus \tau' \rangle$,
the quasi-monotonicity of~$Q'_1$ and~$Q'_2$, and the fact that they avoid vertices forbidden for~$(\tau',T)$;
note that each vertex in~$V(\T \setminus \{T\})$ is either forbidden for~$(\tau',T)$, or forbidden for~$(\tau \setminus \tau',T)$ (or both), 
so no edge of~$E(\T \setminus \{T\})$  can be shared by $R$ and $Q'_1 \cup Q'_2$. 

It remains to show that $(P_1,P_2)$ is indeed a partial solution for~$(u,v,\tau)$; Inequality~\ref{eq:qpqp} then implies its optimality.
By Claim~\ref{clm:caseB-Rdisjoint}, it suffices to show that $R$ is vertex-disjoint from $Q'_1 \cup Q'_2$.
To this end, we show in Claims~\ref{clm:caseB-Q1intersectsR} and~\ref{clm:caseB-Q2intersectsR} that assuming otherwise 
contradicts the optimality of~$(Q_1,Q_2)$.


\smallskip 
Let $y$ be the vertex on~$R$ closest to~$v$ that is contained in~$Q'_1 \cup Q'_2$, assuming that such a vertex exists. 

\begin{claim}
\label{clm:caseB-Q1intersectsR}
If $y \in V(Q'_1)$, then 
$(u,v,\tau)$ admits a partial solution of weight less than $w(Q_1)+w(Q_2)$. 
\end{claim}

\begin{claimproof}
Define $P'_2=P_2 \setminus W$ 
for the closed walk $W=Q'_1[y,u'] \cup R[u',y]$; see Figure~\ref{fig:yQ1} for an illustration. 
Note that $Q'_1$ and $R$ cannot share vertices of~$T$ other than~$u'$, because
$Q'_1$ contains no vertex of $T$ outside $T_{(1,i')}$ due to its plainness and quasi-monotonicity, 
and $R$ contains no vertex of~$T$ outside $T_{(j',i)}$ except for~$u'$ and~$v$, because $R$ is in~$G \langle T[v',u],u',v, \tau \setminus \tau' \rangle$;
from $i'<j'$ thus follows $V(Q'_1) \cap V(R) \cap V(T)=\{u'\}$. 
Therefore, $W$ does not contain any edge of~$T$ more than once.

It is also clear that $W$ does not contain any edge of~$E(\T \setminus \{T\})$ more than once:
if some~$T' \in \T \setminus \{T\}$ shares an edge with $Q'_1$, 
then $T' \in \tau'$  holds by condition~(e) for $(Q'_1,Q'_2)$ being a partial solution for~$(u',v',\tau')$; 
however, then the vertices of~$T'$ are forbidden for~$(\tau \setminus \tau',T)$ and hence cannot be contained in~$R$.
This implies that $W$ has non-negative weight by Lemma~\ref{lem:closed-walk}.
Using Inequality~\ref{eq:qpqp} this implies 
\begin{align}
\label{eq:p1p2}
w(P_1) + w(P'_2) &= w(P_1)+w(P_2) -w(W) \leq  w(P_1)+ w(P_2) \leq w(Q_1) + w(Q_2).
\end{align}

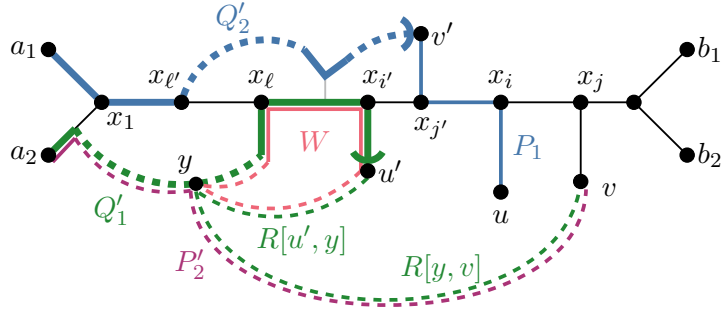
\begin{figure}
    \centering  
    \begin{tikzpicture}[xscale=0.7, yscale=0.7]

  \node[draw, circle, fill=black, inner sep=1.7pt] (a1) at (-1, 1) {};
  \node[draw, circle, fill=black, inner sep=1.7pt] (a2) at (-1, -1) {};
  \node[draw, circle, fill=black, inner sep=1.7pt] (x1) at (0, 0) {};
  \node[draw, circle, fill=black, inner sep=1.7pt] (xell')  at (1.5, 0) {};
  \node[draw, circle, fill=black, inner sep=1.7pt] (xell)  at (3, 0) {};
  \coordinate[inner sep=0pt] (w)  at (3, -1) {};
  \node[draw, circle, fill=black, inner sep=1.7pt] (xi')  at (5, 0) {};
  \node[draw, circle, fill=black, inner sep=1.7pt] (xj')  at (6, 0) {};
  \node[draw, circle, fill=black, inner sep=1.7pt] (xi) at (7.5, 0) {};
  \node[draw, circle, fill=black, inner sep=1.7pt] (xj) at (9, 0) {};  
  \node[draw, circle, fill=black, inner sep=1.7pt] (u')  at (5, -1.3) {};        
  \node[draw, circle, fill=black, inner sep=1.7pt] (v')  at (6, 1.3) {};          
  \node[draw, circle, fill=black, inner sep=1.7pt] (u)  at (7.5, -1.7) {};      
  \node[draw, circle, fill=black, inner sep=1.7pt] (v) at (9, -1.5) {};
  \node[draw, circle, fill=black, inner sep=1.7pt] (xr) at (10, 0) {};
  \node[draw, circle, fill=black, inner sep=1.7pt] (b1) at (11, 1) {};
  \node[draw, circle, fill=black, inner sep=1.7pt] (b2) at (11, -1) {};
  \coordinate[inner sep=0pt] (h0) at (-0.5,-0.5) {};
  \coordinate[inner sep=0pt] (h4) at (4.2,0) {};
  \coordinate[inner sep=0pt] (h2) at (4.2,0.5) {};
  \coordinate[inner sep=0pt] (h1) at (3.8,1) {};  
  \coordinate[inner sep=0pt] (h3) at (4.6,0.8) {};  
   

  \node[left] at (a1) {$a_1$};
  \node[left] at (a2) {$a_2$};
  \node[right] at (b1) {$b_1$};
  \node[right] at (b2) {$b_2$};
  \node[below, yshift=-3pt] at (u) {$u$};
  \node[right, yshift=-3pt, xshift=4pt] at (v) {$v$};  
  \node[right] at (u') {$u'$};
  \node[right] at (v') {$v'$};  
  \node[above, yshift=1.7pt, xshift=0pt] at (xell) {$x_\ell$};
  \node[above, yshift=1.7pt, xshift=-6pt] at (xell') {$x_{\ell'}$};
  \node[above, yshift=1.7pt, xshift=0pt] at (xi) {$x_i$};
  \node[above, yshift=1.7pt, xshift=4pt] at (xi') {$x_{i'}$};
  \node[above, xshift=4pt] at (xj) {$x_j$};    
  \node[below, xshift=4pt, yshift=-2pt] at (xj') {$x_{j'}$};      
  \node[below right, xshift=-2pt] at (x1) {$x_1$};

  \draw[line width=0.7pt] (a2) to  (x1) to (xr) to (b1);
  \draw[line width=0.7pt] (xr) to (b2);
  \draw[line width=0.7pt,grey] (h2) to (h4);
  \draw[line width=0.7pt] (v) to (xj);  

  \draw[line width=2.6pt, green] (a2) to (h0); 
  \draw[line width=2.6pt, green, dashed] (h0) to [bend right=50] node(y)[pos=0.66,circle, fill=black, line width=1pt, inner sep=1.7pt] {} 
  		node[pos=0.3, below, yshift=-4pt,xshift=-4pt] {$Q'_1$}(w);

  \path (u') +(0,-0.13) coordinate (hu');
  \path (y) +(0,-0.13) coordinate (hy);
  \coordinate (yy) at (y);
  \node[above, xshift=-4pt] at (yy) {$y$};    
  \coordinate (uu') at (u');  
  \path (w) +(0.13,-0.13) coordinate (hw);
  \path (xell) +(0.13,-0.13) coordinate (hell);
  \path (xi') +(-0.13,-0.13) coordinate (hi');  
  \path (u') +(-0.13,0) coordinate (hhu');    
  \draw[line width=1.4pt, loopcolor, dashed] (hhu') to [bend left=40] (yy);
  \draw[line width=1.4pt, loopcolor, dashed] (y) to [bend right=30] (hw);
  \draw[line width=1.4pt, loopcolor] (hw) to (hell) to node[midway, below,yshift=-4pt] {$W$} (hi') to (hhu');
  
  \draw[line width=1.4pt, green, dashed] (hu') to [bend left=40] node[pos=0.4, below,yshift=0pt] {$R[u',y]$}  (hy);
  \draw[line width=1.4pt, green, dashed] (y.south) to [bend right=80] node[pos=0.6, above,yshift=0pt] {$R[y,v]$} (v);

  \draw[line width=2.6pt, blue] (a1) to (x1) to (xell');
  \draw[line width=2.6pt, blue, dashed] (xell') to [bend left=50] node[midway,above] {$Q'_2$} (h1);
  \draw[line width=2.6pt, blue] (h1) to (h2) to (h3);
  \draw[line width=2.6pt, blue, dashed, arrows = {-Parenthesis[]}] (h3) to [bend left=20] (v');  
  \draw[line width=1.4pt, blue]  (v') to (xj') to (xi) to node[midway, right] {$P_1$} (u);
  
  \draw[line width=2.6pt, green, arrows = {-Parenthesis[]}] (w) to (xell) to (xi') to (u');

	  \path (a2) +(0.08,-0.13) coordinate (ba2);
	  \path (h0) +(0,-0.18) coordinate (bh0);
	  \path (yy) +(-0.1,-0.13) coordinate (by);  
	  \path (v) +(0.1,-0.15) coordinate (rv);  
    
  \draw[line width=1.4pt, alterpathcolor] (ba2) to (bh0); 
  \draw[line width=1.4pt, alterpathcolor, dashed] (bh0) to [bend right=33](by);   
  \draw[line width=1.4pt, alterpathcolor, dashed] (by) to [bend right=88] node[pos=0.2, left, xshift=-6pt] {$P'_2$}(rv);

   \node[circle,fill=black, inner sep=1.9pt] (y) at (yy){};
  \node[draw, circle, fill=black, inner sep=1.9pt] (a1) at (-1, 1) {};
  \node[draw, circle, fill=black, inner sep=1.9pt] (a2) at (-1, -1) {};
  \node[draw, circle, fill=black, inner sep=1.9pt] (x1) at (0, 0) {};
  \node[draw, circle, fill=black, inner sep=1.9pt] (xell')  at (1.5, 0) {};
  \node[draw, circle, fill=black, inner sep=1.9pt] (xell)  at (3, 0) {};
  \node[draw, circle, fill=black, inner sep=1.9pt] (xi')  at (5, 0) {};
  \node[draw, circle, fill=black, inner sep=1.9pt] (xj')  at (6, 0) {};
  \node[draw, circle, fill=black, inner sep=1.9pt] (xi) at (7.5, 0) {};
  \node[draw, circle, fill=black, inner sep=1.9pt] (xj) at (9, 0) {};  
  \node[draw, circle, fill=black, inner sep=1.9pt] (u')  at (5, -1.3) {};        
  \node[draw, circle, fill=black, inner sep=1.9pt] (v')  at (6, 1.3) {};          
  \node[draw, circle, fill=black, inner sep=1.9pt] (u)  at (7.5, -1.7) {};      
  \node[draw, circle, fill=black, inner sep=1.9pt] (v) at (9, -1.5) {};
  \node[draw, circle, fill=black, inner sep=1.9pt] (xr) at (10, 0) {};
  \node[draw, circle, fill=black, inner sep=1.9pt] (b1) at (11, 1) {};
  \node[draw, circle, fill=black, inner sep=1.9pt] (b2) at (11, -1) {};

\end{tikzpicture}
  \caption{Illustration for Claim~\ref{clm:caseB-Q1intersectsR}. 
  Paths~$P_1$ and~$P_2$ are shown in \textcolor{blue}{\textbf{blue}} and in \textcolor{green}{\textbf{green}}, respectively. 
  Their subpaths $Q'_2$ and $Q'_1$ are depicted in bold, with their endings  marked by a parenthesis-shaped delimiter.
    We highlighted the closed walk~$W$ in \textcolor{coralred}{\textbf{coral red}}, and path~$P'_2$ in \textcolor{purple}{\textbf{purple}}.
 }
  \label{fig:yQ1}

\end{figure}

It is clear by our choice of~$y$ that $R[y,v]$ shares no vertices with~$Q'_1$ or $Q'_2$.
Using the same arguments as in the proof of Claim~\ref{clm:caseB-Rdisjoint}, 
we obtain that $P'_1$ and $P_2$ are permissively disjoint $(\{a_1,a_2\},\{u,v\})$-paths. 
It is also clear that they satisfy conditions~(c)--(f) of being a partial solution for~$(u,v,\tau)$, 
and that $P_1$ is plain and quasi-monotone.
To see that $P'_2$ is also plain and quasi-monotone, recall that $y$ cannot be on~$T$.
Since $P'_2$ is obtained by deleting the subpath of~$Q'_1$ starting at~$y$, and then appending a subpath of~$R$ (not containing any vertex of~$X$) 
it follows that $P'_2$ is quasi-monotone and plain.

Therefore, $(P_1,P'_2)$ satisfies all conditions for being a partial solution for~$(u,v,\tau)$, except the condition of being locally cheapest.
\smallskip

We are going to show that $(P_1,P'_2)$ is \emph{not} locally cheapest, but $\Amend(P_1,P'_2)$ is a partial solution for~$(u,v,\tau)$. 
Observe that this suffices to prove the claim: 
if $(P_1,P'_2)$ is \emph{not} locally cheapest, 
then by Observation~\ref{obs:locally-cheapest} and Inequality~\ref{eq:p1p2} we get that 
$\Amend(P_1,P'_2)$ has weight less than $w(P_1)+w(P'_2) \leq w(Q_1)+w(Q_2)$.

\smallskip
By Claim~\ref{clm:caseB-almostpartsol}, we know that $P_1$ and~$P_2$ are not in contact at any tree of~$\T$ other than~$T$, 
so any shortcut for $(P_1,P'_2)$ must be on~$T$.
Recall also that $P_1$ and~$P_2$ are locally cheapest, again by Claim~\ref{clm:caseB-almostpartsol}. 
Hence, any shortcut for $(P_1,P'_2)$ must be the result of deleting the closed walk~$W=Q'_1[y,u'] \cup R[u',y] \subseteq P_2$.
This means that if $T[q,q']$ is a shortcut for~$(P_1,P'_2)$, then $T[q,q']$ contains an inner vertex $\hat{q} \in V(W)$. 
In addition, note that $\hat{q}$ cannot be on~$R$ unless $\hat{q}=u'$: since $V(P_1)\cap V(T_{(j',i)})=V(T[v',u])$, 
we know that $P_1$ contains no two vertices that may be separated on~$T$ 
by any vertex of~$T_{(j',i)}$. 
Hence, any shortcut for~$(P_1,P'_2)$ is of the form $T[q,q']$ where 
$q$ and $q'$ are on~$P_1$, 
and $T[q,q']$ contains a vertex $\hat{q} \in V(Q'_1[y,u'])$.

Let $(\widetilde{P}_1,\widetilde{P}_2)=\Amend(P_1,P'_2)$. 
By the above arguments we have $\widetilde{P}_2=P'_2 = Q'_1 \cup R \setminus W$. 
Observe that by amending shortcuts we cannot violate permissive disjointness, 
and it is easy to see that conditions~(c)--(f) remain true as well.
Moreover, after amending all shortcuts, the resulting pair of paths must be locally cheapest. 
Recall that we already proved that $\widetilde{P}_2=P'_2$ is plain and quasi-monotone. 
Hence, to prove that $(\widetilde{P}_1,\widetilde{P}_2)$ is a partial solution for~$(u,v,\tau)$
it suffices to check the plainness and quasi-monotonicity of $\widetilde{P}_1$.


Let $x_\ell$ be the first vertex on~$Q'_1$ after $y$ that is on~$X$.
By $y \notin V(T)$ we know that $Q'_1$ does not contain $T[x_1,x_\ell]$, and since $Q'_1$ and $Q'_2$ are locally cheapest,
$T[x_1,x_\ell]$ must contain a vertex of~$Q'_2$. Let $x_{\ell'}$ denote the vertex on $T[x_1,x_\ell]$ closest to~$x_\ell$ that belongs to~$Q'_2$. 
We claim that $\widetilde{P}_1= Q'_2 \setminus Q'_2[x_{\ell'},x_{j'}] \cup T[x_{\ell'},x_{j'}]$. 
To see this, note that the deletion of the vertices of~$Q'_1[x_\ell,u']$ 
removes all vertices of~$Q'_1$ from $T[x_{\ell'},x_{j'}]$; these vertices (weakly) follow $x_\ell$ on~$Q'_1$ by its quasi-monotonicity. 
Hence, as $x_{\ell'}$ and $x_{j'}$ are both on~$Q'_2$, 
we obtain $T[x_{\ell'},x_{j'}] \subseteq \widetilde{P}_1$.
Observe also that if a vertex~$z$ of $Q'_1[y,x_\ell]$  is contained in some~$T_h$ with $h<\ell$, then $x_h \in V(Q'_1 \cup Q'_2)$ by Lemma~\ref{lem:paths-vs-X},
which implies $x_h \in V(Q'_2)$ by the definition of~$x_\ell$, and thus we get $h \leq \ell'$ from the definition of~$x_{\ell'}$. 
However, then the removal of~$z$ from~$Q'_1$ cannot result in a shortcut, since $Q'_2$ is plain and thus its intersection with~$T_h$ is a path. 
Therefore, $\widetilde{P}_1= Q'_2 \setminus Q'_2[x_{\ell'},x_{j'}] \cup T[x_{\ell'},x_{j'}]$ indeed holds.
As a consequence, the plainness and quasi-monotonicity of~$P_1$ immediately implies that $\widetilde{P}_1$ is plain and quasi-monotone as well. 
Hence, $(\widetilde{P}_1,\widetilde{P}_2)$ is a partial solution for~$(u,v,\tau)$.

To see that $(P_1,P'_2)$ is not locally cheapest,
observe that at least one shortcut is created when deleting~$W$: we know that $T[x_{\ell'},x_{j'}] \subseteq \widetilde{P}_1$,
but $x_{i'}$ lies on the path~$T[x_{\ell'},x_{j'}]$ and is a vertex of~$W$. Hence $w(\widetilde{P}_1)+w(\widetilde{P}_2)<w(P_1)+w(P_2) \leq w(Q_1)+w(Q_2)$.
\end{claimproof}

\smallskip
The proof of Claim~\ref{clm:caseB-Q2intersectsR} is similar to the proof of Claim~\ref{clm:caseB-Q1intersectsR}, 
but it is not entirely symmetric.

\begin{claim}
\label{clm:caseB-Q2intersectsR}
If $y \in V(Q'_2)$, then 
$(u,v,\tau)$ admits a partial solution of weight less than $w(Q_1)+w(Q_2)$. 
\end{claim}

\begin{claimproof}
Define $P'_1=Q'_1 \setminus T[x_{i'},u'] \cup T[x_{i'},u]$ and $P'_2=Q'_2 \setminus Q'_2[y,v'] \cup R[y,v]$, and
consider the closed walk
\[W= R[u',y] \cup Q'_2[y,v'] \cup T[u',v'].\]

See Figure~\ref{fig:yQ2} for an illustration.
We argue that $W$ does not contain any edge of~$T$ more than once. 
First, by the definition of $G \langle T[v',u],u',v,\tau \setminus \tau' \rangle$ we know that 
any vertex of~$T$ contained in the path $R$ belongs to~$T_{(j',i)}$, but not to~$T[v',u]$. 
Hence, $R$ shares no edges with~$T[u',v']$.
Condition~(d) for $Q'_2$ being a partial solution for~$(u',v',\tau')$
requires $V(T_{(i'+1,r)}) \cap V(Q'_2) \subseteq \{v'\}$.
This shows that $Q'_2$ cannot share an edge of~$T$ with~$R$ or with~$T[x_{i'},v']$; note that $y \notin V(T)$ also follows.
Lastly, $Q'_2[y,v']$ shares no edge of~$T$ with $T[u',x_{i'}]$ by
the (permissive) disjointness of $Q'_1$ and $Q'_2$. 

It is also clear that $W$ does not contain any edge of~$E(\T \setminus \{T\})$ more than once:
if some~$T' \in \T \setminus \{T\}$ shares an edge with $Q'_2$, 
then $T' \in \tau'$  holds by condition~(e) for $(Q'_1,Q'_2)$ being a partial solution for~$(u',v',\tau')$; 
however, then the vertices of~$T'$ are forbidden for~$(\tau \setminus \tau',T)$ and hence cannot be contained in~$R$.
Therefore, by Lemma~\ref{lem:closed-walk} we know $w(W)\geq 0$.
By Inequality~\ref{eq:qpqp} this implies 
\begin{align} 
w(P'_1) + w(P'_2)
& = w(P_1)+ w(P_2)- w(T[x_{i'},u'] \cup R[u',y]  \cup Q'_2[y,v'] \cup T[v',x_{j'}])  \notag \\
& \phantom{=,} + w(T[x_{i'},v']) \notag \\
& = w(P_1)+w(P_2)- w(W \setminus T[x_{i'},x_{j'}]) + w(T[x_{i'},x_{j'}]) \notag \\
& < w(P_1)+w(P_2) \leq w(Q_1) + w(Q_2). \label{eq:vmi}
\end{align}

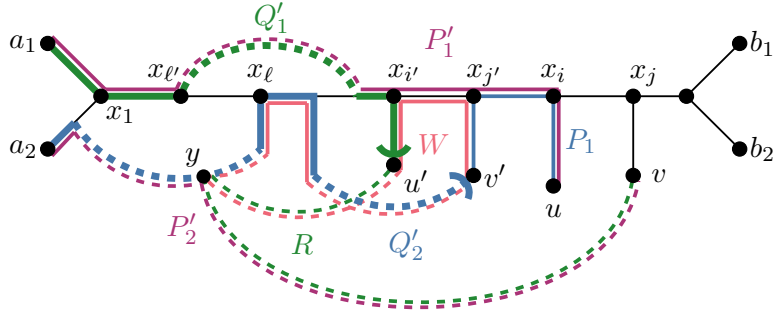
\begin{figure}
    \centering  
    \begin{tikzpicture}[xscale=0.7, yscale=0.7]

  \node[draw, circle, fill=black, inner sep=1.7pt] (a1) at (-1, 1) {};
  \node[draw, circle, fill=black, inner sep=1.7pt] (a2) at (-1, -1) {};
  \node[draw, circle, fill=black, inner sep=1.7pt] (x1) at (0, 0) {};
  \node[draw, circle, fill=black, inner sep=1.7pt] (xell')  at (1.5, 0) {};
  \node[draw, circle, fill=black, inner sep=1.7pt] (xell)  at (3, 0) {};
  \coordinate[inner sep=0pt] (w)  at (3, -1) {};
  \node[draw, circle, fill=black, inner sep=1.7pt] (xi')  at (5.5, 0) {};
  \node[draw, circle, fill=black, inner sep=1.7pt] (xj')  at (7, 0) {};
  \node[draw, circle, fill=black, inner sep=1.7pt] (xi) at (8.5, 0) {};
  \node[draw, circle, fill=black, inner sep=1.7pt] (xj) at (10, 0) {};  
  \node[draw, circle, fill=black, inner sep=1.7pt] (u')  at (5.5, -1.3) {};        
  \node[draw, circle, fill=black, inner sep=1.7pt] (v')  at (7, -1.5) {};          
  \node[draw, circle, fill=black, inner sep=1.7pt] (u)  at (8.5, -1.7) {};      
  \node[draw, circle, fill=black, inner sep=1.7pt] (v) at (10, -1.5) {};
  \node[draw, circle, fill=black, inner sep=1.7pt] (xr) at (11, 0) {};
  \node[draw, circle, fill=black, inner sep=1.7pt] (b1) at (12, 1) {};
  \node[draw, circle, fill=black, inner sep=1.7pt] (b2) at (12, -1) {};
  \coordinate[inner sep=0pt] (h0) at (-0.5,-0.5) {};
    \coordinate[inner sep=0pt] (h7) at (4.8,0) {};  
    \coordinate[inner sep=0pt] (h5) at (4,0) {};      
    \coordinate[inner sep=0pt] (h6) at (4,-1.5) {};          
    \coordinate[inner sep=0pt] (h8) at (3,-1) {};              
   

  \node[left] at (a1) {$a_1$};
  \node[left] at (a2) {$a_2$};
  \node[right] at (b1) {$b_1$};
  \node[right] at (b2) {$b_2$};
  \node[below, yshift=-3pt] at (u) {$u$};
  \node[right, yshift=0pt, xshift=3pt] at (v) {$v$};  
  \node[below right, yshift=2pt] at (u') {$u'$};
  \node[right] at (v') {$v'$};  
  \node[above, yshift=1.7pt, xshift=0pt] at (xell) {$x_\ell$};
  \node[above, yshift=1.7pt, xshift=-6pt] at (xell') {$x_{\ell'}$};
  \node[above, yshift=1.7pt, xshift=0pt] at (xi) {$x_i$};
  \node[above, yshift=1.7pt, xshift=4pt] at (xi') {$x_{i'}$};
  \node[above, xshift=4pt] at (xj) {$x_j$};    
  \node[above, xshift=4pt] at (xj') {$x_{j'}$};      
  \node[below right, xshift=-2pt] at (x1) {$x_1$};

  \draw[line width=0.7pt] (a2) to  (x1) to (xr) to (b1);
  \draw[line width=0.7pt] (xr) to (b2);
  \draw[line width=0.7pt] (v) to (xj);

  \draw[line width=2.6pt, blue] (a2) to (h0);
  \draw[line width=2.6pt, blue, dashed] (h0) to [bend right=50]  node(y)[pos=0.7,circle, fill=black, line width=1pt, inner sep=1.7pt] {} (h8);
  \coordinate (yy) at (y);
  \node[above, xshift=-4pt] at (yy) {$y$};    

  \draw[line width=2.6pt, green] (a1) to (x1) to (xell'); 
  \draw[line width=2.6pt, green, dashed] (xell') to [bend left=70] node[midway,above] {$Q'_1$} (h7);
  \path (yy) +(0.1,0.07) coordinate (ay);   	  
  \draw[line width=1.4pt, green, dashed] (u') to [bend left=50] (ay);  
  \draw[line width=1.4pt, green, dashed] (yy) to [bend right=80] node[pos=0.3, above,yshift=6pt] {$R$} (v);  

  \path (a1) +(0.08,0.13) coordinate (fa1);
  \path (x1) +(0.08,0.13) coordinate (fx1);  
  \path (xell') +(-0.08,0.13) coordinate (fxell');  
  \path (h7) +(0.08,0.13) coordinate (fh7);  
  \path (xi) +(0.1,0.13) coordinate (fxi);  
  \path (u) +(0.1,0) coordinate (fu); 
  \path (a2) +(0.08,-0.13) coordinate (ba2);
  \path (h0) +(0,-0.18) coordinate (bh0);
  \path (yy) +(-0.04,-0.13) coordinate (by);  
  \path (v) +(0.1,-0.1) coordinate (rv);  
  \path (xj) +(0.1,0) coordinate (rxj);  
  
  \draw[line width=1.4pt, alterpathcolor] (fa1) to (fx1);
  \draw[line width=1.4pt, alterpathcolor] (fx1) to (fxell'); 
  \draw[line width=1.4pt, alterpathcolor, dashed] (fxell') to [bend left=70] (fh7);
  \draw[line width=1.4pt, alterpathcolor] (fh7) to node[pos=0.4, above, yshift=8pt] {$P'_1$} (fxi); 
  \draw[line width=1.4pt, alterpathcolor] (fxi) to (fu); 
  
  \draw[line width=1.4pt, alterpathcolor] (ba2) to (bh0); 
  \draw[line width=1.4pt, alterpathcolor, dashed] (bh0) to [bend right=33](by);   
  \draw[line width=1.4pt, alterpathcolor, dashed] (by) to [bend right=80] node[pos=0.1, left, xshift=-4pt] {$P'_2$}(rv);   

  \path (v') +(-0.13,-0.2) coordinate (lv'); 
  \path (xj') +(-0.13,-0.1) coordinate (lxj'); 
  \path (xi') +(0.13,-0.1) coordinate (lxi');   
  \path (u') +(0.13,-0.1) coordinate (lu');   	  
  \path (yy) +(0,-0.03) coordinate (fy);   	  
  \path (h8) +(0.13,-0.1) coordinate (rh8);   	    
  \path (xell) +(0.13,-0.13) coordinate (rxell);  
  \path (h5) +(-0.13,-0.13) coordinate (lh5);   	       	    
  \path (h6) +(-0.13,-0.13) coordinate (lh6);   	       	      
  
  \draw[line width=1.4pt, loopcolor] (lv') to node[midway, left, xshift=-2pt] {$W$} (lxj');
  \draw[line width=1.4pt, loopcolor] (lxj') to (lxi');  
  \draw[line width=1.4pt, loopcolor] (lxi') to (lu');    
  \draw[line width=1.4pt, loopcolor] (lxi') to (lu');    
  \draw[line width=1.4pt, loopcolor, dashed] (lu') to [bend left=50] (yy);  
  \draw[line width=1.4pt, loopcolor, dashed] (yy) to [bend right=20] (rh8);    
  \draw[line width=1.4pt, loopcolor] (rh8) to (rxell);    
  \draw[line width=1.4pt, loopcolor] (rxell) to (lh5);      
  \draw[line width=1.4pt, loopcolor] (lh5) to (lh6);        
  \draw[line width=1.4pt, loopcolor, dashed] (lh6) to [bend right=40] (lv');    

  \draw[line width=2.6pt, green, arrows = {-Parenthesis[]}] (h7) to (xi') to (u');

  \draw[line width=2.6pt, blue] (h8) to (xell) to (h5) to (h6);
  \draw[line width=2.6pt, blue, dashed, arrows = {-Parenthesis[]}] (h6) to [bend right=40] node[pos=0.6, below,yshift=-4pt] {$Q'_2$}(v');  
  \draw[line width=1.4pt, blue] (v') to (xj') to (xi) to node[midway, right] {$P_1$} (u);

  \node[draw, circle, fill=black, inner sep=1.9pt] (a1) at (-1, 1) {};
  \node[draw, circle, fill=black, inner sep=1.9pt] (a2) at (-1, -1) {};
  \node[draw, circle, fill=black, inner sep=1.9pt] (x1) at (0, 0) {};
  \node[draw, circle, fill=black, inner sep=1.9pt] (xell')  at (1.5, 0) {};
  \node[draw, circle, fill=black, inner sep=1.9pt] (xell)  at (3, 0) {};
  \node[draw, circle, fill=black, inner sep=1.9pt] (xi')  at (5.5, 0) {};
  \node[draw, circle, fill=black, inner sep=1.9pt] (xj')  at (7, 0) {};
  \node[draw, circle, fill=black, inner sep=1.9pt] (xi) at (8.5, 0) {};
  \node[draw, circle, fill=black, inner sep=1.9pt] (xj) at (10, 0) {};  
  \node[draw, circle, fill=black, inner sep=1.9pt] (u')  at (5.5, -1.3) {};        
  \node[draw, circle, fill=black, inner sep=1.9pt] (v')  at (7, -1.5) {};          
  \node[draw, circle, fill=black, inner sep=1.9pt] (u)  at (8.5, -1.7) {};      
  \node[draw, circle, fill=black, inner sep=1.9pt] (v) at (10, -1.5) {};
  \node[draw, circle, fill=black, inner sep=1.9pt] (xr) at (11, 0) {};
  \node[draw, circle, fill=black, inner sep=1.9pt] (b1) at (12, 1) {};
  \node[draw, circle, fill=black, inner sep=1.9pt] (b2) at (12, -1) {};
    
   \node[circle,fill=black, inner sep=1.9pt] (y) at (yy){};

\end{tikzpicture}
  \caption{Illustration for Claim~\ref{clm:caseB-Q2intersectsR}. 
  Paths~$P_1$ and~$P_2$ are shown in \textcolor{blue}{\textbf{blue}} and in \textcolor{green}{\textbf{green}}, respectively,
 and their subpaths $Q'_2$ and $Q'_1$ are depicted in bold, with their endings marked by a parenthesis-shaped delimiter.
   Paths~$P'_1$ and~$P'_2$ are shown both in \textcolor{purple}{\textbf{purple}}, 
    and we highlighted the closed walk~$W$ in \textcolor{coralred}{\textbf{coral red}}.
 }
  \label{fig:yQ2}
\end{figure}

It is clear by our choice of~$y$ that $R[y,v]$ shares no vertices with~$Q'_1$ or $Q'_2$.
Using the same arguments as in the proof of Claim~\ref{clm:caseB-Rdisjoint}, 
we obtain that $P'_1$ and $P'_2$ are permissively disjoint $(\{a_1,a_2\},\{u,v\})$-paths. 
It is also clear that they satisfy conditions~(c)--(f), 
and that both $P'_1$ and $P'_2$ are plain and quasi-monotone.
Therefore, $(P'_1,P'_2)$ satisfies all conditions for being a partial solution for~$(u,v,\tau)$, except for the condition of being locally cheapest.

\smallskip
We are going to show that $\Amend(P'_1,P'_2)$ is a partial solution for~$(u,v,\tau)$. 
Observe that this suffices to prove the claim, as
by Observation~\ref{obs:locally-cheapest} and Inequality~\ref{eq:vmi}, 
$\Amend(P'_1,P'_2)$ has weight at most $w(P'_1)+w(P'_2) < w(Q_1)+w(Q_2)$.

\smallskip
By Claim~\ref{clm:caseB-almostpartsol}, we know that $P_1$ and~$P_2$ are not in contact at any tree of~$\T$ other than~$T$, 
so any shortcut for $(P'_1,P'_2)$ must be on~$T$.
Let $(\widetilde{P}_1,\widetilde{P}_2)=\Amend(P'_1,P'_2)$. 
Observe that by amending shortcuts we cannot violate permissive disjointness, 
and it is obvious that conditions~(c)--(f) remain true as well. 
Moreover, after amending all shortcuts, the resulting pair of paths must be locally cheapest. 
Hence, to prove that $(\widetilde{P}_1,\widetilde{P}_2)$ is a partial solution for~$(u,v,\tau)$
it suffices to check the plainness and quasi-monotonicity of $\widetilde{P}_1$ and $\widetilde{P}_2$.

Suppose for contradiction that the condition of plainness or quasi-monotonicity is violated for~$\widetilde{P}_h$ for some~$h \in [2]$
at some vertex~$x_k$; 
by the definition of these properties (see Definitions~\ref{def:plain} and~\ref{def:quasimonotone}) 
we have $x_k \in V(\widetilde{P}_h)$. 
If $x_k$ is also contained in $P'_h$, then exactly by the plainness and quasi-monotonicity of~$P'_1$ and~$P'_2$, these properties
will not get violated at~$x_k$ when amending shortcuts.
Hence, $x_k \notin V(P'_h)$ and thus must be contained in a shortcut for~$(P'_1,P'_2)$. 
Since $P'_1$ ends with $T[x_{i'},u]$, we have that $\widetilde{P}_1$ must also end with~$T[x_{i'},u]$; 
therefore, $x_k \notin V(T[x_{i'},x_i])$.
Note that $k > i$ is not possible either, since $R$ and thus $P'_2$ contains no vertex of $T_{(i+1,r)}$ other than~$v$.
Thus, we get $k<i'$. 

Note that $x_k$ cannot be contained in a shortcut that results from the deletion of $R[y,u'] \cup T[u',x_{i'}]$, 
since such shortcuts cannot contain a vertex of $T_{(1,i'-1)}$.
Consequently, $x_k$ is contained in a shortcut resulting from the deletion of~$Q'_2[y,v']$; 
it also follows that the two endpoints of such a shortcut must lie on~$Q'_1$. 
In particular, $x_k$ is a vertex of $\widetilde{P}_1$ not contained in~$P'_1$.
Since $Q'_1$ and~$Q'_2$ are locally cheapest and plain, we get that 
$Q'_2[y,v']$ must contain a vertex of~$T[x_1,x_{i'-1}]$.

Let $x_\ell$ be the first vertex on~$Q'_2$ after $y$ that is on~$X$ (by our previous sentence, such a vertex exists). 
By $y \notin V(T)$ we know that $Q'_2$ does not contain $T[x_1,x_\ell]$, and since $Q'_1$ and $Q'_2$ are locally cheapest,
it must contain a vertex of~$Q'_1$. Let $x_{\ell'}$ denote the vertex on $T[x_1,x_\ell]$ closest to~$x_\ell$ that belongs to~$Q'_1$. 
We claim that $\widetilde{P}_1= Q'_1 \setminus Q'_1[x_{\ell'},x_{i'}] \cup T[x_{\ell'},u]$. 
To see this, note that the deletion of the vertices of~$Q'_2[x_\ell,v']$ 
removes all vertices of~$Q'_2$ from $T[x_{\ell'},x_{i'}]$; these vertices (weakly) follow $x_\ell$ on~$Q'_2$ by its quasi-monotonicity. 
Hence, as $x_{\ell'}$ and $x_{i'}$ are both on~$Q'_2$, 
we obtain $T[x_{\ell'},x_{i'}] \subseteq \widetilde{P}_1$.
Observe also that if a vertex~$z$ of $Q'_2[y,x_\ell]$  is contained in some~$T_h$ with $h<\ell$, then $x_h \in V(Q'_1 \cup Q'_2)$ by Lemma~\ref{lem:paths-vs-X},
which implies $x_h \in V(Q'_1)$ by the definition of~$x_\ell$, and thus we get $h \leq \ell'$ from the definition of~$x_{\ell'}$. 
However, then the removal of~$z$ from~$Q'_2$ cannot result in a shortcut, since $Q'_1$ is plain and thus its intersection with~$T_h$ is a path. 
Therefore, $\widetilde{P}_1= Q'_1 \setminus Q'_1[x_{\ell'},x_{i'}] \cup T[x_{i'},u]$ indeed holds.
As a consequence, the plainness and quasi-monotonicity of~$P_1$ immediately implies that $\widetilde{P}_1$ is plain and quasi-monotone as well,
a contradiction to our assumption $x_k$ is a vertex on $\widetilde{P}_1$ where one of these properties is violated. 
Hence, $(\widetilde{P}_1,\widetilde{P}_2)$ is a partial solution for~$(u,v,\tau)$.
\end{claimproof}

\smallskip
A direct consequence of Claims~\ref{clm:caseB-Q1intersectsR} and~\ref{clm:caseB-Q2intersectsR}
is that neither $y \in V(Q'_1)$ nor $y \in V(Q_2)$ is possible, as that would contradict the optimality of~$(Q_1,Q_2)$.
Hence, $V(Q'_1 \cup Q'_2) \cap V(R)=\emptyset$, so by Claim~\ref{clm:caseB-Rdisjoint} and Inequality~\ref{eq:qpqp}
we know that $(P_1,P_2)$ is a partial solution for~$(u,v,\tau)$ of minimum weight, as required. 
\end{proof}


\subsection{Assembling the Parts}
\label{sec:permdis}

Using Algorithm~\ref{alg:PartSol} as a subroutine, 
it is now straightforward to construct an algorithm that computes
two permissively disjoint $(\{a_1,a_2\},\{b_1,b_2\})$-paths of minimum total weight; see Algorithm~\ref{alg:PermDisj}.

\begin{varalgorithm}{\textsc{PermDisjoint}}
\caption{Computing two permissively disjoint $(\{a_1,a_2\},\{b_1,b_2\})$-paths of minimum total weight in~$G$. 
}
\label{alg:PermDisj}
\begin{algorithmic}[1]
\Require{Vertices $a_1,a_2,b_1,$ and $b_2$ in~$T$ for which $T[a_1,b_1] \cap T[a_2,b_2] \neq \emptyset$.}
\Ensure{Two permissively disjoint $(\{a_1,a_2\},\{b_1,b_2\})$-paths of minimum total weight, or $\varnothing$ if no such paths exist. }
	\ForAll{$u,v \in V(T) \textrm{ and } \tau \subseteq \T \setminus \{T\}$} $F(u,v,\tau)\leftarrow \varnothing$.
	\EndFor
	\For{$i=1$ to $r$}
		\ForAll{$j$ where $j>i$ or $j=r$}
		   \ForAll{$u \in V(T_i)$, $v \in V(T_j)$ and $ \tau  \subseteq \T \setminus \{T\}$}
		       \State $F(u,v,\tau) \leftarrow $ \textsc{PartSol}$(u,v,\tau)$.					
		   \EndFor
		\EndFor		
	\EndFor
	\If{$F(b_1,b_2,\T \setminus \{T\})=\varnothing$ and $F(b_2,b_1,\T \setminus \{T\})=\varnothing$} {\bf return} $\varnothing$.
	\Else{ {\bf return} a path pair in $\{F(b_1,b_2,\T \setminus \{T\}),F(b_2,b_1,\T \setminus \{T\})\}$ with minimum weight.} \label{line:PD-final}
	\EndIf
\end{algorithmic}
\end{varalgorithm}

Using the correctness of Algorithm~\ref{alg:PartSol}, as established by Lemma~\ref{lem:partsol-alg}, 
and the observation in Lemma~\ref{lem:partsol-to-permdisj} on how partial solutions can be used to 
find two permissively disjoint $(\{a_1,a_2\},\{b_1,b_2\})$-paths of minimum total weight, we immediately obtain the following.

\begin{corollary}
\label{cor:perm-disjoint-paths}
For each constant~$c \in \mathbb{N}$, there is a polynomial-time algorithm that finds two permissively disjoint $(\{a_1,a_2\},\{b_1,b_2\})$-paths of minimum total weight in~$G$ (if such paths exist), where the set of negative edges in~$G$ spans $c$ trees.
\end{corollary}

Using Corollary~\ref{cor:perm-disjoint-paths}, we can now give the proof of Theorem~\ref{thm:DISP-main}.
Algorithm~\ref{alg:DISP} contains a pseudo-code description of our algorithm solving \DISP{}. 

\begin{varalgorithm}{STDP}
\caption{Solving \DISP{} with conservative weights. 
}
\label{alg:DISP}
\begin{algorithmic}[1]
\Require{An instance $(G,w,s,t)$ where $w$ is conservative on~$G$.}
\Ensure{A solution for $(G,w,s,t)$ with minimum weight, or $\varnothing$ if no solution exist.}
\State Let $\mathcal{S}=\emptyset$.
\ForAll{$E' \subseteq E$ such that $E \setminus E'\subseteq E^-$ and $|E \setminus E'|\leq 1$} 	\Comment{Separable solutions.} \label{line:DISP-part1-start}
	\State Create the instance~$(G[E'],w_{|E'},s,t)$. \label{line:DISP-create-instE'}
	\State Let $\T^{\not\ni s,t}=\{T: T$ is a maximal tree in~$G[E' \cap E^-]$ with $s,t \notin V(T)\}$.
	\ForAll{$Z \subseteq V$ such that $|Z \cap V(T)|=1$ for each $T \in \T^{\not\ni s,t}$}  \label{line:DISP-choose-Z}
		\State Create the network~$N_Z$ from instance~$(G[E'],w_{|E'},s,t)$. 	\label{line:DISP-create-netN_Z}	\Comment{Use Def.~\ref{def:N_z}.}
		\If{$\exists$ a flow~$f$ of value~$2$ in $N_Z$}
			\State Compute a minimum-cost flow $f$ of value 2 in $N_Z$.			\label{line:DISP-find-flow2}
			\State Construct a solution~$(S_1,S_2)$ from~$f$ using Lemma~\ref{lem:separablesol-noncontact}. \label{line:DISP-separable-output}
			\State $\mathcal{S} \leftarrow (S_1,S_2)$.							\label{line:DISP-part1-end}
		\EndIf 
	\EndFor
\EndFor
\ForAll{$T \in \T$} 						\Comment{Non-separable solutions.}   \label{line:DISP-choose-tree}
	\ForAll{partitions $(\T_s,\T_0,\T_t)$ of $\T \setminus \{T\}$} \label{line:DISP-choose-part}
		\If{$\T_s \neq \emptyset$}									\label{line:DISP-Ts-notempty}
			\ForAll{$a_1,a_2 \in V(T)$ with $a_1 \neq a_2$}				\label{line:DISP-choose-a1a2}
				\State Create sub-instances $I_s^1$ and~$I_s^2$.		\Comment{Use Def.~\ref{def:sub-instances-Ts}.} \label{line:DISP-def-Ts-subinstances}
				\State Compute $(Q^{\mynearrow}_1,Q^{\mynearrow}_2)= \textup{STDP}(I_s^1)$.
				\State Compute $(Q^{\mysearrow}_1,Q^{\mysearrow}_2)= \textup{STDP}(I_s^2)$.
				\If{$(Q^{\mynearrow}_1,Q^{\mynearrow}_2)\neq \varnothing$ and
					 $(Q^{\mysearrow}_1,Q^{\mysearrow}_2) \neq \varnothing$} 
					\State Create a solution $(S_1,S_2)$ 
					from $(Q^{\mynearrow}_1,Q^{\mynearrow}_2)$ 
					and $(Q^{\mysearrow}_1,Q^{\mysearrow}_2)$ using Lemma~\ref{lem:subinstances-correct}.\label{line:DISP-Ts-recurse}%
					\State $\mathcal{S} \leftarrow (S_1,S_2)$.\label{line:DISP-Ts-output}%
				\EndIf
			\EndFor
		\ElsIf{$\T_t \neq \emptyset$}								\label{line:DISP-Tt-notempty}
			\ForAll{$b_1,b_2 \in V(T)$ with $b_1 \neq b_2$}			\label{line:DISP-choose-b1b2}
				\State Create sub-instances $I_t^1$ and~$I_t^2$.		\Comment{Use Def.~\ref{def:sub-instances-Tt}.} \label{line:DISP-def-Tt-subinstances}
				\State Compute $(Q^{\mynearrow}_1,Q^{\mynearrow}_2)= \textup{STDP}(I_t^1)$.
				\State Compute $(Q^{\mysearrow}_1,Q^{\mysearrow}_2)= \textup{STDP}(I_t^2)$.
				\If{$(Q^{\mynearrow}_1,Q^{\mynearrow}_2)\neq \varnothing$ and
					 $(Q^{\mysearrow}_1,Q^{\mysearrow}_2) \neq \varnothing$} 
					\State Create a solution $(S_1,S_2)$ 
					from $(Q^{\mynearrow}_1,Q^{\mynearrow}_2)$ 
					and $(Q^{\mysearrow}_1,Q^{\mysearrow}_2)$ using Lemma~\ref{lem:subinstances-correct}.\label{line:DISP-Tt-recurse}%
					\State $\mathcal{S} \leftarrow (S_1,S_2)$.	\label{line:DISP-Tt-output}
				\EndIf
			\EndFor
		\Else															\Comment{$\T_s=\T_t=\emptyset$.}
			\ForAll{$a_1,a_2,b_1,b_2 \in V(T)$ that constitute a reasonable guess}					\label{line:DISP-choose-4vert}
				\If{$\exists$ a flow~$f$ of value~$4$ in $N_{(a_1,b_1,a_2,b_2)}$}			\Comment{Use Def.~\ref{def:N_aabb}.} \label{line:DISP-testflow}
					\If{$\exists$ two permissively disjoint $(\{a_1,a_2\},\{b_1,b_2\})$-paths in~$G$} 			\label{line:DISP-testpaths}
						\State Compute a minimum-cost flow $f$ of value 4 in $N_{(a_1,b_1,a_2,b_2)}$.			\label{line:DISP-find-flow4}
						\State Compute permissively disjoint $(\{a_1,a_2\},\{b_1,b_2\})$-paths $Q_1$ and $Q_2$. \\
													\Comment{Use Algorithm~\ref{alg:PermDisj} in Section~\ref{sec:partsol}} \label{line:DISP-use-PD}
						\State Construct a solution~$(S_1,S_2)$ from~$f$, $Q_1$, and $Q_2$ using Lemma~\ref{lem:type2bsol}. \label{line:DISP-nonseparable-construct}
						\State $\mathcal{S} \leftarrow (S_1,S_2)$.		\label{line:DISP-nonseparable-output}
					\EndIf
				\EndIf
			\EndFor			
		\EndIf
	\EndFor
\EndFor 
\If{$\mathcal{S}=\emptyset$} {\bf return} $\varnothing$.
\Else{ Let $S^\star$ be the cheapest pair among those in $\mathcal{S}$, and {\bf return} $S^\star$.} 
\EndIf
\end{algorithmic}
\end{varalgorithm}

\bigskip
\noindent
{\bf Proof of Theorem~\ref{thm:DISP-main}.}
We show that given an instance $(G,w,s,t)$ of \DISP{} where $w$ is conservative on~$G$, 
Algorithm~\ref{alg:DISP} computes a solution for~$(G,w,s,t)$ of minimum weight whenever a solution exists. 
We apply induction based on the number of negative trees in~$\T$.
Suppose that $(P_1,P_2)$ is an optimal solution for~$(G,w,s,t)$ with total weight~$w^\star$. 

First, if $(P_1,P_2)$ is separable, then there exists some $E' \subseteq E$ with $|E \setminus E'| \leq 1$ such that 
the instance $(G[E'],w_{|E'},s,t)$ has a strongly separable solution with weight at most~$w^\star$:
if $(P_1,P_2)$ itself is strongly separable, then this holds for $E'=E$, 
and if $(P_1,P_2)$ is separable but not strongly separable, then by Lemma~\ref{lem:separablesol-contact},
it holds for some $E'=E \setminus \{e\}$, $e \in E^-$.
In either case, Algorithm~\ref{alg:DISP} will create at least one instance~$(G[E'],w_{|E'},s,t)$  on Line~\ref{line:DISP-create-instE'}
that admits a strongly separable solution of weight~$w^\star$.
For this particular instance, Lemma~\ref{lem:separablesol-noncontact} guarantees that for some choice of~$Z$, 
Algorithm~\ref{alg:DISP} will find on Line~\ref{line:DISP-find-flow2} a flow of value~2 and cost at most~$w^\star$ in~$N_Z$, 
from which it obtains a solution~$(S_1,S_2)$ of weight at most~$w^\star$ for~$(G[E'],w_{|E'},s,t)$ using again Lemma~\ref{lem:separablesol-noncontact}.
Note that since $G[E']$ is a subgraph of~$G$, the obtained path pair $(S_1,S_2)$, of weight at most~$w^\star$, is also a solution for $(G,w,s,t)$.

Second, if $(P_1,P_2)$ is not separable, then by Definition~\ref{def:separable} there exists a tree $T \in \T$ at which $P_1$ and~$P_2$ are in contact. 
By Lemma~\ref{lem:partitioned-comps}, there exists a $T$-valid partition of~$\T \setminus \{T\}$.
Let us consider the iteration of Algorithm~\ref{alg:DISP} when it chooses $T$ on Line~\ref{line:DISP-choose-tree}
and a valid $T$-partition~$(\T_s,\T_0,\T_t)$ on Line~\ref{line:DISP-choose-part}.
If $\T_s$ is not empty, then Algorithm~\ref{alg:DISP} guesses the first vertices on the solution paths, $a_1$ and~$a_2$, that are contained in~$T$.
Algorithm~\ref{alg:DISP} then uses recursion to solve the sub-instances~$I_s^1$ and~$I_s^2$, and applies
Lemma~\ref{lem:subinstances-correct} to create a solution for $(G,w,s,t)$ using the obtained solutions for~$I_s^1$ and~$I_s^2$.
To see that the weighted graphs underlying $I_s^1$ and $I_s^2$ are conservative, note that the addition of vertex~$a^\star$ does not create a negative-weight cycle: any cycle~$C$ through~$a^\star$ necessarily includes the edges~$a_1 a^\star$ and~$a_2 a^\star$, plus a path between $a_1$ and $a_2$ whose weight is at least $w(T[a_1,a_2])$ by statement~(1) of Lemma~\ref{lem:T-min-pathlength};
thus we have
$w(C) \geq 2w_{a^\star}+w(T[a_1,a_2])=0$.
Observe that the number of negative trees in~$I_s^1$ is $|\T_s|<|\T|$ because $T \notin \T_s$,
and the number of negative trees in~$I_s^2$ is $|\T \setminus \T_s|<|\T|$ because $\T_s \neq \emptyset$.
Hence, due to our inductive hypothesis we know that these recursive calls for~$I_s^1$ and for~$I_s^2$ return minimum-weight solutions for them. 
By Lemma~\ref{lem:subinstances-correct}, this implies that the solution $(S_1,S_2)$ for~$(G,w,s,t)$ 
obtained on Line~\ref{line:DISP-Ts-recurse} has weight at most~$w(P_1)+w(P_2)$.

Similarly, if $\T_t$ is not empty, then Algorithm~\ref{alg:DISP} guesses the last vertices on the solution paths, $b_1$ and~$b_2$, 
that are contained in~$T$.
Algorithm~\ref{alg:DISP} then uses recursion to solve the sub-instances~$I_t^1$ and~$I_t^2$, and applies
Lemma~\ref{lem:subinstances-correct} to create a solution for $(G,w,s,t)$ using the obtained solutions for~$I_t^1$ and~$I_t^2$.
Again, the weighted graphs in~$I_t^1$ and~$I_t^2$ are conservative due to the value of $w_{b^\star}$ and  statement~(1) of Lemma~\ref{lem:T-min-pathlength}.
Observe that the number of negative trees in~$I_t^1$ is $|\T \setminus \T_t|<|\T|$ because $\T_t \neq \emptyset$,
and the number of negative trees in~$I_t^2$ is $|\T_t|<|\T|$ because $T \notin \T_t$.
Hence, due to our inductive hypothesis we know that these recursive calls for $I_t^1$ and for~$I_t^2$ return minimum-weight solutions for them. 
By Lemma~\ref{lem:subinstances-correct}, this implies that the solution $(S_1,S_2)$ for~$(G,w,s,t)$ 
obtained on Line~\ref{line:DISP-Tt-recurse} has weight at most~$w(P_1)+w(P_2)$.
 
If $\T_s=\T_t=\emptyset$, then Algorithm~\ref{alg:DISP} guesses the first and the last vertices on the solution paths contained in~$T$;
consider the iteration when the guessed vertices, $a_i$ and~$b_i$, are the first and last vertices of~$P_i$ contained in~$T$, for both $i \in [2]$.
Note that since $(\emptyset,\T \setminus \{T\},\emptyset)$ is a valid $T$-partition, we know that the paths~$P[s,a_i]$ and~$P[b_i,t]$ for $i \in[2]$ 
contain no edges of~$E^-$.
Hence, these four paths directly yield a flow of value~4 in the network~$N_{(a_1,b_1,a_2,b_2)}$. 
Also, the paths $P_1[a_1,b_1]$ and $P_2[a_2,b_2]$ are two permissively disjoint  $(\{a_1,a_2\},\{b_1,b_2\})$-paths in~$G$.
Thus Algorithm~\ref{alg:DISP} satisfies the conditions on Lines~\ref{line:DISP-testflow} and~\ref{line:DISP-testpaths}, 
and proceeds with computing a minimum-cost flow of value~4 in~$N_{(a_1,b_1,a_2,b_2)}$, 
as well as two permissively disjoint $(\{a_1,a_2\},\{b_1,b_2\})$-paths of minimum weight in~$G$ using Algorithm~\ref{alg:PermDisj}.
Note that by Lemma~\ref{lem:type2bsol} the returned flow~$f$ will have cost at most~$\sum_{i=1}^2 w(P_i[s,a_i])+w(P_i[b_i,t])$, 
and by Corollary~\ref{cor:perm-disjoint-paths}, Algorithm~\ref{alg:PermDisj} returns two 
permissively disjoint $(\{a_1,a_2\},\{b_1,b_2\})$-paths of weight at most~$\sum_{i=1}^2 w(P_i[a_i,b_i])$. 
Finally, the algorithm applies Lemma~\ref{lem:type2bsol} to construct a solution of weight at most 
\[\sum_{i=1}^2 w(P_i[s,a_i])+w(P_i[b_i,t])+w(P_i[a_i,b_i]) =\sum_{i=1}^2 w(P_i),\] as required.

\smallskip

It remains to analyze the running time of Algorithm~\ref{alg:DISP}. Let $n=|V|$ 
and recall that the number of trees spanned by all negative edges in~$G$ is~$c=|\T|$. 
We denote by $T(n,c)$ the running time of Algorithm~\ref{alg:DISP} as a function of these two integers.

There are $|E^-|+1\leq n$ possibilities for choosing~$E^-$ on Line~\ref{line:DISP-part1-start}, 
and at most $\binom{n}{c+1}$ possibilities for choosing~$Z$ on Line~\ref{line:DISP-choose-Z}, because 
removing an edge of~$E^-$ from~$G$ may increase parameter~$c$, that is,
the number of connected components of the subgraph spanned by all negative-weight edges. 
Creating the network~$N_Z$ and computing a minimum-cost flow of value~$2$ in it takes $O(n^3)$ time by using the 
Successive Shortest Path algorithm~\cite{Jewell62,Iri60,BG60}, 
and from that it takes linear time to construct a solution for~$(G,w,s,t)$.
Thus, the computations performed on Lines~\ref{line:DISP-part1-start}--\ref{line:DISP-part1-end} take $O(n^{c+5})$ time.

There are $c$ possibilities to choose~$T$ on Line~\ref{line:DISP-choose-tree}, and 
there are $3^{c-1}$ possibilities to choose the partition~$(\T_s,\T_0,\T_t)$ of~$\T  \setminus \{T\}$ 
on Line~\ref{line:DISP-choose-part}.
The iteration on Line~\ref{line:DISP-choose-a1a2} yields a factor of less than~$n^2$. 
Both recursive calls take time at most $T(n,c-1)$, since
the the number of negative trees decreases in both $I_s^1$ and $I_s^2$ (as we have already shown in this proof).
Applying Lemma~\ref{lem:subinstances-correct} (which in turn relies on Lemma~\ref{lem:uncrossing}) takes linear time.
Hence, the computations on Lines~\ref{line:DISP-choose-a1a2}--\ref{line:DISP-Ts-output} take time $O(n^2) T(n,c-1)$.
Similarly, Lines~\ref{line:DISP-choose-b1b2}--\ref{line:DISP-Tt-output} also take time $O(n^2) T(n,c-1)$.

There are $O(n^4)$ possibilities to choose the vertices~$a_1,b_1,a_2$, and $b_2$ on Line~\ref{line:DISP-choose-4vert}. 
The minimum-cost flow computation on Line~\ref{line:DISP-find-flow4} takes $O(n^3)$ time, but the bottleneck here is the
running time of Algorithm~\ref{alg:PermDisj}, called on Line~\ref{line:DISP-use-PD}.
Algorithm~\ref{alg:PermDisj} performs a call to Algorithm~\ref{alg:PartSol} 
for at most~$2^{c-1} \binom{n}{2}$ triples~$(u,v,\tau)$.
Algorithm~\ref{alg:PartSol} in turn iterates over all possible values for vertices~$u'$ and~$v'$ 
as well as over all values of~$\tau' \subseteq \tau$, and for each triple~$(u',v',\tau')$ 
computes a shortest $(u',v)$-path in the auxiliary graph $G \langle T[v',u],u',v,\tau \setminus \tau' \rangle$.\footnote{
In fact, there are two more shortest-path computations performed by the algorithm 
during Lines~\ref{line:PS-case1-start}--\ref{line:PS-comp1} of Algorithm~\ref{alg:PartSol}, but the running time of Algorithm~\ref{alg:PartSol} is dominated by the computations on \hbox{Lines~\ref{line:PS-case-s2tart}--\ref{line:PS-case2-end}}.} 
Note that $G \langle T[v',u],u',v \rangle$ implicitly depends also on the path~$X$ (via the definition of the trees $T_i$, $i \in [r]$). Since there are $O(n^2)$ possibilities for choosing~$X$, then $n^4$ possibilities to choose $u,v,u'$ and~$v'$, and then $2^{c-1}$ possibilities to choose the subset of negative trees (other than~$T$) contained in the given auxiliary graph,
we obtain that the number of shortest-path computations we need to perform during 
Lines~\ref{line:DISP-choose-4vert}--\ref{line:DISP-nonseparable-output} of Algorithm~\ref{alg:DISP} is $O(2^{c-1} n^6)$ in total.

Let us now consider an implementation of Lines~\ref{line:DISP-choose-4vert}--\ref{line:DISP-nonseparable-output} of Algorithm~\ref{alg:DISP} where we compute all necessary shortest paths in all possible auxiliary graphs as a preprocessing step: since we can compute a shortest path between two given vertices in an undirected conservative graph in $O(n^3)$ time (see e.g.,~\cite[Section 29.2]{schrijver-book}), this takes $O(2^{c-1} n^9)$ time in total. Assuming that we already have the results of this preprocessing step, a single call to Algorithm~\ref{alg:PartSol} (assuming that all partial solutions called for on line~\ref{line:PS-callF} are available) takes $O(2^{c-1} n^3)$ time, because there are $O(2^{c-1} n^2)$ possibilities for choosing $u',v'$ and~$\tau'$ (recall that $T \notin \tau'$), and an iteration for fixed values of~$u',v'$ and~$\tau'$ takes linear time. Therefore, running Algorithm~\ref{alg:PermDisj} for fixed values of~$a_1,b_1,a_2$ and $b_2$ takes $2^{c-1} n^2 O( 2^{c-1} n^3)=O(2^{2c-2}n^5)$ time (as it calls Algorithm~\ref{alg:PartSol} at most $2^{c-1} n^2$ times), which yields a running time of $O(2^{2c-2} n^9)$ in total after the preprocessing phase. Hence, the overall running time of Lines~\ref{line:DISP-choose-4vert}--\ref{line:DISP-nonseparable-output} (including all necessary preprocessing steps)  is $O(2^{2c-2} n^9)$.

This yields the following recursive formula for the total running time of Algorithm~\ref{alg:DISP}:
\[T(n,c)=O(n^{c+5})+ c \dot 3^{c-1} \big( O(n^2) T(n,c-1) + O(2^{2c-2} n^9) \big)
\]
from which we can conclude $T(n,c)=O(n^{2c+7})$.
In particular, for $c=1$ we obtain a running time of $O(n^9)$.
\hfill$\qedsymbol$

\section{Conclusion}
\label{sec:conclusion}

We have presented a polynomial-time algorithm for solving the \DISP{} problem 
on undirected graphs~$G$ with conservative edge weights, assuming that the number of connected components in the subgraph~$G[E^-]$ 
spanned by all negative-weight edges is a fixed constant~$c$.
The running time of our algorithm is $O(n^{2c+9})$ on an $n$-vertex graph.
Is it possible to give a substantially faster algorithm for this problem? 
In particular, is it possible to give a fixed-parameter tractable algorithm for \DISP{} on undirected conservative graphs 
when parameterized by~$c$?

More generally, is it possible to find in polynomial time $k$ openly disjoint $(s,t)$-paths with minimum total weight 
for some fixed~$k \geq 3$ in undirected conservative graphs with constant values of~$c$?

\subsection*{Funding}
Ildik\'o Schlotter acknowledges the support of the Hungarian Academy of Sciences under its Momentum Programme (LP2021-2) and its J\'anos Bolyai Research Scholarship. 

\bibliographystyle{plainurl}
\bibliography{2paths}

\end{document}